\documentclass[preprint,12pt]{elsarticle}

\usepackage{graphicx}

\renewcommand{\bibname}{References}%

\let\ElseVierBibliography\bibliography%
\renewcommand{\bibliography}[1]{%
\section*{\bibname}%
\ElseVierBibliography{#1}%
}%

\usepackage{amsfonts,amstext,amsmath,amssymb,amsthm}
\usepackage{dsfont}
\usepackage{txfonts}

\usepackage{float}

\usepackage{chngpage}

\usepackage{subfig}

\usepackage{tabularx, booktabs}
\newcolumntype{Y}{>{\centering\arraybackslash}X}

\usepackage{empheq}
\usepackage{paralist}

\usepackage[T1]{fontenc}

\usepackage{setspace}
\usepackage{multirow}
\usepackage{bigstrut}
\usepackage{booktabs}
\usepackage{rotating}

\usepackage{hyperxmp}

\usepackage[%
    bookmarks=true,         
    unicode=false,          
    pdftoolbar=true,        
    pdfmenubar=true,        
    pdffitwindow=false,     
    pdfstartview={FitH},    
    pdftitle={Comparative performance analysis of the Cumulative Sum chart and the Shiryaev--Roberts procedure for detecting changes in autocorrelated data},
    pdfauthor={Aleksey S. Polunchenko and Vasanthan Raghavan},     
    pdfsubject={Comparative performance analysis of the Cumulative Sum chart and the Shiryaev--Roberts procedure for detecting changes in autocorrelated data},
    pdfcreator={Aleksey S. Polunchenko},   
    pdfproducer={MikTeX}, 
    pdfkeywords={Sequential Analysis,Quickest change-point detection,Shiryaev--Roberts procedure,Generalized Shiryaev--Roberts procedure,Sequential analysis}, 
    pdfnewwindow=true,      
    colorlinks=false,       
    linkcolor=red,          
    citecolor=green,        
    filecolor=magenta,      
    urlcolor=cyan           
]{hyperref}

\usepackage[T1]{fontenc}

\biboptions{numbers,compress}

\usepackage{multirow}

\newcommand\ignore[1]{}

\newcommand{\hsp}{\hspace{0.1in} }
\newcommand{\hspp}{\hspace{0.05in} }
\newcommand{\hsppp}{\hspace{0.02in} }

\renewcommand{\Pr}{\mathbb{P}} 
\DeclareMathOperator{\EV}{\mathbb{E}} 

\DeclareMathOperator{\LLR}{\mathcal{L}}
\DeclareMathOperator{\LR}{\Lambda}

\DeclareMathOperator{\KL}{\mathds{I}}
\DeclareMathOperator{\ADD}{ADD}
\DeclareMathOperator{\Var}{Var}

\DeclareMathOperator{\ARL}{ARL}
\DeclareMathOperator{\SADD}{SADD}
\DeclareMathOperator{\ESADD}{ESADD}
\DeclareMathOperator{\STADD}{STADD}

\DeclareMathOperator*{\arginf}{arg\,inf}

\DeclareMathOperator*{\esssup}{ess\,sup}

\newcommand{\T}{T}
\newcommand{\TCS}{\tau_{\sf cs}}
\newcommand{\TSR}{\tau_{\sf sr}}

\newcommand{\abs}[1]{\left\vert#1\right\vert}

\graphicspath{{./gfx/}, {./gfx2/}, {./new_figs/}}

\theoremstyle{plain} 

\newtheorem{lemma}{Lemma}
\newtheorem{proposition}{Proposition}

\theoremstyle{remark}
\newtheorem*{remark}{Remark}

\journal{Applied Stochastic Models in Business and Industry}

\usepackage[norefs,nocites]{refcheck}
\usepackage{silence}
\begin{document}

\begin{frontmatter}



\title{\large{\bf\uppercase{Comparative performance analysis of the Cumulative Sum chart and the Shiryaev--Roberts procedure for detecting changes in autocorrelated data}}}


\author[bu]{Aleksey\ S.\ Polunchenko\corref{cor-author}}
\ead{aleksey@binghamton.edu}
\ead[url]{http://people.math.binghamton.edu/aleksey}
\author[qualcomm]{Vasanthan Raghavan}
\ead{vasanthan\_raghavan@ieee.org}
\cortext[cor-author]{Address correspondence to A.\ S.\ Polunchenko, Department of
Mathematical Sciences, State University of New York (SUNY) at Binghamton,
Binghamton, NY 13902--6000, USA; Tel: +1 (607) 777-6906; Fax: +1 (607) 777-2450; Email:~\href{mailto:aleksey@binghamton.edu}{aleksey@binghamton.edu}}
\address[bu]{Department of Mathematical Sciences, State University of New York
(SUNY) at Binghamton\\Binghamtom, NY 13902--6000, USA}
\address[qualcomm]{Qualcomm, Inc.,\\ Bridgewater, NJ 08807, USA}

\begin{abstract}
We consider the problem of quickest change-point detection where the observations form a first-order autoregressive (AR) process driven by temporally independent standard Gaussian noise. Subject to possible change are both the drift of the AR(1) process ($\mu$) as well as its correlation
coefficient ($\lambda$), both known. The change is abrupt and persistent, and is of known magnitude, with $\abs{\lambda}<1$ throughout. For this scenario, we carry out a comparative performance analysis of the popular Cumulative Sum (CUSUM) chart and its less well-known but worthy competitor---the Shiryaev--Roberts (SR) procedure. Specifically, the performance is measured through Pollak's Supremum (conditional) Average Delay to Detection (SADD) constrained to a pre-specified level of the Average Run Length (ARL) to false alarm. Particular attention is drawn to the sensitivity of each procedure's SADD and ARL with respect to the value of $\lambda$ before and after the change. The performance is studied through the solution of the respective integral renewal equations obtained via Monte Carlo simulations. The simulations are designed to estimate the sought performance metrics in an unbiased and asymptotically strongly consistent manner, and to within a prescribed proportional closeness (also asymptotically). Our extensive numerical studies suggest that both the CUSUM chart and the SR procedure are asymptotically {\em second-order} optimal, even though the CUSUM chart is found to be slightly better than the SR procedure, irrespective of the model parameters. Moreover, the existence of a {\em worst-case} post-change correlation parameter corresponding to the poorest detectability of the change for a given ARL to false alarm is established as well. To the best of our knowledge, this is the first time the performance of the SR procedure is studied for autocorrelated data.
\end{abstract}


\end{frontmatter}
\section{Introduction}
\label{sec:Intro}

Sequential (quickest) change-point detection is concerned with the design and analysis of reliable statistical machinery for quick detection of potential changes in the attributes of a random process. Specifically, the process is assumed to be continuously monitored through observations made sequentially, and should their behavior suggest that the process may have statistically changed, the aim is to conclude so within the fewest observations possible, subject to a tolerance level on the risk of false alarm.
A sequential change-point detection procedure is a 
stopping-time adapted to the observations, and provides a rule to stop and
declare that a change may be in effect. This problem finds applications in many
branches of science and engineering: quality control, biostatistics, economics,
seismology, 
and communication systems; see,
e.g.,~\cite{Basseville+Nikiforov:Book93,Kenett+Zacks:Book1998,Montgomery:Book2012}.

In the simplest change-point detection problem, the observations are independent
and identically distributed (i.i.d.) with known pre- and post-change
distributions. In this setting, the problem is well-understood and has been
solved to optimize different objectives. For a recent survey,
see~\cite{Tartakovsky+Moustakides:SA2010,Polunchenko+Tartakovsky:MCAP2012}
and references therein. In general, two solutions stand out:
Page's Cumulative Sum (CUSUM) chart~\cite{Page:B54} and the
Shiryaev--Roberts (SR) procedure due to the independent work
of Shiryaev~\cite{Shiryaev:SMD61,Shiryaev:TPA63}
and Roberts~\cite{Roberts:T66}. While the two procedures are statistically different,
both are optimal under different sets of criteria. In particular, Moustakides~\cite{Moustakides:AS86} and Ritov~\cite{Ritov:AS90} have shown that the CUSUM chart is exactly minimax-optimal in the sense of minimizing the detection delay under the most unfavorable
set of observations and change-point. This type of minimax optimality was proposed by Lorden~\cite{Lorden:AMS71}.
On the other hand, Pollak and Tartakovsky~\cite{Pollak+Tartakovsky:SS09} showed that the SR procedure
is optimal in the stationary setting, a scenario more suitable for detecting changes
that occur at a distant time-horizon. Given that the Lorden criterion is
often conservative, Pollak and Siegmund~\cite{Pollak+Siegmund:AS1975} introduced a more reasonable metric of
detection delay under the most unfavorable change-point, but averaged over the
observations; see also~\cite{Pollak:AS85}. While the structure of the exactly optimal solution is unknown for the
Pollak criterion, both the CUSUM chart and SR procedure are asymptotically optimal
as the false alarm risk vanishes; see,
e.g.,~\cite{Tartakovsky+Veeravalli:TPA05}. Thus, there has been a good
justification for comparing the two procedures with each other.

This comparative analysis has been done extensively in the i.i.d.\ case and we
now present a brief sampling of this literature. The study in~\cite{Pollak+Siegmund:B85}
offered a comprehensive asymptotic analysis (in the low false alarm regime) for
the problem of detecting a change in the drift of Brownian motion. The
conclusion in~\cite{Pollak+Siegmund:B85} was that the CUSUM chart is better for
changes that occur in the beginning, whereas the SR procedure out-performs the CUSUM
chart for change at infinity. Dragalin~\cite{Dragalin:S94} developed an accurate
numerical technique to capture the performance of the CUSUM chart in detecting a change
in the mean of a Gaussian sequence. More recently, Moustakides et. al~\cite{Moustakides+etal:CommStat09,Moustakides+etal:SS11} have developed an exact analytical characterization of
the two procedures under either criteria through a set of {\em integral-equations}.
These equations are in turn solved numerically using simple computational techniques. Confirming the findings
of Pollak and Siegmund~\cite{Pollak+Siegmund:B85} and Pollak and Tartakovsky~\cite{Pollak+Tartakovsky:SS09}, these computations
show that the CUSUM chart is superior to the SR procedure under the Pollak criterion,
whereas the SR procedure is better in the stationary sense.

Despite the strong theoretical focus on the i.i.d.\ problem, observations are
often serially correlated in industrial practice with a first-order autoregressive
(AR(1)) process model being a good fit in many scenarios. A change could occur
either due to a shift in the mean level or in the correlation coefficient (of the
AR process) or both attributes simultaneously; see different examples
in~\cite{Basseville+Nikiforov:Book93,Kenett+Zacks:Book1998,Goldsmith_Whitfield_1961,Bagshaw_Johnson_1974,Ermer_Chow_Wu_1979,Steiner_et_al_2000}, etc. Many works in the literature have shown that when the traditional Shewhart,
Exponentially Weighted Moving Average (EWMA) and CUSUM charts designed for i.i.d.\
observations are used with AR processes, they result in seriously misleading
conclusions; e.g., see typical case-studies
in~\cite{Berthouex_et_al_1978,Wardell_et_al_1992,Alwan_Roberts_1995}.

Motivated by these observations, modified versions of the traditional control
charts accommodating serial correlations have been proposed;
see~\cite{Vasilopoulos_et_al_1978,Berthouex_et_al_1978,Ermer_Chow_Wu_1979,Alwan_Roberts_1988,Montgomery_Mastrangelo_1991,Lu_Reynolds_2001,Apley_Tsung_2002} for some extensions along these lines. Most of these procedures decompose the
correlated data into common cause effects and residuals (innovations) that are
mutually independent. Under certain assumptions, the statistical properties of
the residuals can be characterized. Specifically, the limiting distribution of
the likelihood ratio function under either a mean shift or a correlation change
in the AR process has been studied and various approximations to the average
run length (ARL) to false alarm 
are obtained
in~\cite{Bagshaw_Johnson_1975,Picard_1985,Nikiforov_1986,Alwan_Roberts_1988,Harris_Ross_1991,Woodall_1992,yashchin_1993,Wardell_et_al_1994,runger_1995,Apley_Shi_1999}. In particular, Davis, Huang, and Yao~\cite{Davis_Huang_Yao_1995} has shown that the likelihood ratio
statistic converges weakly to the extreme value distribution and use this property
to characterize the performance of the CUSUM chart, whereas Timmer et.~al~\cite{timmer_1998}
estimate the ARL of the CUSUM chart using a Markov chain representation. The
studies in~\cite{Berkes_2009} and~\cite{Gombay_2009} use the efficient score vector
representation of the likelihood ratio statistic of weighted CUSUM charts to characterize
their performance.

In spite of this vast literature, optimality properties of the CUSUM chart have
been explored only under certain special settings and only up to first-order. For
example, Moustakides~\cite{Moustakides:IT98} has shown that the CUSUM test statistic reduces to
the i.i.d.\ statistic in the special case where only the mean of the AR process
changes and the CUSUM chart is thus optimal in the Lorden sense. First-order
optimality of the CUSUM chart and the SR procedure under more general observation
models has also been established; see
e.g.,~\cite{Lai:IEEE-IT1998,yakir_et_al_1999} and~\cite{Tartakovsky+Veeravalli:TPA05}.
First-order optimality of the CUSUM chart under certain general change-point models
relevant in practice is explored in~\cite{knoth_2012}. Nevertheless, a comparative
performance of the CUSUM chart with the SR procedure in the model parameter space
has not been explored in the AR setting in full generality. The focus of this
paper is on this task and we provide a comparative study of the two procedures in
the non-asymptotic setting with {\it correlated} observations. We are not aware of
any similar prior work in this area.

This paper is organized as follows. In Section 2, we set the backdrop for this
paper by elaborating on the problem set-up, developing the notation, and
connecting it with prior results in this area. In Section 3, the KL number
between autoregressive processes is studied as a function of the pre- and post-change
model parameters. In Section 4, we derive the integral equations for the
performance metrics of interest and provide a simple numerical solution that
allows for efficient computation of the operating characteristics.
In Section 5,
we present the results of our numerical studies and discuss the findings.
Section 6 summarizes and concludes the paper.

\section{Problem Formulation and Preliminary Background}
\label{sec:problem-formulation+background}

The aim of this section is to formally state the problem, present the CUSUM chart and the Shiryaev--Roberts (SR) procedure, both set up appropriately, and review their (asymptotic)
optimality properties.

Let $\{X_n\}_{n\ge0}$ be an observation sequence formed sequentially from the output of an AR(1) process driven by temporally-independent standard Gaussian noise $\{\varepsilon_n\}_{n \ge 1}$, i.e., $\varepsilon_n \sim \mathcal{N}(0,1)$, $n\ge1$, and 
$\varepsilon_i$ is independent of $\varepsilon_j$ if
$i\neq j$. Let the statistical nature of the observation sequence, $\{X_n\}_{n\ge0}$, be temporally
piece-wise:
\begin{align}\label{eq:def-AR-eqn}
X_n
&=
\begin{cases}
\mu_{\infty}+\lambda_{\infty} X_{n-1}+\varepsilon_n,\;\text{for $1\le n\le\nu$;}\\
\mu_{0}+\lambda_{0} X_{n-1}+\varepsilon_n,\;\text{for $n\ge\nu+1$},
\end{cases}
\end{align}
where $\mu_d\in\mathbb{R}$ and $\lambda_d$ is such that $\abs{\lambda_d}<1$.
$\mu_d$ and $\lambda_d$ are known for both $d = \{0,\infty\}$, $X_0=x_0\in\mathbb{R}$
is a given deterministic value, and $\nu$ is a parameter discussed next. The data model
in~\eqref{eq:def-AR-eqn} says that the observation sequence, $\{X_n\}_{n\ge0}$, as it
is formed in a one-observation-at-a-time manner, undergoes a spontaneous change in its
statistical nature. The quickest change-point detection problem is to as quickly and
reliably as possible establish that the statistical nature has changed. The challenge
is that the time instance $\nu$, which is referred to as the {\em change-point}, is
{\em not} known in advance. A solution to the problem is a detection procedure identified
with a $\{X_n\}_{n\ge0}$-adapted stopping time, $\T$, and a ``good'' procedure is one
whose detection delay cost is the smallest possible within a given tolerable range of
the false alarm risk.

\begin{remark}
A noteworthy feature of the AR(1) model in~\eqref{eq:def-AR-eqn} is that
$X_{0}, X_1, \cdots, X_{\nu}$ (the pre-change observations) and
$X_{\nu+1}, X_{\nu+2}, \cdots$ (the post-change observations) are {\em not} independent
as $X_{\nu+1}$ (the first data point affected by change) is correlated with $X_{\nu}$
(the final data point not yet affected by change). This is different from the general
$\mathrm{AR}(m)$ model considered, e.g., in~\cite{Nikiforov_1986}, where the pre- and
post-change pieces of the observations sequence are assumed independent.
\end{remark}

More specifically, in this work, we will take the {\em minimax approach}, i.e., regard the change-point, $\nu$, as unknown (but not random); for an overview of other approaches, see, e.g.,~\cite{Tartakovsky+Veeravalli:TPA05,Tartakovsky+Moustakides:SA2010,Polunchenko+Tartakovsky:MCAP2012,Polunchenko+etal:JSM2013}. From now on, the notation $\nu=0$ is to be understood as the case where the parameters of $\{X_n\}$ are $\mu_0$ and $\lambda_0$ for all $n\ge1$, i.e., the data, $\{X_n\}_{n\ge1}$, are affected by change {\em ab initio}. Similarly, the notation $\nu=\infty$ is to mean that the parameters of $\{X_n\}$ are $\mu_\infty$ and $\lambda_\infty$ for all $n\ge1$.

Let $\mathcal{H}_k\colon\nu=k$ 
be the hypothesis that the change-point, $\nu$, is at epoch $k$, $0\le k<\infty$.
Let $\mathcal{H}_{\infty}\colon\nu=\infty$ be the hypothesis that $\nu=\infty$, i.e.,
that the process' parameters never change. Further, let $\Pr_k$ (and $\EV_k$) be the
probability measure (and the corresponding
expectation) given a known change-point $\nu=k$, where $0\le k\le\infty$. In particular,
$\Pr_\infty$ ($\EV_\infty$) is the probability measure (corresponding expectation) assuming
that the AR(1) process' parameters are $\mu_\infty$ and $\lambda_\infty$ for all $n\ge1$,
and never change (i.e., $\nu=\infty$). Likewise, $\Pr_0$ ($\EV_0$) is the probability
measure (corresponding expectation) assuming that the AR(1) process' parameters are $\mu_0$ and
$\lambda_0$ for all $n\ge1$ (i.e., $\nu=0$).

Under the minimax approach, the standard method to gauge the false alarm risk is through Lorden's~\cite{Lorden:AMS71} Average Run Length (ARL) to false alarm; it is defined as $\ARL(\T)\triangleq\EV_\infty[\T]$, and captures the average number of observations sampled before a false alarm is sounded. Let
\begin{align*}
\Delta(\gamma)
&\triangleq
\Bigl\{\T\colon\ARL(\T)\ge\gamma\Bigr\},\;\gamma>1,
\end{align*}
denote the class of procedures with the ARL to false alarm of at least $\gamma>1$, a pre-selected tolerance level. For the detection delay cost, we will use the criterion proposed by Pollak~\cite{Pollak:AS85}; see also~\cite[Section~5]{Pollak+Siegmund:AS1975}. It is known as the Supremum (conditional) Average Detection Delay (SADD), and is defined as,
\begin{align}
\SADD(\T)
&\triangleq
\sup_{0\le k<\infty}
\ADD_k(\T)\text{ with }\ADD_k(\T)\triangleq\EV_k[\T-k|\T>k].
\end{align}

The overarching problem of interest in this framework is to find
$\T_{\mathrm{opt}}\in\Delta(\gamma)$ that minimizes $\SADD(\T)$ over all
$\T\in\Delta(\gamma)$ for all $\gamma>1$, or more succinctly, to
\begin{align}\label{eq:def-Pollak-minmax-problem}
\text{find}\,\T_{\mathrm{opt}}
&=
\arginf_{\T\in\Delta(\gamma)}\SADD(\T),
\end{align}
for every $\gamma>1$. This problem is still open. Even in the basic i.i.d.\ case, only a
partial solution has been offered so far~\cite{Polunchenko+Tartakovsky:AS10,Tartakovsky+Polunchenko:IWAP10,Moustakides+etal:SS11,Tartakovsky+etal:TPA2012}. Specifically, as shown in~\cite{Polunchenko+Tartakovsky:AS10,Tartakovsky+Polunchenko:IWAP10}, the so-called generalized Shiryaev--Roberts (GSR) procedure (due to~\cite{Moustakides+etal:SS11}) is exactly SADD-optimal under two specific i.i.d. scenarios. This result was then extended in~\cite{Tartakovsky+etal:TPA2012} where, under a general i.i.d.\ scenario, the GSR procedure was demonstrated to minimize the SADD asymptotically, as $\gamma\to\infty$, to within an
$o(1)$ term; here $o(1) \to 0$, as $\gamma\to\infty$. However, beyond the basic i.i.d. case,
not much progress has been made so far, and only the asymptotic theory has been outlined.
In the general case, it was demonstrated in~\cite{Lai:IEEE-IT1998} that under certain regularity conditions the CUSUM chart and the SR procedure are asymptotically first-order optimal.


In the AR(1) case, the joint cumulative distribution functions (c.d.f.'s) of the sample $\boldsymbol{X}_{1:n}\triangleq(X_1,\ldots,X_n)$, $n\ge1$, under the ${\mathcal H}_{\infty}$
and ${\mathcal H}_k$ hypotheses are given by
\begin{align*}
\Pr(\boldsymbol{X}_{1:n}|\mathcal{H}_\infty)
&=
\prod_{j=1}^n\Pr_\infty(X_j|X_{j-1})
\;\text{and }
\Pr(\boldsymbol{X}_{1:n}|\mathcal{H}_k)=
\prod_{j=1}^k\Pr_\infty(X_j|X_{j-1})\prod_{j=k+1}^n\Pr_0(X_j|X_{j-1}),
\end{align*}
where here (and throughout the sequel), it is to be understood that
$\prod_{i}^{j}\equiv1$ whenever $i>j$.
Consequently, for the respective likelihood ratio (LR), we obtain
\begin{align*}
\LR_{1:n,\nu=k}
&\triangleq
\dfrac{d\Pr(\boldsymbol{X}_{1:n}|\mathcal{H}_k)}{d\Pr(\boldsymbol{X}_{1:n}|\mathcal{H}_\infty)}
= \prod_{j=k+1}^n\LR_j(X_j,X_{j-1}),
\end{align*}
where
\begin{align}\label{eq:instant-LR-formula}
\begin{split}
\LR_n(X_n,X_{n-1})
&\triangleq
\exp\Biggl\{\left(X_n-\dfrac{1}{2}\bigl[X_{n-1}(\lambda_0+\lambda_\infty)+(\mu_0+\mu_\infty)\bigr]\right)\\
&\qquad\qquad\qquad\qquad\Bigl[X_{n-1}(\lambda_0-\lambda_\infty)+(\mu_0-\mu_\infty)\Bigr]\Biggl\},\;n\ge1,
\end{split}
\end{align}
is the ``instantaneous'' LR for the $n$-th data point, $X_n$, {\em conditioned} on the $(n-1)$-th
data point, $X_{n-1}$; for notational brevity, we will also refer to $\LR_n(X_n,X_{n-1})$ as simply $\LR_n$, unless it is necessary to stress that it is a function of $X_{n}$ and $X_{n-1}$.

\begin{remark}
As a special case of the AR(1) model in~\eqref{eq:def-AR-eqn}, suppose, for the moment, that the change is {\em only} in the drift and the correlation coefficient is {\em not} affected, i.e., let $\lambda_{\infty}=\lambda_0\,(\triangleq\lambda)$, but $\mu_\infty\neq\mu_0$; then the LR formula given above reduces to
\begin{align*}
\LR_{n}
&=
\exp\left\{ (\mu_0-\mu_{\infty})\left[\tilde{\varepsilon}_{n}-\dfrac{\mu_0+\mu_{\infty}}{2}\right]\right\},\;\;n\ge1,
\end{align*}
where $\tilde{\varepsilon}_{n}\triangleq X_{n}-\lambda X_{n-1}$, $n\ge1$. This is easily recognized as the LR formula for the well-studied i.i.d.\ data model, i.e., when a sequence of independent unit-variance Gaussian random variables undergoes an abrupt and persistent shift in the mean from $\mu_\infty$ to $\mu_0$; see, e.g.,~\cite{Tartakovsky+Ivanova:PIT92,Kenett+Zacks:Book1998,Mahmoud_2008,Tartakovsky+etal:IWSM2009,Moustakides+etal:CommStat09,Moustakides+etal:SS11,Polunchenko+etal:SA2014,Polunchenko+etal:ASMBI2014} among many other references. Hence, when $\lambda_{\infty}=\lambda_0$, the AR(1) model in~\eqref{eq:def-AR-eqn} is equivalent to the basic i.i.d.\ model, and the presence of correlation in the observations is completely irrelevant; cf.~\cite[Section~IIIB,~p.~1967]{Moustakides:IT98}. We shall therefore always require that at least the correlation coefficient is affected by the change, i.e., $\lambda_\infty\neq\lambda_0$.
\end{remark}

We now switch attention to the objective of this work, which is to study the SR procedure
for detecting change in the AR process parameters and benchmarking its performance with that
of the CUSUM chart. The SR procedure corresponding to a threshold $A$ is defined as
\begin{align}\label{eq:def-Tsr}
\TSR(A)
&\triangleq
\inf\bigl\{n\ge1\colon R_n\ge A\bigr\},
\;\;
\text{such that}
\;\;
\inf\{\varnothing\}=\infty,
\end{align}
where the SR decision statistic, $\{R_n\}_{n\ge0}$, is defined as
\begin{align}\label{eq:def-Rn}
R_n
&\triangleq
\sum_{k=1}^n\LR_{1:n,\nu=k}
=
\sum_{k=0}^{n-1} \prod_{i=k}^n\LR_i,\,\,n\ge1,\;\;\text{with}\;\; R_0=0,
\end{align}
where $\{\LR_n\}_{n\ge1}$ is as in~\eqref{eq:instant-LR-formula}, and we note the recursion
\begin{align}\label{eq:statistic-SR-def}
R_{n}
&=
(1+R_{n-1})\LR_{n},\;\;n\ge1,\;\;\text{with}\;\; R_0=0.
\end{align}

We remark that, as can be seen from~\eqref{eq:statistic-SR-def}, the SR detection statistic, $\{R_n\}_{n\ge0}$, starts off at zero, i.e., $R_0=0$. This is the original definition of Shiryaev~\cite{Shiryaev:SMD61,Shiryaev:TPA63} and Roberts~\cite{Roberts:T66}. However, if the detection statistic is given a specifically designed {\em headstart}, i.e., if $R_0=r\ge0$, then the performance of the procedure may improve substantially.
The headstarted version of the SR procedure (the generalized SR procedure) is
proposed in~\cite{Moustakides+etal:SS11} and studied in~\cite{Tartakovsky+etal:TPA2012}.
The basic proposal of giving headstart to a procedure was first proposed
in~\cite{Lucas+Saccucci:T90} in the context of the CUSUM chart.

Contrary to the quasi-Bayesian background of the SR procedure, the CUSUM chart is based on the maximum likelihood argument, and its stopping time is defined as
\begin{align}\label{eq:def-Tcs}
\TCS(A)
&\triangleq
\inf\bigl\{n\ge1\colon V_n\ge A\bigr\},
\;\;
\text{such that}
\;\;
\inf\{\varnothing\}=\infty,
\end{align}
where the decision statistic, $\{V_n\}_{n\ge0}$, is given by
\begin{align}
V_n
&\triangleq
\max_{0 \le k\le n-1}\LR_{1:n,\nu=k},\;\;n\ge1,\;\;\text{with}\;\; V_0=0,
\end{align}
and we note the recursion
\begin{align}\label{eq:statistic-CS-def}
V_{n}
&=
\max\{1,V_{n-1}\}\LR_{n},\;\;n\ge1,\;\;\text{with}\;\; V_0=0.
\end{align}

As mentioned in the Introduction, a majority of the change-point detection theory developed
to date is restricted to the i.i.d.\ model, and is largely of asymptotic character. The
cornerstone of the asymptotic theory for the i.i.d.\ model is the condition
\begin{align}\label{eq:iid-asymptotic-opt-condition}
\dfrac{1}{n-k}\log(\LR_{1:n,\nu=k})
&=
\dfrac{1}{n-k}\sum_{j=k+1}^{n}\log(\LR_{j})
\xrightarrow[n\to\infty]{p}D(\Pr_0\parallel\Pr_\infty)\triangleq\KL,
\end{align}
to be valid under the probability measure $\Pr_k$ for all $0\le k<\infty$; cf.~\cite{Lai:JRSS95,Lai:IEEE-IT1998,Tartakovsky:USC-CAMS-TR1998}. The quantity $\KL\triangleq D(\Pr_0\parallel\Pr_\infty)$ is the Kullback--Leibler divergence or information number (see~\cite{Kullback+Leibler:AMS1951}) defined as
\begin{align}
\KL
&\triangleq
D(\Pr_0\parallel\Pr_\infty)
\triangleq
\lim_{n\to\infty}\dfrac{1}{n}
\int\log\left(\dfrac{d\Pr_0(\boldsymbol{X}_{0:n})}{d\Pr_{\infty}(\boldsymbol{X}_{0:n})}\right)\,d\Pr_0(\boldsymbol{X}_{0:n}),
\end{align}
and it can be interpreted as the {\em directional} distance from the probability measure $\Pr_0$ to the probability measure $\Pr_\infty$. If the i.i.d.\ scenario is such that
$\KL$ is finite, then the condition in~(\ref{eq:iid-asymptotic-opt-condition}) is
automatically (over-)fulfilled by the Strong Law of Large Numbers. Thus, from~\cite{Pollak:AS85,Pollak:AS87} and an argument given in~\cite{Tar1}, it can be deduced
that under the i.i.d.\ assumption both the CUSUM chart and the SR procedure minimize the SADD to within an additive term of order $O(1)$ asymptotically, as $\gamma\to\infty$. That is, $\SADD(\TCS)-\inf_{\T\in\Delta(\gamma)}\SADD(\T)=O(1)$, as $\ARL(\TCS)=\gamma\to\infty$, and $\SADD(\TSR)-\inf_{\T\in\Delta(\gamma)}\SADD(\T)=O(1)$, as $\ARL(\TSR)=\gamma\to\infty$. This is
known as asymptotic second-order optimality.

However, except in the i.i.d.\ case, condition~\eqref{eq:iid-asymptotic-opt-condition} is too
weak to even guarantee that the moment sequence of the stopping time of interest is bounded
from above, let alone to ensure any asymptotic optimality thereof. Hence, unless condition~\eqref{eq:iid-asymptotic-opt-condition} is strengthened, it is not feasible to extend
the asymptotic theory for the i.i.d.\ model to the general non-i.i.d.\ case. This strengthened condition has been obtained in~\cite{Lai:JRSS95,Lai:IEEE-IT1998,Tartakovsky:USC-CAMS-TR1998,Tartakovsky:WCBSMS2000}  and the condition we need is
\begin{align*}
\dfrac{1}{n-k}\log( \LR_{1:n,\nu=k})
&=
\dfrac{1}{n-k}\sum_{j=k+1}^{n}\log(\LR_{j})
\xrightarrow[n\to\infty]{a.s.}\KL
\end{align*}
under $\Pr_k$ for every $k$, $0\le k<\infty$, with the constraint on the rate of convergence:
\begin{align}
\sum_{n=k+1}^\infty
\Pr_k\big( \vert \log(\LR_{1:n,\nu=k})  - (n-k)\KL \vert>(n-k)\epsilon\big)<\infty,\text{for every}
{\hspace{0.03in}} \epsilon>0,
\end{align}
and every $0\le k<\infty$. Together, these two conditions are known as complete convergence~\cite{Hsu+Robbins:NAS-USA1947}.

The complete convergence condition is not restrictive, and is generally met in practice; in particular, the condition is true for correlated Markov processes, such as the AR(1) model in~\eqref{eq:def-AR-eqn}. In fact, for the AR(1) model in~\eqref{eq:def-AR-eqn}, the complete convergence condition has already been verified in~\cite[Example~1,~p.~2455]{Dragalin+etal:IEEE-IT1999}, although in a slightly different context. Hence, it is safe to deduce from~\cite{Lai:JRSS95,Lai:IEEE-IT1998,Tartakovsky:USC-CAMS-TR1998,Tartakovsky:WCBSMS2000}, that both the CUSUM chart and the SR procedure are asymptotically first-order optimal as $\ARL(\T)=\gamma\to\infty$.
That is,
\begin{align*}
\SADD(\TCS)\sim\SADD(\TSR)\sim\inf_{\T\in\Delta(\gamma)}\SADD(\T)
&\ge
\dfrac{\log\gamma}{\KL}[1+o(1)],
\end{align*}
where $o(1)\to0$, as $\gamma\to\infty$.

To conclude this section, we point out that the KL number, $\KL$, is the key in
understanding what one can expect (performance-wise) from a detection procedure, be it the CUSUM chart or the SR procedure. In particular, the higher the KL number, $\KL$, the lower the SADD, i.e., the better the performance. It is therefore worthwhile to analyze the effect that each of the four parameters of the AR(1) model in~\eqref{eq:def-AR-eqn} has on the KL number. This analysis is undertaken in the next section, and, in particular, it is shown that the nature of the dependence of the KL number on each of the four parameters may be counterintuitive.

\section{Analysis of the Kullback--Leibler Information Number}
\label{sec:KL}

The KL number captures the discrimination between the post- and pre-change hypotheses, and is thus a measure of the detectability of the change. This section's aim is take a careful look at the KL
number of the AR(1) model in~\eqref{eq:def-AR-eqn}.

\begin{proposition}
The KL number, $\KL$, for the AR(1) model in~\eqref{eq:def-AR-eqn} is given by the formula
\begin{align}\label{eq:AR-data-model-KL-formula}
\begin{split}
\KL
&\triangleq
\KL(\mu_\infty,\mu_0,\lambda_\infty,\lambda_0)\\
&=
\dfrac{1}{2} \cdot \dfrac{(\lambda_0-\lambda_{\infty})^2}{1-\lambda_0^2}
+\frac{(1-\lambda_\infty)^2}{2} \cdot
\left[\dfrac{\mu_0}{1-\lambda_0}-\dfrac{\mu_{\infty}}{1-\lambda_\infty}\right]^2.
\end{split}
\end{align}
\end{proposition}
\begin{proof}
The desired formula can be derived directly from the KL number's definition, which is
\begin{align*}
\KL
&\triangleq
\lim_{n\rightarrow\infty}\dfrac{1}{n}
\int\log\left(\dfrac{d\Pr_0(\boldsymbol{X}_{0:n})}{d\Pr_{\infty}(\boldsymbol{X}_{0:n})}\right)\,d\Pr_0(\boldsymbol{X}_{0:n}).
\end{align*}

Now, since the Radon--Nikodym derivative ($d\Pr_0/d\Pr_\infty$) under the log in the
integral (in the right-hand side above) has already been computed
in~\eqref{eq:instant-LR-formula}, we obtain
\begin{align*}
\begin{split}
\KL
&=
\lim_{n\to\infty}\dfrac{1}{n}\int\sum_{i=1}^n\Biggl\{X_i(\mu_0-\mu_{\infty})
-X_{i-1}(\lambda_0\mu_0-\lambda_{\infty}\mu_{\infty}) \\
&\qquad\qquad \qquad\qquad
+X_i X_{i-1}(\lambda_0-\lambda_{\infty})-X_{i-1}^2\, \left( \dfrac{\lambda_0^2-\lambda_{\infty}^2}{2} \right) -\left( \dfrac{\mu_0^2-\mu_{\infty}^2}{2}\right) \Biggr\}d\Pr_0(\boldsymbol{X}_{0:n})\\
&=
(\mu_0-\mu_{\infty})\lim_{n\to\infty}\dfrac{1}{n}\sum_{i=1}^n\EV_{0}[X_i]
-(\lambda_0\mu_0-\lambda_{\infty}\mu_{\infty})\lim_{n\to\infty}\dfrac{1}{n}\sum_{i=1}^n\EV_{0}[X_{i-1}] \\
&\qquad\qquad \qquad\qquad
+ (\lambda_0-\lambda_{\infty})\lim_{n\to\infty}\dfrac{1}{n}\sum_{i=1}^n\EV_0[ X_i X_{i-1}]
\\ &\qquad\qquad\qquad\qquad\qquad\qquad\qquad
- \left( \dfrac{\lambda_0^2-\lambda_{\infty}^2}{2} \right)
\lim_{n\to\infty}\dfrac{1}{n}\sum_{i=1}^n\EV_0[ X_{i-1}^2]-
\left( \dfrac{\mu_0^2-\mu_{\infty}^2}{2} \right)\\
& =
[\mu_0(1-\lambda_0)-\mu_{\infty}(1-\lambda_\infty)]\lim_{n\to\infty}\dfrac{1}{n}\sum_{i=1}^n\EV_{0}[X_i]
+ (\lambda_0-\lambda_{\infty})\lim_{n\to\infty}\dfrac{1}{n}\sum_{i=1}^n\EV_0[ X_i X_{i-1}]
\\ &\qquad\qquad \qquad\qquad\qquad\qquad\qquad
- \left( \dfrac{\lambda_0^2-\lambda_{\infty}^2}{2} \right)
\lim_{n\to\infty}\dfrac{1}{n}\sum_{i=1}^n\EV_0[ X_{i}^2]-
\left(\dfrac{\mu_0^2-\mu_{\infty}^2}{2} \right).
\end{split}
\end{align*}

Now, recall the basic result that if $\{a_n\}_{n\ge1}$ is a convergent sequence such that $\lim_{n\to\infty} a_n\triangleq a$, then its so-called Ce{\`s}aro mean sequence (see, e.g.,~\cite[Chapter~V,~Section~5.4,~p.~96]{Hardy:Book1991}), i.e., the sequence $\{b_n\}_{n\ge1}$ formed as
\begin{align*}
b_n
&\triangleq
\dfrac{1}{n}\sum_{i=1}^n a_i,\;\;n\ge1,
\end{align*}
is also convergent with the same limit, i.e., $\lim_{n\to\infty} b_n=a$; see, e.g.,~\cite[Chapter~V,~Section~5.7,~p.~100--102]{Hardy:Book1991}. Hence, with the aid of the $\Pr_d$-stationarity of the AR(1) model in~\eqref{eq:def-AR-eqn} under $d=\{0,\infty\}$, i.e., the assumption that $\vert\lambda_d\vert<1$, $d=\{0,\infty\}$, it is easily established that
\begin{align*}
\begin{split}
\lim_{n\to\infty}\dfrac{1}{n}\sum_{i=1}^n\EV_{0}[X_i]
&=
\lim_{n\to\infty}\EV_0[X_n]=\dfrac{\mu_0}{1-\lambda_0},\\
\lim_{n\to\infty}\dfrac{1}{n}\sum_{i=1}^n\EV_0[X_i^2]
&=
\lim_{n\to\infty}\EV_0[X_n^2]
=
\dfrac{\mu_0^2}{(1-\lambda_0)^2}+\dfrac{1}{1-\lambda_0^2},\\
\lim_{n\to\infty}\dfrac{1}{n}\sum_{i=1}^n\EV_0[X_iX_{i-1}]
&=
\lim_{n\to\infty}\EV_0[X_nX_{n-1}]
=
\dfrac{\mu_0^2}{(1-\lambda_0)^2}+\dfrac{\lambda_0}{1-\lambda_0^2}.
\end{split}
\end{align*}
The desired formula for the KL number follows once these computations are plugged into the
right-hand side of the expression above and simplifying it.
\end{proof}

As a ``sanity check'', it is easily verified from~\eqref{eq:AR-data-model-KL-formula} that $\KL\in[0,+\infty]$ for all possible values of the four model
parameters $\{\lambda_{\infty},\lambda_0,\mu_{\infty},\mu_{0}\}$. The smallest value $\KL=0$ is achieved for the case when there is no change at all, i.e., when $\mu_{\infty}=\mu_0$ and $\lambda_{\infty}=\lambda_0\,(\triangleq\lambda)$, $\lambda\in(-1,1)$. On the other hand, $\KL\to+\infty$, as $\lambda_0\to\pm1$.

We now review several special cases of the AR(1) model that will be considered in the sequel. To start with, observe that in the case when only the drift changes, i.e., when
$\lambda_{\infty}=\lambda_0\,(\triangleq\lambda)$, the obtained formula for $\KL$ reduces to
\begin{align}\label{kl_div2}
\KL
=
\dfrac{(\mu_0-\mu_{\infty})^2}{2},
\end{align}
which is independent of $\lambda$ and is a (symmetric) function only of $\vert\mu_0-\mu_{\infty}\vert$, i.e., of the discernibility in the drift of the process, $\{X_n\}_{n\ge0}$, between the pre- and post-change regimes. This is expected, and~\eqref{kl_div2} is also the same as the KL number
between observations that are both i.i.d.\
pre- and post-change; $\mathcal{N}(\mu_{\infty},1)$ and $\mathcal{N}(\mu_0,1)$, respectively. This suggests that the various change-point detection procedures in this case should be
similar in performance to detecting a change in i.i.d.\ processes. This result has been established in~\cite{Moustakides:IT98} where the CUSUM chart is shown to be optimal in the Lorden sense.

\ignore{
If the change is only in the correlation coefficient, i.e., if $\mu_{\infty}=\mu_{0}(=\mu)$ and $\lambda_0\neq\lambda_\infty$, then the expression for $\LLR_{n}$ simplifies as
\begin{align}
\LLR_{n+1}
&=
X_n \cdot \left( \frac{ \lambda_0 - \lambda_{\infty}} {\sigma^2} \right)
\cdot \left[ X_{n+1} - X_n \cdot \left( \frac{\lambda_0 + \lambda_{\infty}}{2}
\right) - \mu \right].
\end{align}

In the pre-change i.i.d.\ case ($\lambda_{\infty} = 0$) with $\mu_{\infty} = 0$,
we have
\begin{eqnarray}
Z_{n+1} = \frac{ \mu_0 }{2 \sigma^2} \cdot
\left( 2 X_{n+1} - 2\lambda_0 X_n - \mu_0 \right) +
\frac{ \lambda_0 X_n }{2 \sigma^2 } \cdot
\left( 2 X_{n+1} - \lambda_0 X_n \right) . \nonumber
\end{eqnarray}
}

\begin{figure*}[htb!]
\begin{center}
\begin{tabular}{cc}
\includegraphics[height=2.1in,width=2.8in] {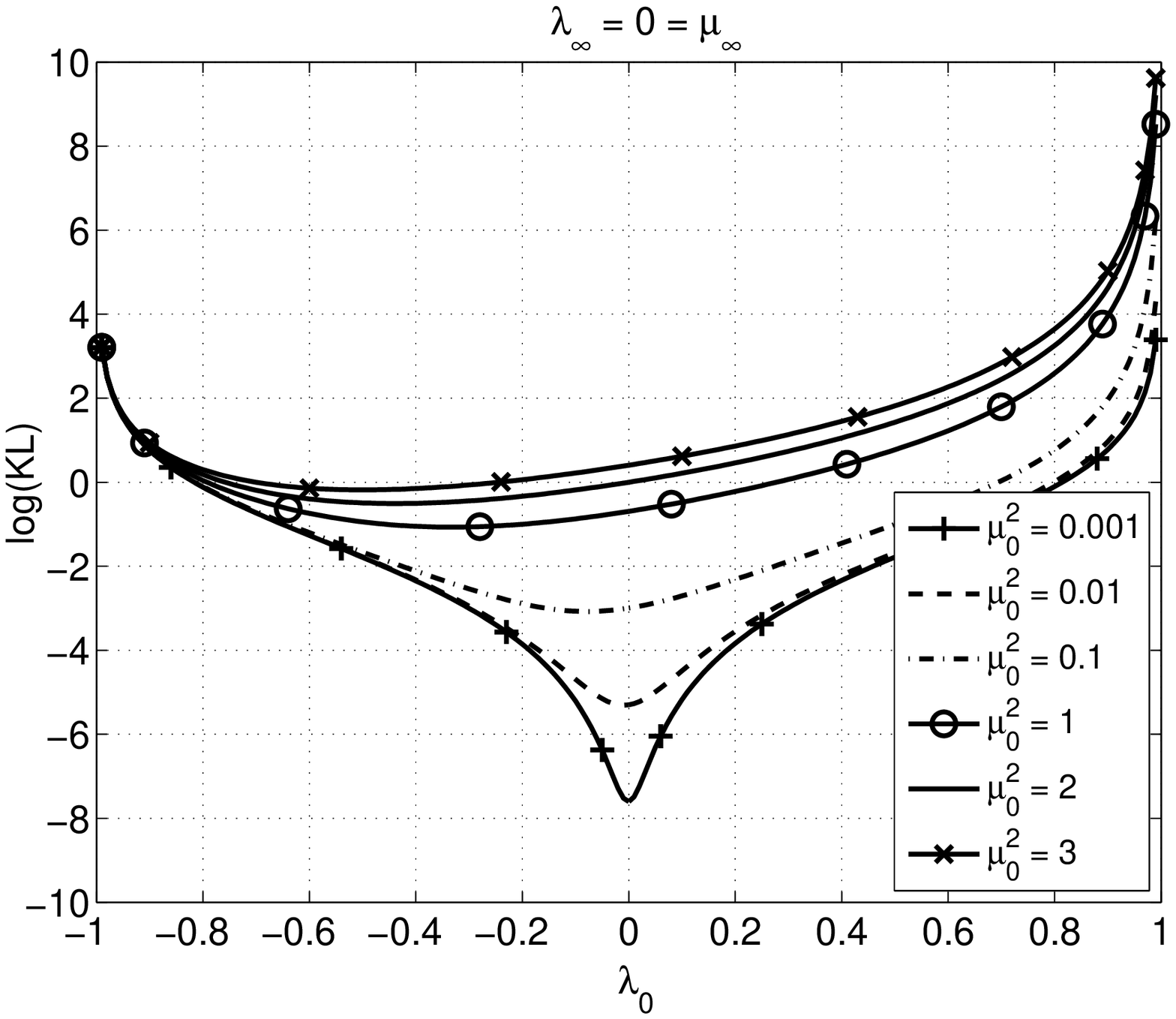}
&
\includegraphics[height=2.1in,width=2.8in] {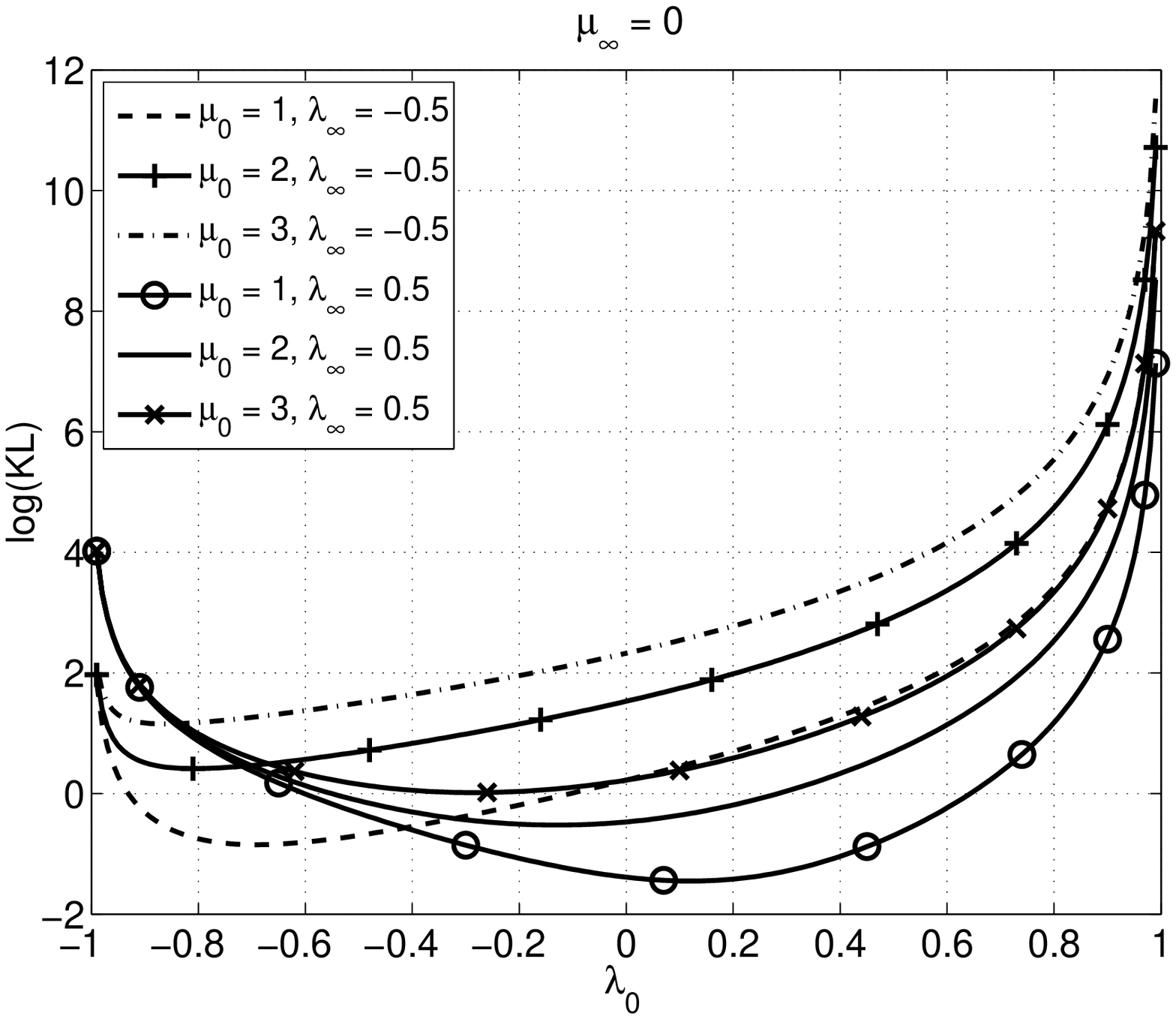}
\\ {\hspace{0.35in}} (a) & {\hspace{0.35in}} (b)
\\
\includegraphics[height=2.1in,width=2.8in] {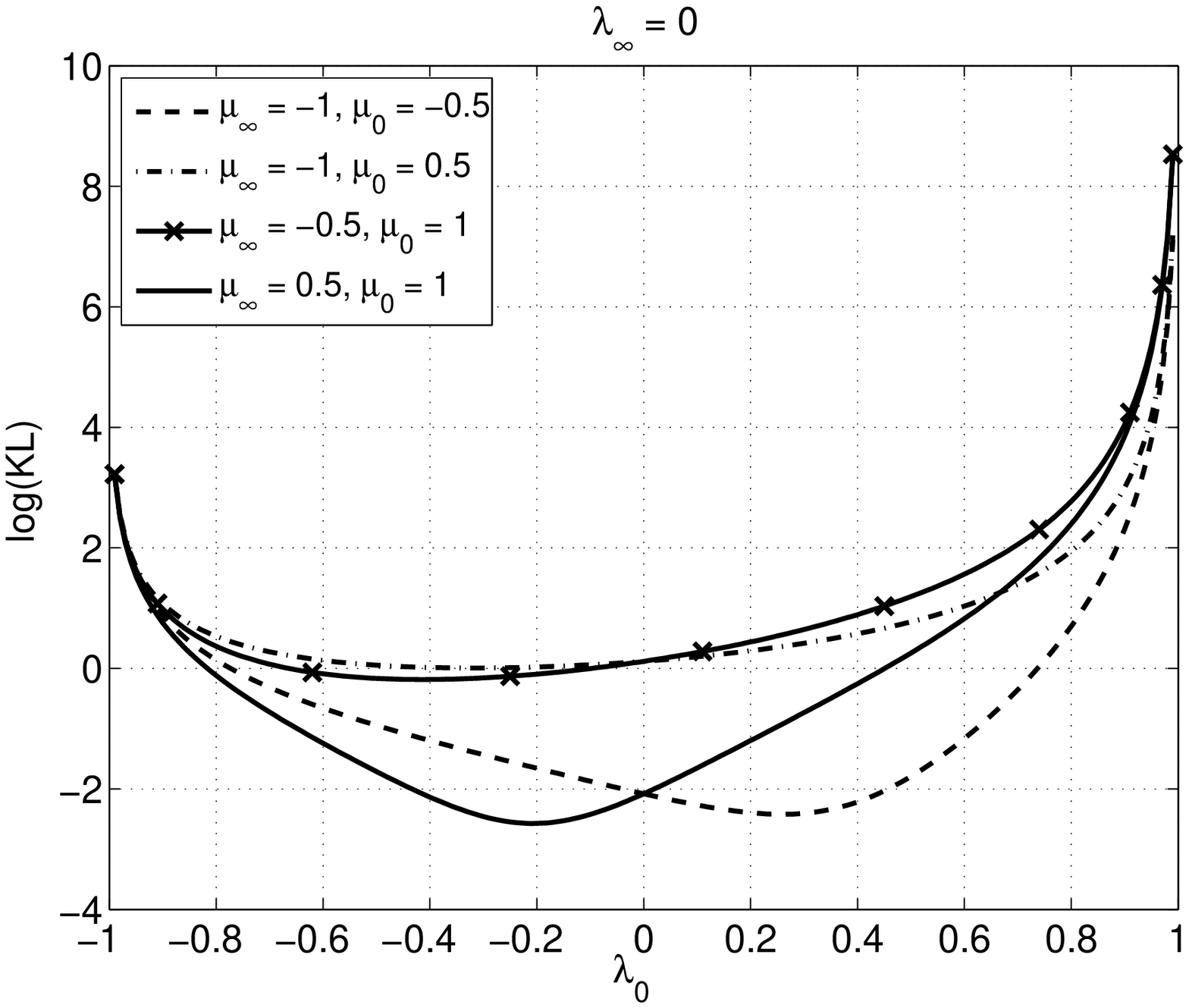}
&
\includegraphics[height=2.1in,width=2.8in] {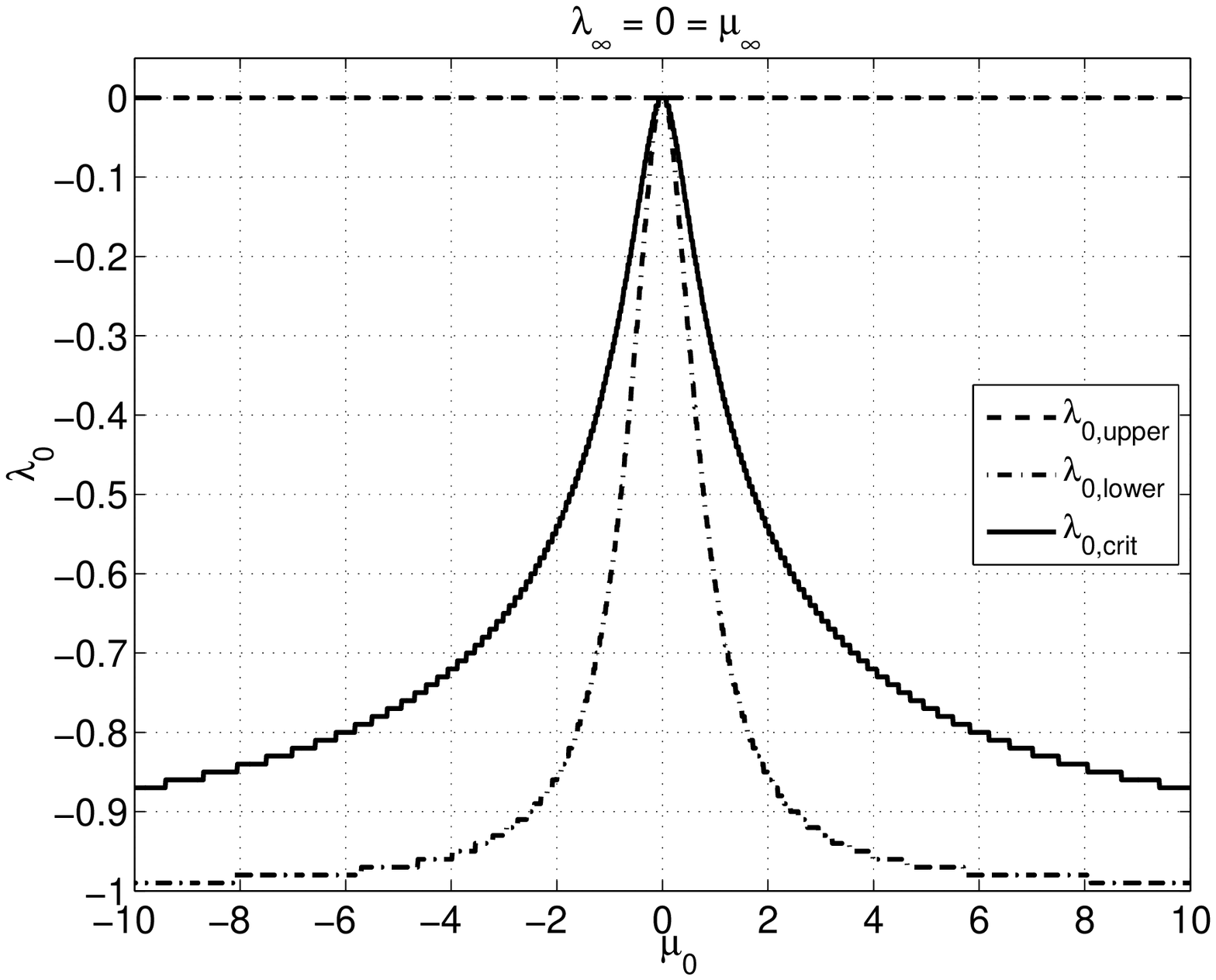}
\\
{\hspace{0.35in}} (c) & {\hspace{0.35in}} (d)
\\
\includegraphics[height=2.1in,width=2.8in] {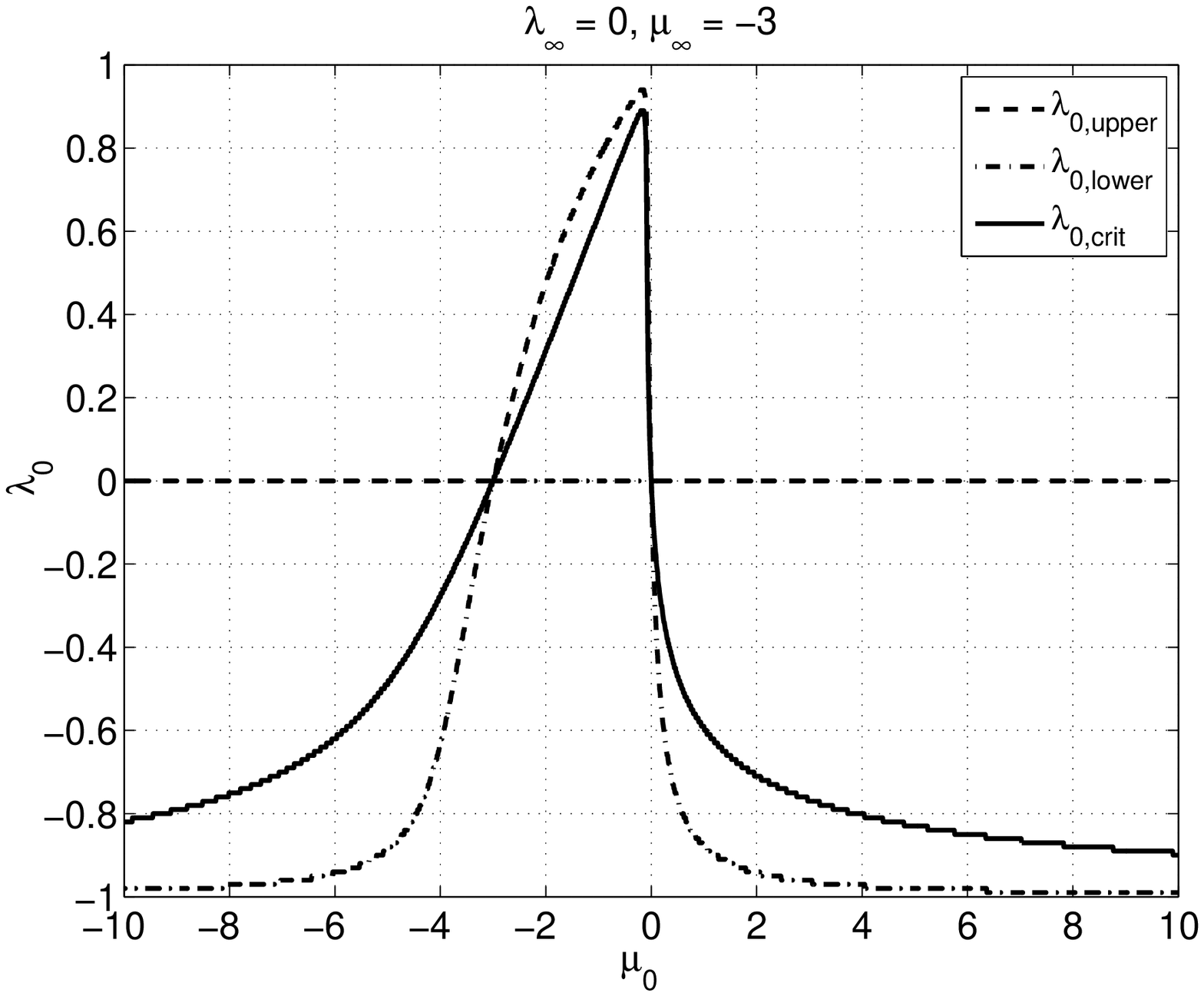}
&
\includegraphics[height=2.1in,width=2.8in] {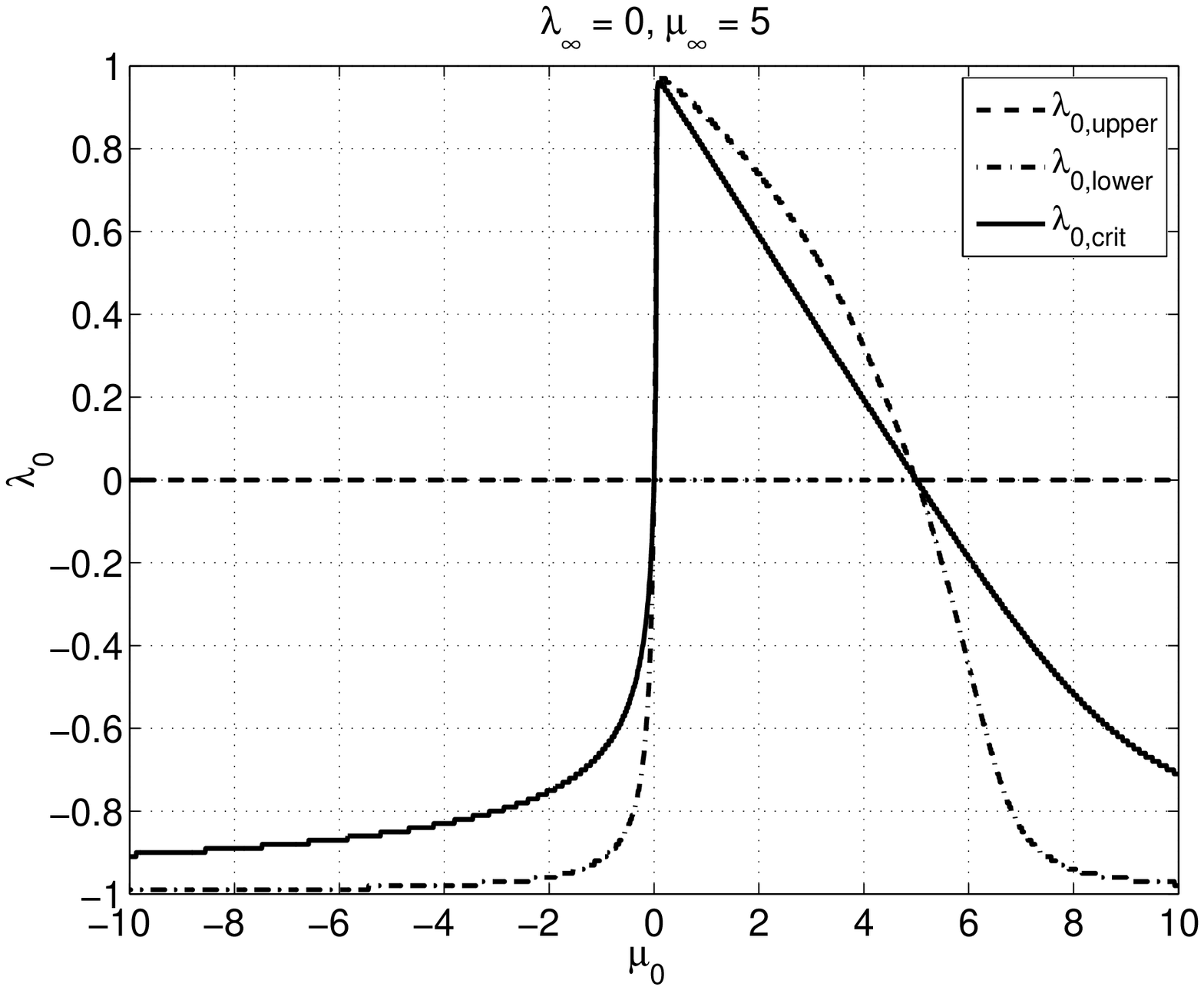}
\\
{\hspace{0.35in}} (e) & {\hspace{0.35in}} (f)
\end{tabular}
\caption{\label{fig_kl}
KL number for different sets of AR process parameters: (a) i.i.d.\
pre-change ($\lambda_{\infty} = 0$) with $\mu_{\infty} = 0$, (b) $\mu_{\infty} =
0$, and (c) $\lambda_{\infty} = 0$. 
(d)-(f) $\lambda_{0, \hsppp {\sf upper}}$, $\lambda_{0, \hsppp {\sf lower}}$
and $\lambda_{0, \hsppp {\sf crit}}$ as a function of $\mu_0$ in
the i.i.d.\ pre-change setting ($\lambda_{\infty} = 0$) for different
values of $\mu_{\infty}$.}
\end{center}
\vspace{-5mm}
\end{figure*}

The KL number in the i.i.d.\ pre-change case ($\lambda_{\infty} = 0$) with
$\mu_{\infty} = 0$ is plotted as a function of $\lambda_0$ for different values of $\mu_0$ in Fig.~\ref{fig_kl}(a).
Similarly, the KL number corresponding to the $\mu_{\infty} = 0$ and
$\lambda_{\infty} = 0$ settings are plotted as a function of $\lambda_0$
for different combinations of parameters in Figs.~\ref{fig_kl}(b) and~(c), respectively. These plots clearly
illustrate the existence of a certain {\em worst-case} (in the sense of
detectability) post-change correlation that leads to the smallest value of
$\KL$ conditioned on the other model parameters. To understand this behavior of $\KL$,
we now study these special cases more carefully.

In the case where $\mu_{\infty}=\mu_{0}(\triangleq\mu)$, $\KL$ reduces to
\begin{align}\label{kl_div3}
\KL
&\triangleq
\dfrac{(\lambda_0-\lambda_{\infty})^2}{2(1-\lambda_0^2)} \cdot
\left[1+\mu^2\, \left(\frac{1+\lambda_0}{1-\lambda_0} \right) \right].
\end{align}
It can be checked that
\begin{align}
\dfrac{\partial\KL}{\partial\lambda_0}
&\triangleq
\dfrac{\lambda_0-\lambda_{\infty}}{(1-\lambda_0)^2}\cdot
\left[\mu^2\cdot\dfrac{1-\lambda_{\infty}}{1-\lambda_0}+\frac{1-\lambda_0\lambda_{\infty}}{(1+\lambda_0)^2} \right].
\end{align}
Thus, for a fixed $\lambda_{\infty}$ and $\mu$, ${\KL}$
in~(\ref{kl_div3}) decreases in $\lambda_0$ for all $\lambda_0 <
\lambda_{\infty}$ with ${\KL} = 0$ obtained at $\lambda_0 =
\lambda_{\infty}$. ${\KL}$ then starts increasing from $0$ as
$\lambda_0$ increases past $\lambda_{\infty}$. Specifically, a higher
correlation in the post-change process increases ${\KL}$ provided
that both the processes are positively correlated. 

When the observations are i.i.d.\ pre-change ($\lambda_{\infty} = 0)$
with $\mu_{\infty} = 0$, ${\KL}$ reduces to
\begin{eqnarray}
\KL = \frac{1}{2(1 - \lambda_0)} \cdot \left[ \frac{\lambda_0^2}
{1 + \lambda_0} + \frac{ \mu_0^2}{1 - \lambda_0} \right].
\label{kl_div4}
\end{eqnarray}
For a fixed $\mu_0$, it can be checked that
\begin{eqnarray}
\frac{ \partial \KL }{\partial \lambda_0} = \frac{1}{(1 - \lambda_0)^3}
\cdot \left[ \frac{\lambda_0(1 - \lambda_0)}{ ( 1 + \lambda_0)^2} + \mu_0^2
\right]. \nonumber
\end{eqnarray}
From the above equation, a trivial calculation shows that ${\KL}$ is globally
minimized at $\lambda_{0, \hsppp {\sf crit}}$, defined as,
\begin{eqnarray}
\lambda_{0, \hsppp {\sf crit}} \triangleq
\frac{ \sqrt{ 8 \mu_0^2 + 1} - (2 \mu_0^2 + 1)
}{ 2 ( \mu_0^2 - 1)} .
\nonumber
\end{eqnarray}
Further, the KL number of the correlated process is smaller than
$\mu_0^2/2$ (the KL number of the corresponding i.i.d.\ problem)
if and only if
$\lambda_0$ belongs to the interval $\left[ \lambda_{0, \hsppp {\sf lower}},
\hsppp \lambda_{0, \hsppp {\sf upper}} \right]$, where
\begin{eqnarray}
\lambda_{0, \hsppp {\sf lower}} \triangleq \max\left( -1, \frac{1}{2} \cdot
\left[1 - \sqrt{\frac{9\mu_0^2+1}{\mu_0^2+1} } \right]
\right) \hspp {\rm and} \hspp
\lambda_{0, \hsppp {\sf upper}} = 0. \nonumber
\end{eqnarray}
It can be checked that $-1 \leq \lambda_{0, \hsppp {\sf lower}}
\leq \lambda_{0, \hsppp {\sf crit}} \leq \lambda_{0, \hsppp {\sf upper}}
= 0$ for all $\mu_0$ with both $\lambda_{0, \hsppp {\sf lower}}$
and $\lambda_{0, \hsppp {\sf crit}}$ decreasing in $\mu_0^2$.
Further, we also have
\begin{eqnarray}
\lambda_{0, \hsppp {\sf lower}} \rightarrow 0 \hspp {\rm and} \hspp
\lambda_{0, \hsppp {\sf crit}} \rightarrow 0 & {\rm as} &
\mu_0 \rightarrow 0
\nonumber \\
\lambda_{0, \hsppp {\sf lower}} \rightarrow -1 \hspp {\rm and} \hspp
\lambda_{0, \hsppp {\sf crit}} \rightarrow -1 & {\rm as} &
\mu_0 \rightarrow \infty.
\nonumber
\end{eqnarray}
See Fig.~\ref{fig_kl}(d) for a plot of the three quantities as a function
of $\mu_0$. In Figs.~\ref{fig_kl}(e)-(f), these three quantities are
plotted as a function of $\mu_0$ when $\mu_{\infty}= -3$
and $\mu_{\infty} = 5$, respectively. Note that, in general, the
behavior of all the three quantities is asymmetric in $\mu_0$.

More generally, if $\lambda_{\infty} \neq 0$, the behavior of
$\lambda_{0, \hsppp {\sf upper}}$, $\lambda_{0, \hsppp {\sf lower}}$
and $\lambda_{0, \hsppp {\sf crit}}$ as a function of
$\lambda_{\infty}$ for different $\mu_{\infty}$ and
$\mu_0$ values is presented in Figs.~\ref{fig_lambda_crit}(a)-(b).
From Figs.~\ref{fig_kl}(d)-(f) and Figs.~\ref{fig_lambda_crit}(a)-(b),
we observe that either $\lambda_{0, \hsppp {\sf lower}}$ or
$\lambda_{0, \hsppp {\sf upper}}$ equals $\lambda_{\infty}$ for every
case in the model parameter space. The observed trends depend on the precise
relationship between $\mu_{\infty}$, $\mu_0$ and $0$ and can be summarized
as follows:
\begin{eqnarray}
\begin{array}{ccl}
0 < \mu_0 < \mu_{\infty} \hsp {\rm or} \hsp \mu_{\infty} < \mu_0 < 0
& \Longrightarrow & \lambda_{0, \hsppp {\sf lower}} = \lambda_{\infty}
\\
\mu_0 < 0 < \mu_{\infty} \hsp {\rm or} \hsp 0 < \mu_{\infty} < \mu_0
\hsp {\rm or} \hsp \mu_0 < \mu_{\infty} < 0 \hsp {\rm or} \hsp
\mu_{\infty} < 0 < \mu_0
& \Longrightarrow &
\lambda_{0, \hsppp {\sf upper}} = \lambda_{\infty} .
\end{array}
\nonumber
\end{eqnarray}

To summarize the above analysis, the KL number, $\KL\triangleq\KL(\mu_\infty,\mu_0,\lambda_\infty,\lambda_0)$, associated with the
AR(1) model in~\eqref{eq:def-AR-eqn} and given by~\eqref{eq:AR-data-model-KL-formula}, is always larger than the KL number for the basic i.i.d.\ problem for post-change
correlation above and below certain cut-off values. In the intervening
regime, the KL number is smaller than the i.i.d.\ problem with a worst-case
$\lambda_0$ given by $\lambda_{0,\hsppp {\sf crit}}$. These trends are illustrated
pictorially in Figs.~\ref{fig_kl}(a)-(c).
\begin{figure}[!htb]
\begin{center}
\begin{tabular}{cc}
\includegraphics[height=2.3in,width=2.8in] {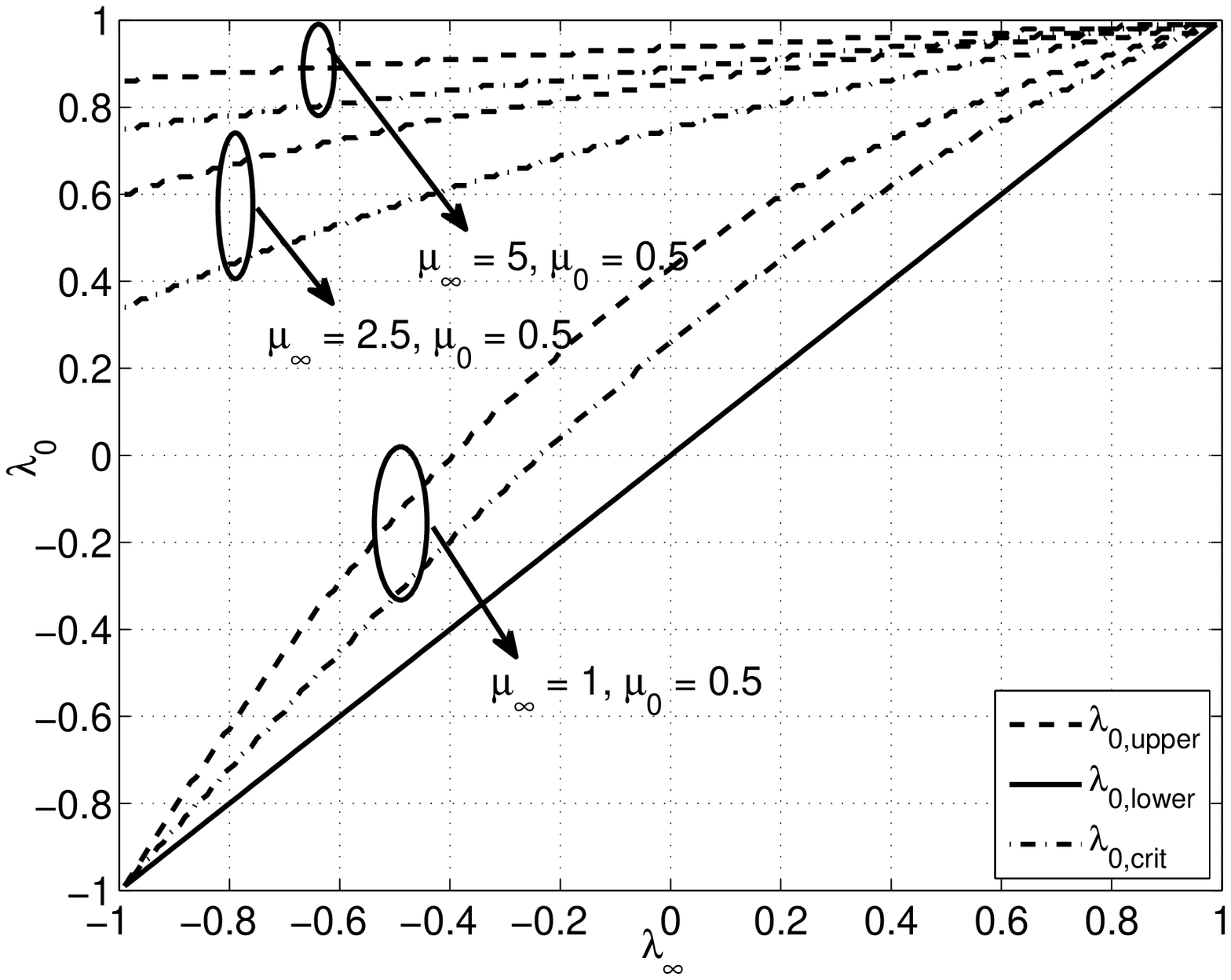}
&
\includegraphics[height=2.3in,width=2.8in] {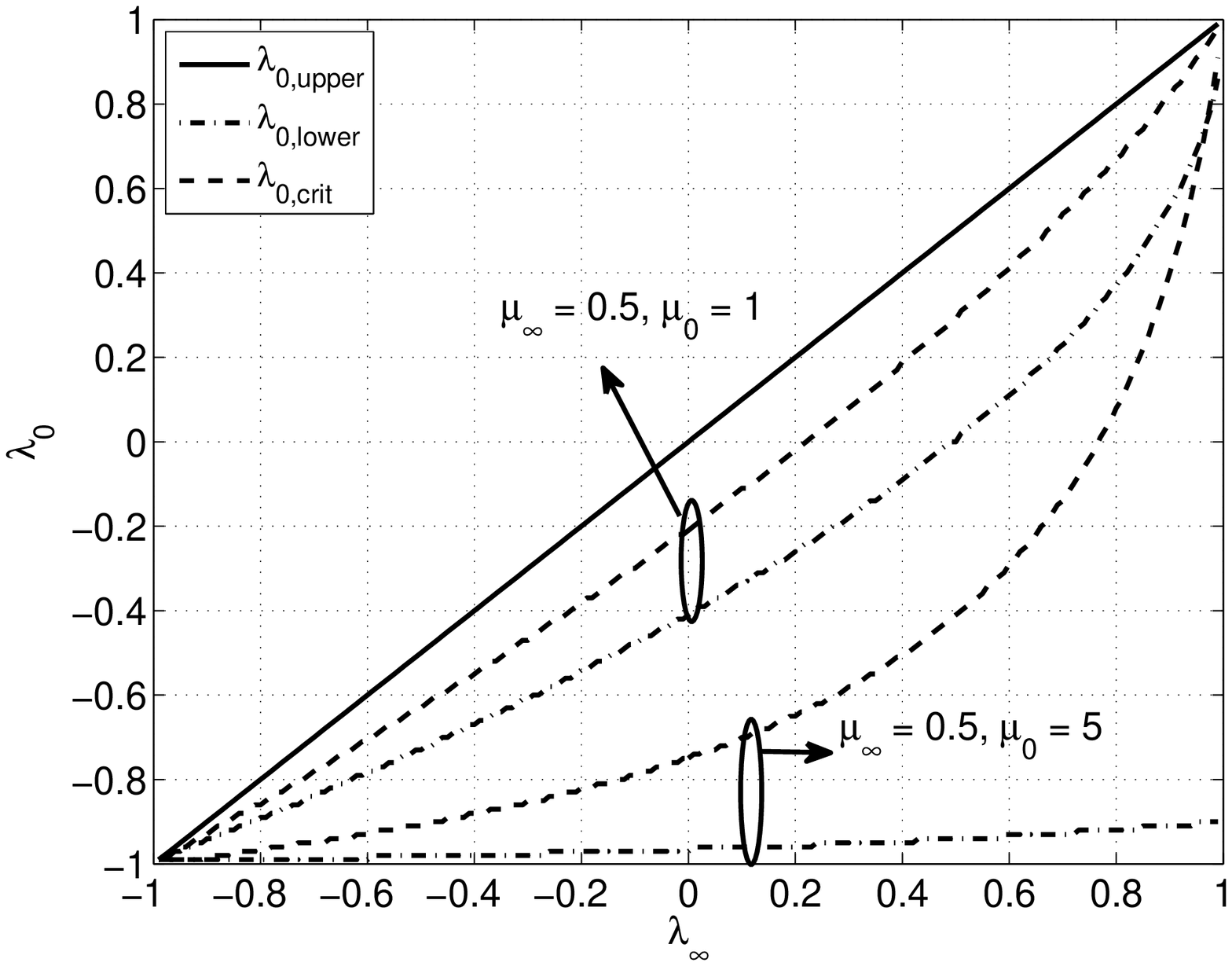}
\\ {\hspace{0.35in}} (a) & {\hspace{0.35in}} (b)
\end{tabular}
\caption{\label{fig_lambda_crit}
$\lambda_{0, \hsppp {\sf upper}}$, $\lambda_{0, \hsppp {\sf lower}}$
and $\lambda_{0, \hsppp {\sf crit}}$ as a function of $\lambda_{\infty}$
with different $\mu_{\infty}$ and $\mu_0$ values.}
\end{center}
\end{figure}

\section{Performance Evaluation}
\label{sec:performance-evaluation}

This section is devoted to developing a numerical framework to evaluate the performance
of the CUSUM chart~\eqref{eq:def-Tcs}--\eqref{eq:statistic-CS-def} and the SR procedure~\eqref{eq:def-Tsr}--\eqref{eq:statistic-SR-def} when applied to the AR(1) model in~\eqref{eq:def-AR-eqn}. Specifically, for each of the stopping times---either $\T=\TCS$
if it is the CUSUM chart, or $\T=\TSR$ if it is the SR procedure---the framework is ``tailored'' to two antagonistic performance characteristics:\begin{inparaenum}[\itshape a)]\item the usual ``in-control'' ARL to false alarm, i.e., $\ARL(\T)$, 
and \item Pollak's~\cite{Pollak:AS85} Supremum (conditional) Average Detection Delay, i.e., $\SADD(\T)$. 
\end{inparaenum} 

The framework here is a build-up to the one previously offered and applied in~\cite{Tartakovsky+etal:IWSM2009,Moustakides+etal:CommStat09,Moustakides+etal:SS11} for the
i.i.d.\ model (with no cross-dependence in the observed data); see also, e.g.,~\cite{Polunchenko+etal:MIPT2013,Polunchenko+etal:SA2014,Polunchenko+etal:ASMBI2014}. Accordingly, just as the prototype framework of~\cite{Tartakovsky+etal:IWSM2009,Moustakides+etal:CommStat09,Moustakides+etal:SS11,Polunchenko+etal:MIPT2013,Polunchenko+etal:SA2014,Polunchenko+etal:ASMBI2014}, our framework is also developed in two stages: we first derive a renewal integral equation for each performance metric involved, and then, as neither one of the obtained equations can be solved analytically, we supply a numerical method to do so, and carry out an analysis of the method's accuracy. What is new in the setting considered here is that the integral equations are not one- but two-dimensional, and (therefore) the numerical method is not deterministic but rather a Monte-Carlo-type estimation technique of {\em prescribed proportional closeness}, a criterion considered, e.g., in~\cite[p.~339]{Ehrenfeld+Littauer:Book1964},~\cite{Simons+Zacks:Stanford1967,Nadas:AMS1969,Willson+Folks:CommStat-SA1983,Zacks:AMS1966,Zacks:HofS-AinR2001}.

\subsection{The renewal equations}
\label{ssec:integral-equations}

To compute the performance of the CUSUM chart~\eqref{eq:def-Tcs}--\eqref{eq:statistic-CS-def} and
that of the SR procedure~\eqref{eq:def-Tsr}--\eqref{eq:statistic-SR-def} when applied to the AR(1) model in~\eqref{eq:def-AR-eqn}, we now derive analytically exact renewal equations on the performance characteristics of interest. 
To begin with, it is direct to see from~\eqref{eq:instant-LR-formula}
that $\LR_n(X_n,X_{n-1})$ is (absolutely) continuous with respect to both $X_n$ and $X_{n-1}$ for all $n\ge1$.
Next, given the observation series, $\{X_n\}_{n\ge0}$, consider the generic detection procedure associated with the generic stopping time
\begin{align}\label{eq:generic-T-def}
\T(x_0,y_0,A)
&\triangleq
\inf\big\{n\ge1\colon Y_n\ge A\big\},\;y_0\ge0,\;x_0\in\mathbb{R},\; A>0,
\end{align}
whose decision-making is based off the generic detection statistic $\{Y_n\}_{n\ge0}$ defined as
\begin{align*}
Y_{n}
&=
\Psi(Y_{n-1})\LR_{n}(X_{n-1},X_{n})\;\text{for}\;n=1,2,\ldots\;\text{with}\; Y_0=y_0\ge 0\;\text{and}\;X_0=x_0\in\mathbb{R},
\end{align*}
where $\Psi(z)$ is a (sufficiently) smooth non-negative-valued function defined (at least) for $z\ge0$, and $y_0\ge0$ and $x_0\in\mathbb{R}$ are given constants; the assumptions that $\Psi(z)\ge0$ for $z\ge0$ and that $y_0\ge0$ are necessary to ensure that $\{Y_n\}_{n\ge0}$ is almost surely non-negative under any probability measure, so that the two-dimensional (homogeneous) Markov process, $\{(Y_n,X_n)\}_{n\ge0}$, is restricted to the half-plane $[0,\infty)\times\mathbb{R}$.

Now note that since the choice of $\Psi(z)$ is flexible, the generic stopping time $\T(x,y,A)$ can be seen to describe a rather large class of LR-based detection procedures; in particular, if $\Psi(z)=1+z$, then $\{Y_n\}_{n\ge0}$ and the SR detection statistic, $\{R_n\}_{n\ge0}$, are identical, and therefore, for this choice of $\Psi(z)$ the generic stopping time $\T(x,y,A)$ and that associated with the SR procedure coincide. Hence, the SR procedure is a special case of $\T(x,y,A)$, as is the CUSUM chart; indeed, if $\Psi(z)=\max\{1,z\}$, then $\{Y_n\}_{n\ge0}$ is the CUSUM detection statistic, $\{V_n\}_{n\ge0}$, and therefore, in this case the generic stopping time $\T(x,y,A)$ is no different from that associated with the CUSUM chart. This flexibility of the generic stopping time $\T(x,y,A)$ can be used to study simultaneously the performance of not only the CUSUM chart or the SR procedure, but also of a far larger number of other procedures (e.g., EWMA procedure).

Let $\Pr_d(Y_{n}\le y_2,X_{n}\le x_2|Y_{n-1}=y_1,X_{n-1}=x_1)$, $d=\{0,\infty\}$, denote the {\em transition probability function} to describe the evolution (in time, $n$) of the two-dimensional (homogeneous) Markov process $\{(Y_n,X_n)\}_{n\ge0}$ under probability measure $\Pr_d$, $d=\{0,\infty\}$; note that $\Pr_d(Y_{n}\le y_2,X_{n}\le x_2|Y_{n-1}=y_1,X_{n-1}=x_1)$, $d=\{0,\infty\}$, is independent of $n$. Let
\begin{align}\label{eq:K-def}
K_d(y_2,x_2|y_1,x_1)
&\triangleq
\dfrac{\partial^2}{\partial y_2\partial x_2}\Pr_d(Y_{n}\le y_2,X_{n}\le x_2|Y_{n-1}=y_1,X_{n-1}=x_1),
\end{align}
$d=\{0,\infty\}$, be the respective {\em transition probability density kernel}; it is clear that $K_d(y_2,x_2|y_1,x_1)$, $d=\{0,\infty\}$, is independent of $n$ as well. A straightforward calculation shows that 
\begin{align*}
& \Pr_d(Y_{n}\le y_2,X_{n}\le x_2|Y_{n-1}=y_1,X_{n-1}=x_1)
\nonumber \\
& =
\left\{
\begin{array}{cc}
\Phi(\min\{x_2, \xi(x_1,y_1,y_2)\}-\mu_d-\lambda_d x_1),
& {\rm if} \hspp \mu_0-\mu_\infty+x_1(\lambda_0-\lambda_\infty)\ge0
\\
\Phi(x_2-\mu_d-\lambda_d x_1)-\Phi(\xi(x_1,y_1,y_2)-\mu_d-\lambda_d x_1),
& {\rm if} \hspp \mu_0-\mu_\infty+x_1(\lambda_0-\lambda_\infty)<0,
\end{array} \right.
\end{align*}
where
\begin{align*}
\xi(x_1,y_1,y_2)
&\triangleq
\dfrac{1}{\mu_0-\mu_\infty+x_1(\lambda_0-\lambda_\infty)}\log\left(\dfrac{y_2}{\Psi(y_1)}\right)
+
\dfrac{\mu_0+\mu_\infty+x_1(\lambda_0+\lambda_\infty)}{2},
\end{align*}
and
\begin{align*}
\Phi(x)
&\triangleq
\dfrac{1}{\sqrt{2 \pi}} \int_{-\infty}^x e^{-\tfrac{t^2}{2}}dt,
\end{align*}
i.e., the standard Gaussian cdf.

We are now in a position to derive the first renewal equation of interest, viz., that on the first moment of $\T(x_0,y_0,A)$ under measure $\Pr_\infty$, i.e., for the ARL to false alarm of the generic stopping time $\T(x_0,y_0,A)$. Specifically, for notational brevity, denote $\ell(x,y,A)\triangleq\EV_\infty[\T(x,y,A)]$. By conditioning on the first observation, $X_1$, and using a routine renewal argument akin to that made in~\cite{Moustakides+etal:SS11} we obtain
\begin{align}\label{eq:ARL-int-eqn}
\ell(x_1,y_1,A)
&=
1+\int_{-\infty}^{\infty}\int_0^A K_\infty(y_2,x_2|y_1,x_1)\,\ell(x_2,y_2,A)\,dy_2\,dx_2.
\end{align}

The double integral in the right-hand side of this equation cannot be separated, since $Y_n$ and $X_n$ are correlated for all $n\ge1$. However, it is because $Y_n$ and $X_n$ are correlated for all $n\ge1$, the double integral is effectively a single integral, and is taken along the curve given by the points $(y_2,x_2)$ for which $K_\infty(y_2,x_2|y_1,x_1)\neq0$. These points satisfy the equation $u(x_1,y_1,y_2)=x_2$, or written explicitly
\begin{align*}
\dfrac{1}{\mu_0-\mu_\infty+x_1(\lambda_0-\lambda_\infty)}\log\left(\dfrac{y_2}{\Psi(y_1)}\right)
+
\dfrac{\mu_0+\mu_\infty+x_1(\lambda_0+\lambda_\infty)}{2}
&=
x_2,
\end{align*}
and for all other values of $(x_1,y_1)$ and $(x_2,y_2)$ the integral is zero, and $\ell(x_1,y_1,A)=1$ irrespective of $A>0$.

Thus, the double integral in the right-hand side of~\eqref{eq:ARL-int-eqn} is to be understood in the Riemann--Stieltjes sense with the measure of integration being $\Pr_\infty(Y_{n}\le y_2,X_{n}\le x_2|Y_{n-1}=y_1,X_{n-1}=x_1)$. Since the latter is a two-dimensional cdf, it is clearly a function of bounded variation, and therefore the existence of the integral is justified.

Next, introduce $\delta_0(x,y,A)\triangleq\EV_0[\T(x,y,A)]$, and observe that
\begin{align}\label{eq:ADD0-int-eqn}
\delta_0(x_1,y_1,A)
&=
1+\int_{-\infty}^{\infty}\int_0^A K_0(y_2,x_2|y_1,x_1)\,\delta_0(x_2,y_2,A)\,dy_2\,dx_2,
\end{align}
which is an exact ``copy'' of equation~\eqref{eq:ARL-int-eqn} except that $K_\infty(y_2,x_2|y_1,x_1)$ is replaced with $K_0(y_2,x_2|y_1,x_1)$.
For $k\ge1$, since $\{R_n^{r=x}\}_{n\ge0}$ is Markovian, one can establish the recursion
\begin{align}\label{eq:ADDk-recursion}
\delta_{k+1}(x_1,y_1,A)
&=
\int_{-\infty}^{\infty}\int_0^A{K}_{\infty}(y_2,x_2|y_1,x_1)\,\delta_k(x_2,y_2,A)\,dy_2\,dx_2,\; k\ge0,
\end{align}
with $\delta_0(x,y,A)$ first found from equation~\eqref{eq:ADD0-int-eqn}; cf.~\cite{Moustakides+etal:SS11}. Using this recursion one can generate the entire functional sequence $\{\delta_k(x,y,A)\}_{k\ge0}$ by repetitive application of the linear integral operator
\begin{align*}
\mathcal{K}_{\infty}\circ u
&\triangleq
[\mathcal{K}_{\infty}\circ u](x_1,y_1)
\triangleq
\int_0^A {K}_{\infty}(y_2,x_2|y_1,x_1)\,u(x_2,y_2)\,dy_2\,dx_2,
\end{align*}
where $u(x,y)$ is assumed to be sufficiently smooth inside the strip $\mathbb{R}\times[0,A]$. Temporarily deferring formal discussion of this operator's properties, note that using this operator notation, recursion~\eqref{eq:ADDk-recursion} can be rewritten as $\delta_{k+1}=\mathcal{K}_{\infty}\circ \delta_{k}$, $k\ge0$, or equivalently, as $\delta_{k}=\mathcal{K}_{\infty}^{k}\circ \delta_{0}$, $k\ge0$, where
\begin{align*}
\mathcal{K}_{\infty}^{k}\circ u
&\triangleq
\underbrace{\mathcal{K}_{\infty}\circ\cdots\circ\mathcal{K}_{\infty}}_{\text{$k$ times}}\circ\,u\;\text{for}\;k\ge1,
\end{align*}
and $\mathcal{K}_{\infty}^{0}$ is the identity operator from now on denoted as $\mathbb{I}$, i.e., $\mathcal{K}_{\infty}^{0}\circ u=\mathbb{Id}\circ u\triangleq u$. Similarly, in the operator form, equation~\eqref{eq:ARL-int-eqn} can be rewritten as $\ell=1+\mathcal{K}_\infty\circ\ell$, and equation~\eqref{eq:ADD0-int-eqn} can be rewritten as $\delta_0=1+\mathcal{K}_0\circ\delta_0$.

%
%
%
\begin{lemma}
For the generic detection procedure $\T(x_0,0,A)$ given by~\eqref{eq:generic-T-def} it is true that
\begin{align*}
\SADD(\T(x_0,0,A))
&\triangleq
\sup_{0\le k<\infty}\ADD_k(\T(x_0,0,A))
\nonumber \\
& =
\ADD_0(\T(x_0,0,A))\triangleq\EV_0[\T(x_0,0,A)].
\end{align*}
\end{lemma}

Equations~\eqref{eq:ARL-int-eqn} and~\eqref{eq:ADD0-int-eqn} provide
a ``complete package'' to compute any of the desired performance characteristics of the CUSUM chart and those of the SR procedure. The question to be considered next is to compute
these characteristics in practice.

\subsection{The numerical solution and its accuracy}

The renewal equations established in the preceding subsection on the performance metrics of interest are (two-dimensional) Fredholm (linear) integral equations of the second kind. Such equations rarely permit an analytical, closed-form solution, even in a single dimension. Hence, a numerical method is in order, and it is the aim of this subsection to propose one.

The branch of numerical analysis concerned with the design and analysis of numerical schemes to solve Fredholm integral equations of the second kind has plenty of powerful methods for efficient solution of these equations in one dimension. However, even then the dimension is two, things get much more complicated. We propose to consider a Markov Chain Monte-Carlo (MCMC) technique.

We start with an observation that although the integral involved in the equation of interest is a double integral, it can actually be reduced to a single equation, as $X_n$ and $R_n$ are dependent on one another. Specifically, this single integral is along the curve described in the $(X_n,R_n)$ space by the relation $R_n=(1+R_{n-1})\LR_n$, where both $X_{n-1}$ and $R_{n-1}$ are assumed fixed. Let us therefore deal with a single-dimensional equivalent of the equation of interest:
\begin{align}\label{eq:generic-fredholm-int-eqn}
u(x)
&=
1+\int{K}(x,y)\,u(y)\,dy,
\end{align}
where $\mathcal{K}(x,y)\ge0$, $\forall (x,y)\in\mathbb{R}^2$, is the respective transition probability density for the appropriate Markov chain.

It is known (and can be easily shown) that the solution to this equation admits the following Neumann series:
\begin{align*}
\begin{aligned}
u(z_0)
&=
1+\int_{[0,A]}{K}(z_0,z_1)\,dz_1+\int_{\otimes_{i=1}^2[0,A]}{K}(z_0,z_1)\,{K}(z_1,z_2)\,dz_1\,dz_2\\
&\qquad\qquad\qquad+\int_{\otimes_{i=1}^3[0,A]}{K}(z_0,z_1)\,{K}(z_1,z_2)\,{K}(z_2,z_3)\,dz_1\,dz_2\,dz_3+\ldots\\
&=
1+\sum_{k=1}^\infty\int_{\otimes_{i=1}^k[0,A]}\left(\prod_{j=1}^k{K}(z_{j-1},z_j)\right)dz_1\,dz_1\ldots dz_k,
\end{aligned}
\end{align*}
where $\otimes$ denotes the usual direct product (as applied to sets).

We first note that equation~\eqref{eq:generic-fredholm-int-eqn} can be solved either at a particular (single) point, or over a particular interval. We are interested in the former, with the point being zero, i.e., when the detection statistic has no headstart. For the actual solution method to obtain $u(0)\triangleq\EV[\T]$, one option would be to use a deterministic numerical scheme to ``linearize'' the integral in the right-hand side of~\eqref{eq:generic-fredholm-int-eqn}, and then, to ensure the linearization is ``optimal'', reduce the integral equation to a system of linear equations for a vector approximating the unknown function~\cite{Atkinson+Han:Book09}. The problem with this approach is that the integral is actually two-dimensional, and is over an unbounded region. As a result, it is difficult to ``chop up'' the region of integration to form a partition of reasonable size. To overcome this problem, we suggest to consider a Monte Carlo technique.

The idea of the basic Monte Carlo approach to evaluate $u(0)\triangleq\EV[\T]$ is to compute it {\em statistically}, i.e., to---in one way or another---estimate it based on a (somehow) generated sample, $\{\T_j\}_{1\le j\le N}$, of $N\ge1$ independent instantiations of the (same) stopping time $\T$. Specifically, let $\widehat{\EV[\T]}$ denote an estimator of $\EV[\T]$. The standard choice for $\widehat{\EV[\T]}$ is to use the sample mean
\begin{align*}
\bar{T}_N
&\triangleq
\dfrac{1}{N}\sum_{j=1}^N \T_j,
\end{align*}
which is well-justified since the sample mean is unbiased (for any $N\ge1$), and, due to the (strong) Law of Large Numbers, is also asymptotically (as $N\to\infty$) consistent.

While it is not a problem to simulate as many independent instantiations of $\T$ as
necessary (even if $N$ is $10^6$ or higher), the question of proximity of the respective estimate $\widehat{\EV[\T]}\triangleq \bar{T}_N$ to the actual true (but unknown) value of $\EV[\T]$ is to be addressed with care. To that end, the standard solution is to construct a $(1-\epsilon)\,\%$-confidence interval, $\epsilon\in(0,1)$, of a prescribed width, $w>0$. Specifically, let $\sigma_T$ denote the standard deviation of $\T$, i.e., $\sigma_{\T}\triangleq\sqrt{\Var[\T]}$. If $\sigma_T$ is known, then from the Central Limit Theorem (CLT) we immediately have
\begin{align*}
\sqrt{N}\,\dfrac{\bar{\T}_N-\EV[\T]}{\sigma_{\T}}
\underset{N\to\infty}{\overset{d}{\longrightarrow}}
\mathcal{N}(0,1),
\end{align*}
and therefore, to ensure that
\begin{align*}
\Pr\left(\dfrac{\vert\bar{\T}_N-\EV[\T]\vert}{\sigma_{\T}}\le w\right)
&\ge1-\epsilon,
\end{align*}
it suffices to take the sample size $N$ as
\begin{align*}
N
&\ge
\left\lfloor2z_{\epsilon/2}\frac{\sigma_{\T}}{w}\right\rfloor + 1,
\end{align*}
where $\lfloor x\rfloor$ is the floor function, and $z_{\epsilon/2}$ is the $\epsilon/2$-th percentile of the standard Gaussian distribution. Rephrasing this, for this choice of $N$, with probability of at least $1-\epsilon\in(0,1)$, the unknown true mean $\EV[\T]$ of the stopping time $\T$ will be contained in the interval $(\bar{\T}_N-w,\bar{\T}_N+w)$, i.e., it will be true that $\Pr\left(\bar{\T}_N-w\le \EV[\T]\le \bar{\T}_N+w\right)\ge 1-\epsilon$.

As may be seen, the problem with this approach is that the standard deviation, $\sigma_{\T}$, is {\em not} known. One could, of course, estimate it, and then build a confidence interval off the {\em estimated} value, $\hat{\sigma}_{\T}$. The issue with this idea, however, is that the distribution of $(\bar{T}_N-\EV[\T])/\hat{\sigma}_{\T}$ may fail to be normal, even asymptotically, as $N\to\infty$. A more elegant way out of this is to note that if the detection statistic's headstart is zero, which
is the case in this work, then $\sigma_{\T}\le\EV[\T]$. With formal proof of this inequality temporarily deferred, let us illustrate what it can lead to. Since $\sigma_{\T}\le\EV[\T]$ and $\EV[\T]>0$ (in fact $\EV[\T]\ge1$), the event $\big\{\vert\bar{\T}_N-\EV[\T]\vert\le w\EV[\T]\big\}$ is contained in the event $\big\{\vert\bar{\T}_N-\EV[\T]\vert\le w\sigma_{\T}\big\}$, and therefore
\begin{align*}
\Pr\left(\dfrac{\vert\bar{\T}_N-\EV[\T]\vert}{\EV[\T]}<w\right)
&\ge
\Pr\left(\dfrac{\vert\bar{\T}_N-\EV[\T]\vert}{\sigma_{\T}}<w\right),
\end{align*}
so that confidence bounds for the {\em relative error}, $\vert\bar{\T}_N-\EV[\T]\vert/\EV[\T]$, are readily available; in this case $w$ is measured in percentages. This criterion is known as {\em prescribed proportional closeness}~\cite[p.~339]{Ehrenfeld+Littauer:Book1964},~\cite{Simons+Zacks:Stanford1967,Nadas:AMS1969,Willson+Folks:CommStat-SA1983,Zacks:AMS1966,Zacks:HofS-AinR2001}.

We now prove the claim made earlier that $\sigma_{\T}\le\EV[\T]$ when the detection statistic has no headstart. Let $\mathfrak{M}_k(x,A)\triangleq\EV[\T^k(x,A)]$ with $k\in\mathbb{N}$, i.e., $\mathfrak{M}_k(x,A)$ is the $k$-th moment of the stopping time $\T(x,A)$; in particular, $\mathfrak{M}_1(x,A)\triangleq\EV[\T(x,A)]$. We will need the following two lemmas.
\begin{lemma}\label{lem:ineq1}
$\mathfrak{M}_1(x,A)\le\mathfrak{M}_1(0,A)$ for any given $A>0$ and $x\ge0$.
\end{lemma}
\begin{proof}
The sought-after inequality $\mathfrak{M}_1(x,A)\le\mathfrak{M}_1(0,A)$ is merely the statement that, on average, the higher the headstart of the detection statistic behind the stopping time $\T(x,A)$, the sooner (on average) the respective detection procedure is to terminate. That is, the closer the detection statistic is initially to the detection threshold, $A>0$, the sooner (on average) it is to reach the threshold, and therefore, the sooner (on average) the detection procedure is to stop.
\end{proof}
\begin{lemma}\label{lem:ineq2}
$\big[(\mathbb{I}-\mathcal{K})^{-1}\circ\mathfrak{M}_1\big](x,A)\le\mathfrak{M}_1(0,A)\,\mathfrak{M}_1(x,A)$ for any given $A>0$ and $x\ge0$.
\end{lemma}
\begin{proof}
The desired result is a direct consequence of Lemma~\ref{lem:ineq1}, i.e., the inequality $\mathfrak{M}_1(x,A)\le\mathfrak{M}_1(0,A)$, $\forall x\ge0$, applied to upperbound each summand in the respective Neumann expansion for $\big[(\mathbb{I}-\mathcal{K})^{-1}\circ\mathfrak{M}_1\big](x,A)$.
\end{proof}

With Lemma~\ref{lem:ineq1} and Lemma~\ref{lem:ineq2} in mind, it is easy to see that $\big[(\mathbb{I}-\mathcal{K})^{-1}\circ\mathfrak{M}_1\big](x,A)\le\mathfrak{M}_1^2(0,A)$. As shown in~\cite{Polunchenko+etal:SA2014}, the second moment, $\mathfrak{M}_2(x,A)\triangleq\EV[\T^2(x,A)]$, of the generic stopping time, $\T(x,A)$, is governed by the integral equation $\mathfrak{M}_2=2\mathfrak{M}_1-1+\mathcal{K}\circ\mathfrak{M}_2$, where $\mathfrak{M}_1(x,A)\triangleq\EV[\T(x,A)]$. The operator solution to this equation is of the form $\mathfrak{M}_2=2(\mathbb{I}-\mathcal{K})^{-1}\circ\mathfrak{M}_1-\mathfrak{M}_1$. Finally, since $\Var[\T]=\mathfrak{M}_2-\mathfrak{M}_1^2$, and $\mathcal{K}\circ\mathfrak{M}_1=\mathfrak{M}_1-1\le\mathfrak{M}_1$, we obtain
\begin{align*}
\mathfrak{M}_2(x,A)
&=
2\big[(\mathbb{I}-\mathcal{K})^{-1}\circ\mathfrak{M}_1\big](x,A)-\mathfrak{M}_1(x,A)\\
&\le
2\mathfrak{M}_1^2(0,A)-\mathfrak{M}_1(x,A)\\
&\le
2\mathfrak{M}_1^2(0,A).
\end{align*}
Therefore, $\mathfrak{M}_2(0,A)\le 2\mathfrak{M}_1^2(0,A)$, and we have
\begin{align*}
\sigma_{\T}^2\triangleq\Var[\T(0,A)]
&\triangleq
\mathfrak{M}_2(0,A)-\mathfrak{M}_1^2(0,A)\\
&\le
\mathfrak{M}_1^2(0,A),
\end{align*}
which is to say that $\sigma_{\T}\le\EV[\T]$, provided the detection statistic starts off zero (no headstart). We stress that the assumption of no headstart is a critical one, and when the detection statistic does have a (non-zero) headstart, the inequality $\sigma_{\T}\le\EV[\T]$ may fail to hold.

We conclude this subsection with a remark on how to reduce the variance of the sample mean $\bar{\T}_N$. The idea is that since the first two terms in the Neumann series are always computable exactly, as such there is no need to estimate either one of them. Hence, instead of sampling the trajectories from $\mathcal{K}(x,y)$ assuming the starting point is 0, one may start each trajectory off a random point, sampled from $\mathcal{K}(x,y)$, but restricted to the interval $[0,A]$.

\begin{figure*}[htb!]
\begin{center}
\begin{tabular}{cc}
\includegraphics[height=2.1in,width=2.8in] {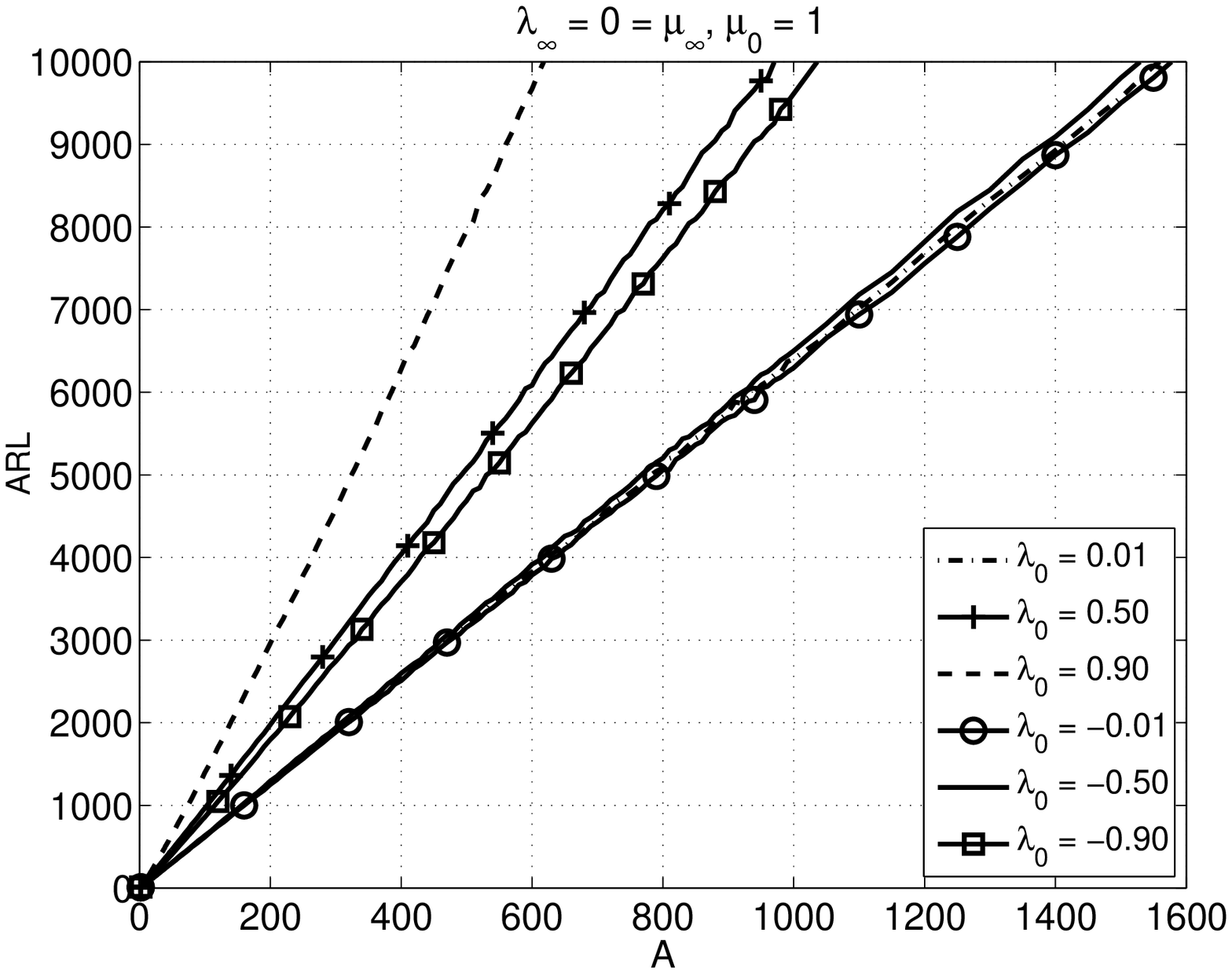}
&
\includegraphics[height=2.1in,width=2.8in] {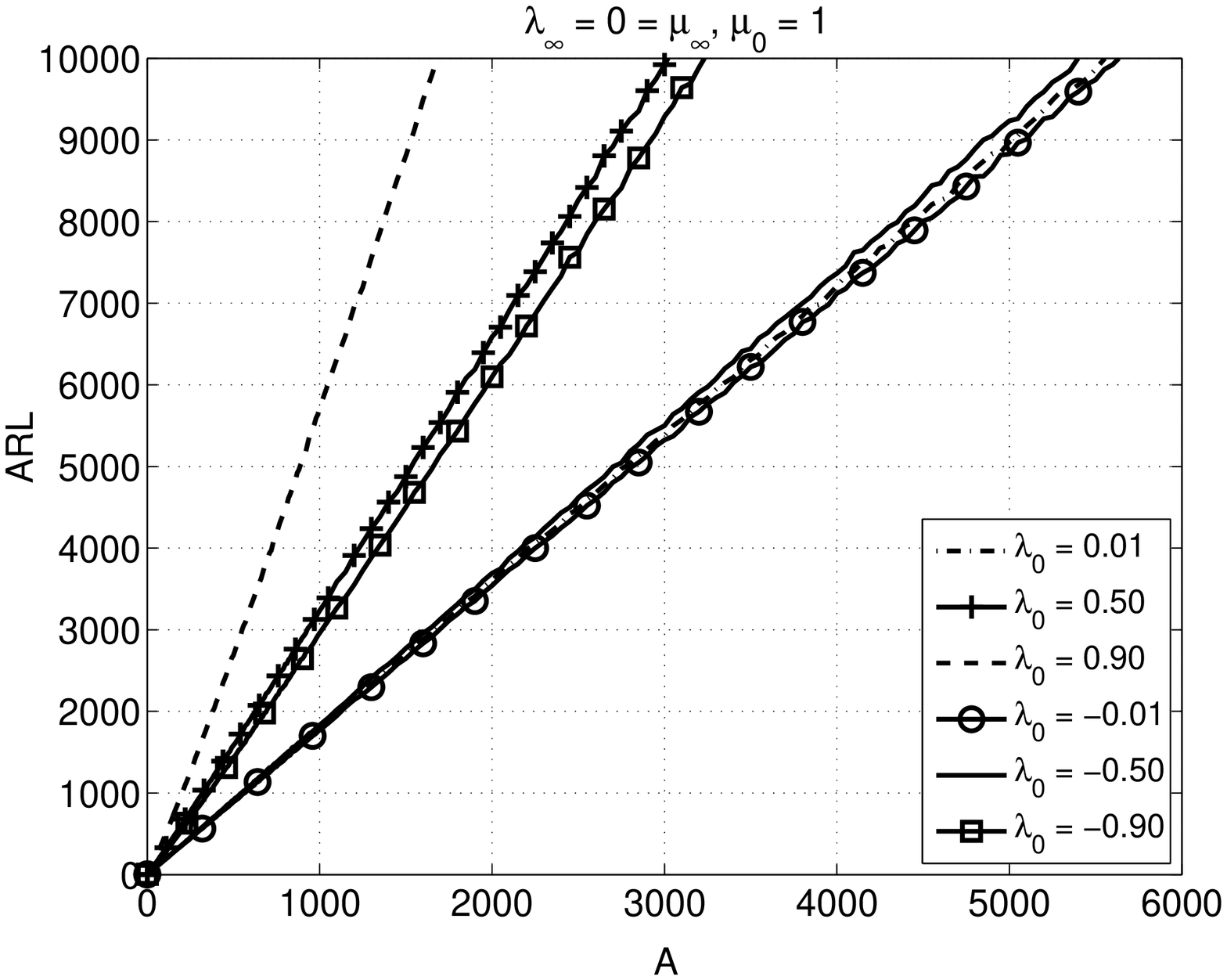}
\\ {\hspace{0.05in}} (a) & (b)
\\
\includegraphics[height=2.1in,width=2.8in] {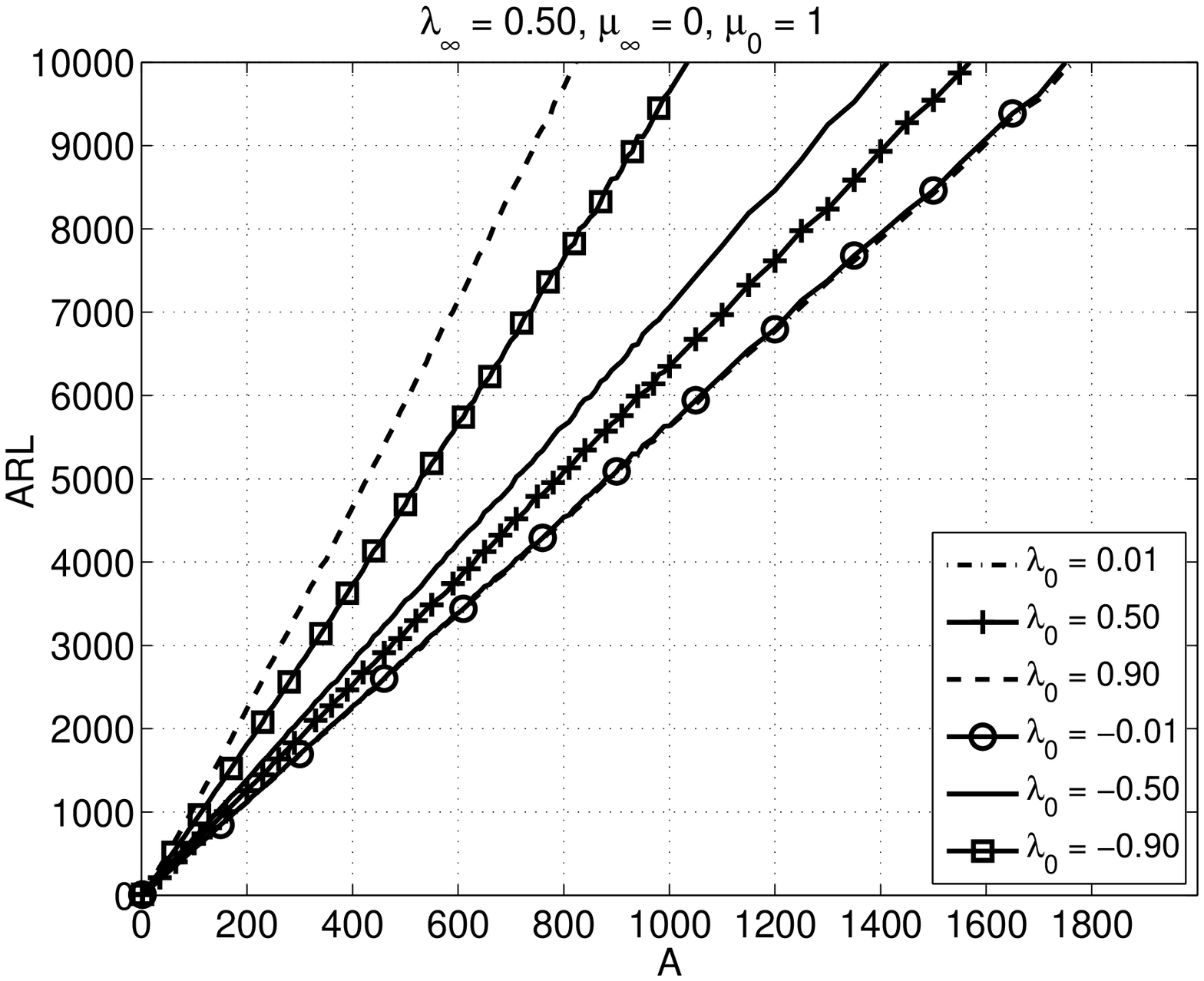}
&
\includegraphics[height=2.1in,width=2.8in] {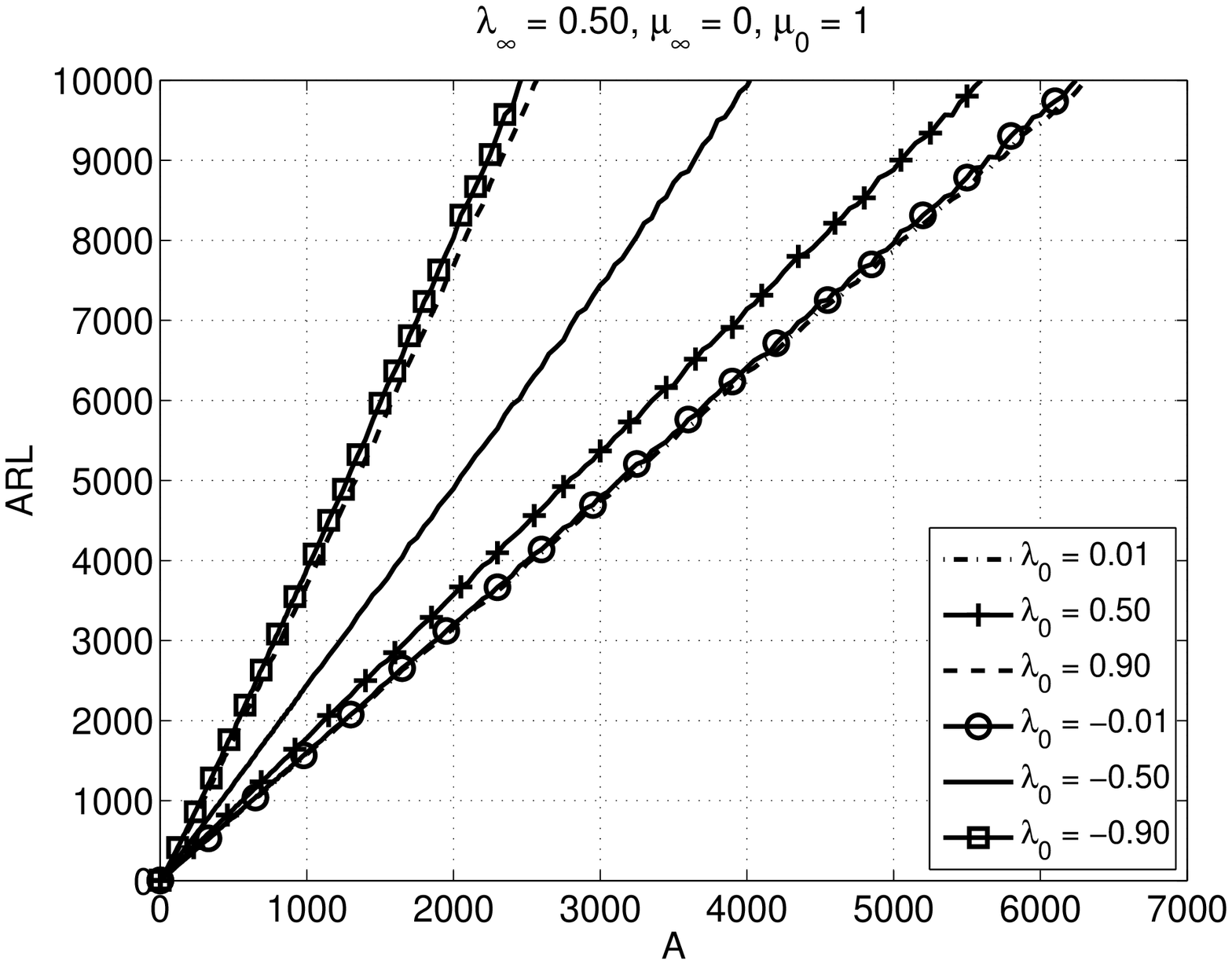}
\\ {\hspace{0.05in}} (c) & (d)
\\
\includegraphics[height=2.1in,width=2.8in] {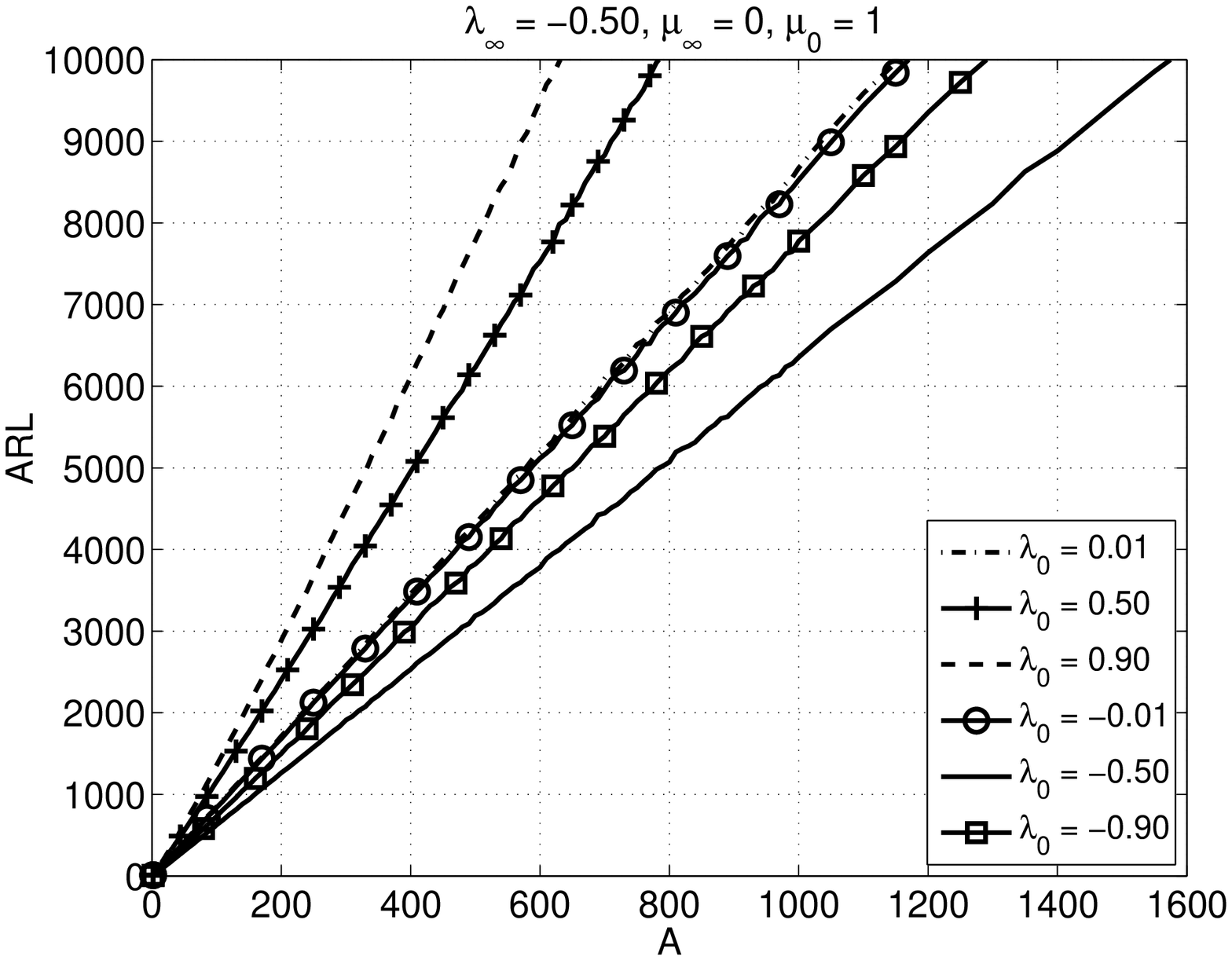}
&
\includegraphics[height=2.1in,width=2.8in] {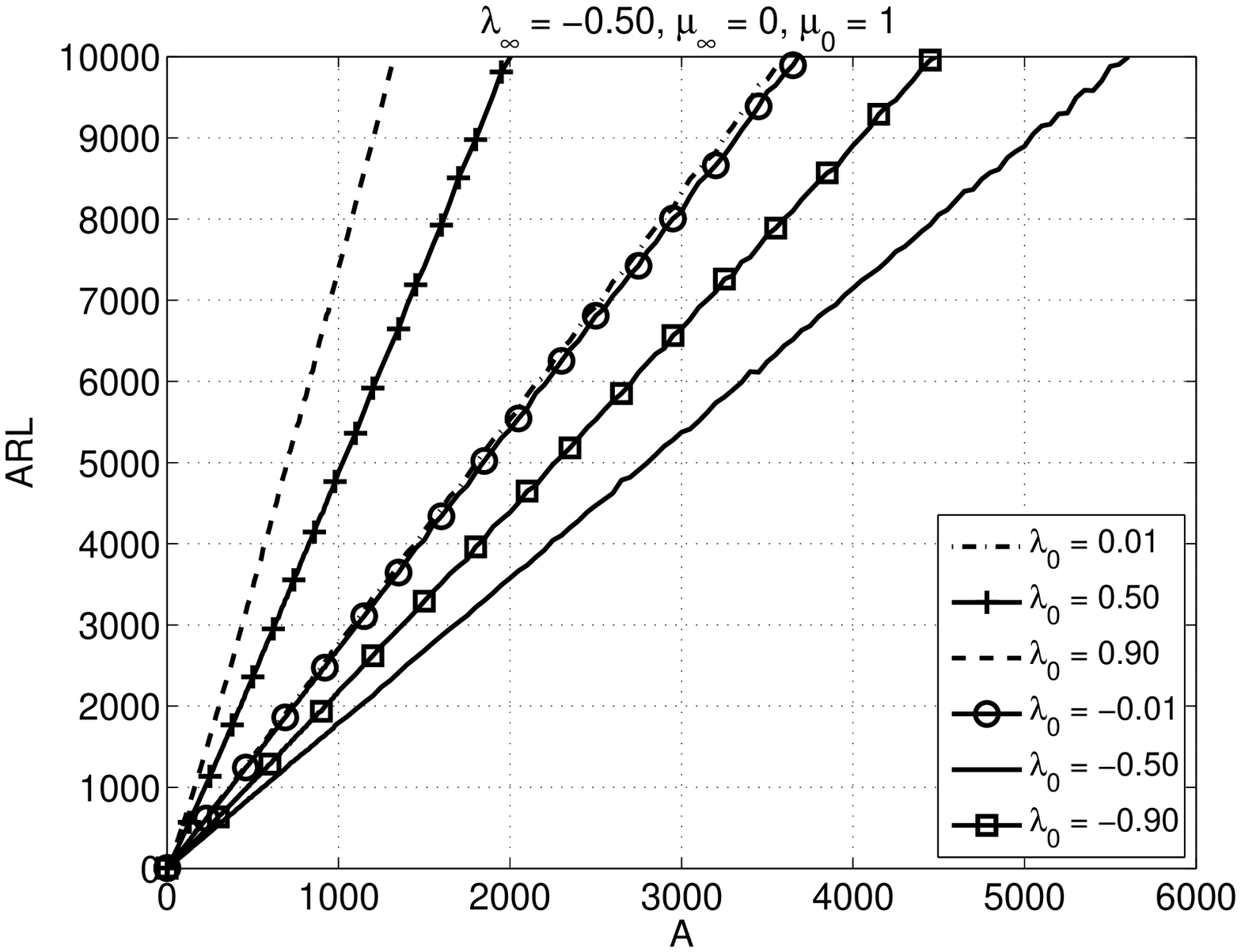}
\\ {\hspace{0.05in}} (e) & (f)
\end{tabular}i
\caption{\label{fig_threshold}
$\ARL$ as a function of threshold $A$ for the (a) CUSUM chart and (b) SR
procedure in the i.i.d.\ pre-change setting with $\mu_{\infty} = 0$.
$\ARL$ vs.\ $A$ for the CUSUM chart and SR procedure in the case where
(c)-(d) $\lambda_{\infty} = 0.50$ and (e)-(f) $\lambda_{\infty} = -0.50$
for different $\lambda_0$ values.}
\end{center}
\vspace{-5mm}
\end{figure*}

\begin{figure*}[htb!]
\begin{center}
\begin{tabular}{cc}
\includegraphics[height=2.3in,width=2.8in] {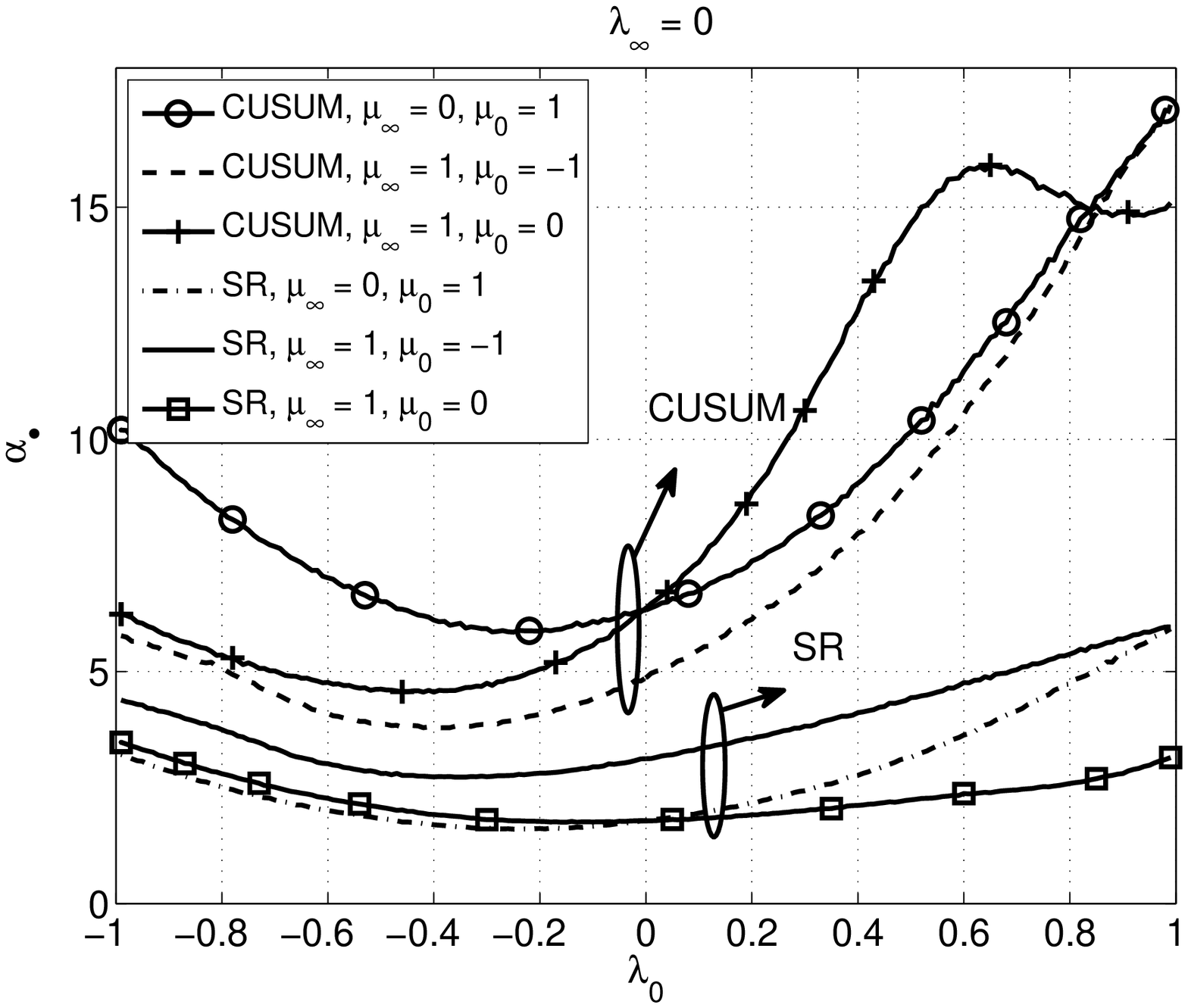}
&
\includegraphics[height=2.3in,width=2.8in] {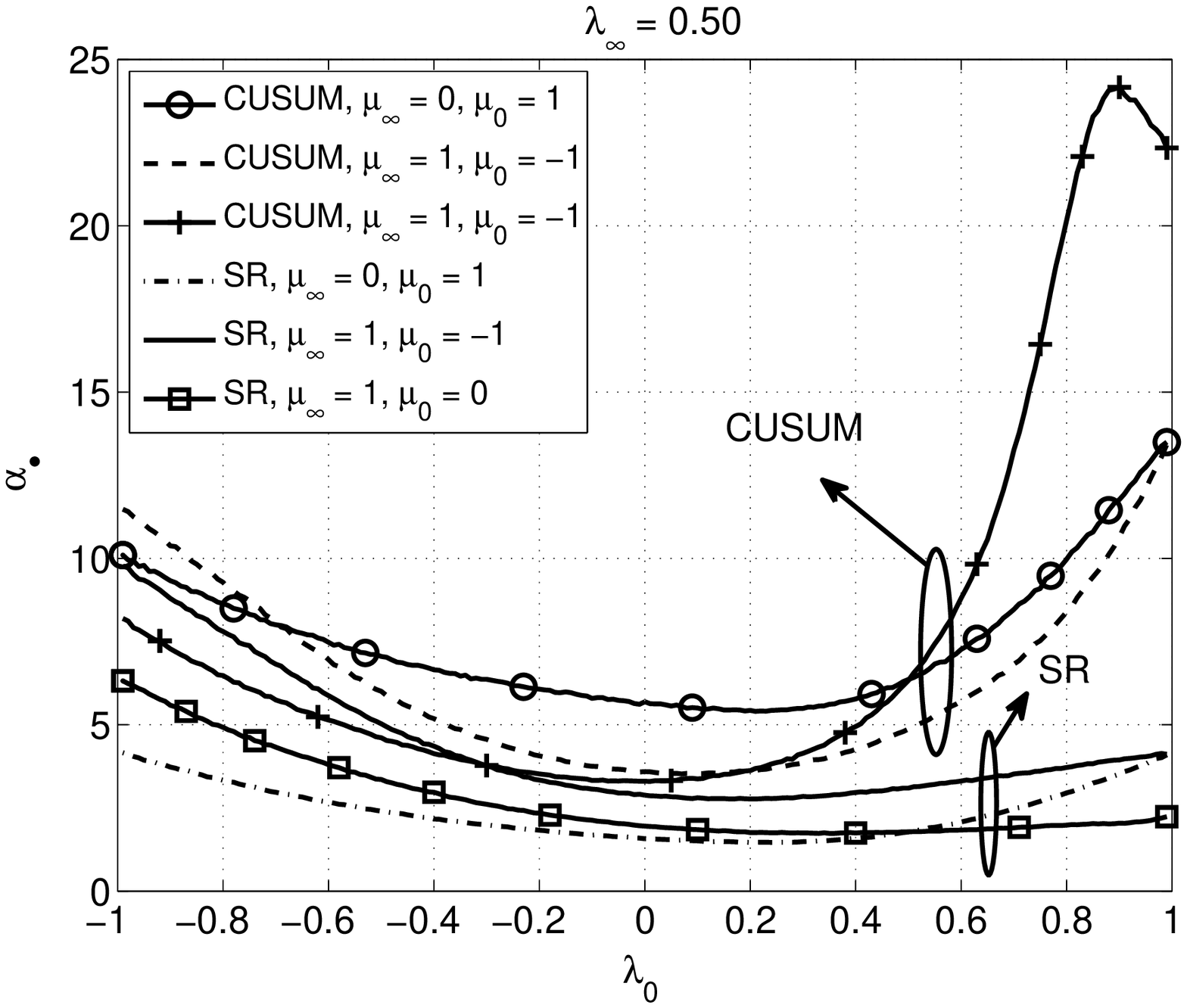}
\\
{\hspace{0.05in}} (a) &  (b) \\
\includegraphics[height=2.3in,width=2.8in] {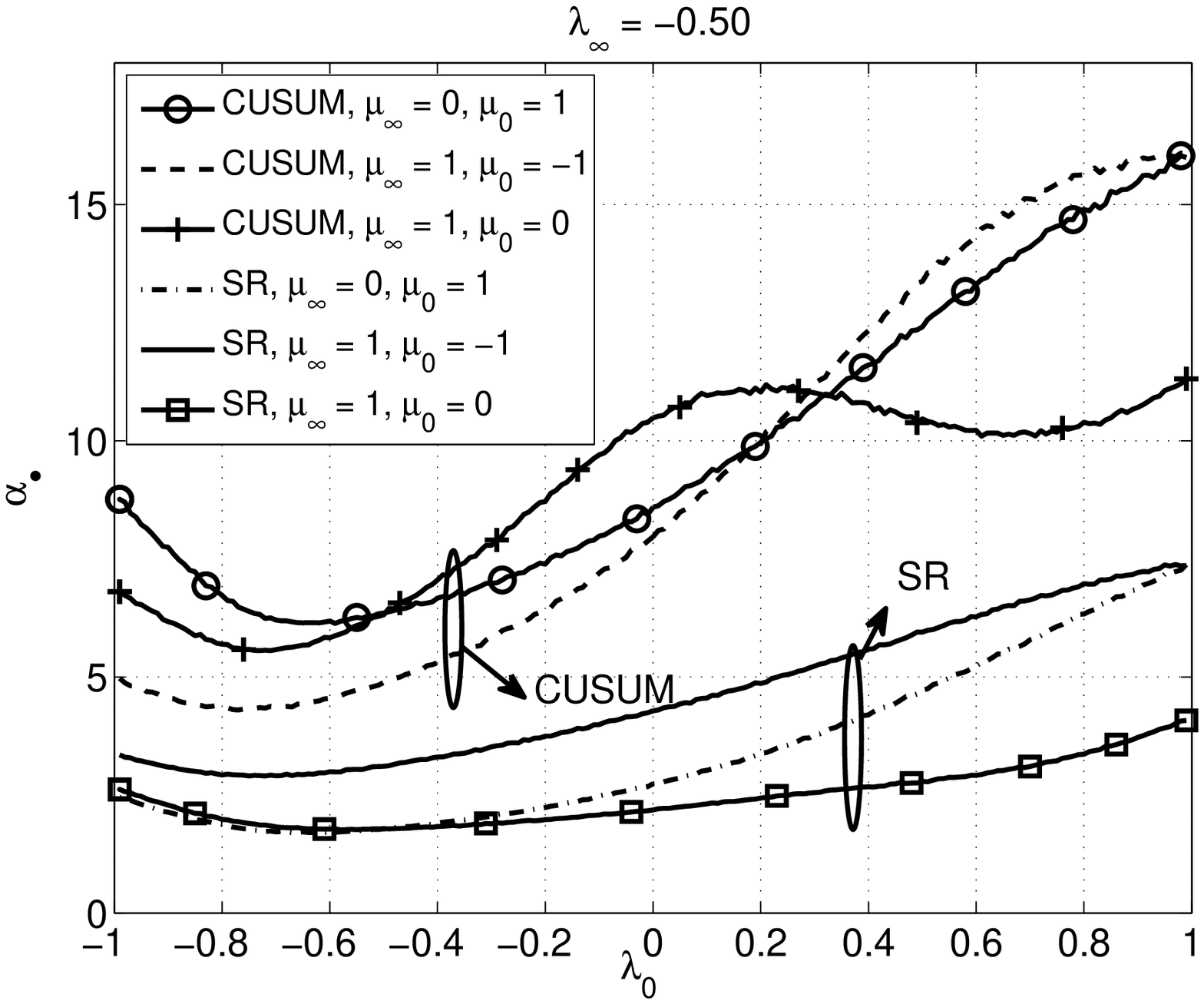}
&
\\ {\hspace{0.05in}} (c) &
\end{tabular}
\caption{\label{fig_partial} Estimates of $\alpha_{\sf cs}$ and $\alpha_{\sf sr}$ as a
function of $\lambda_0$ under three different settings for pre-change and post-change
means.}
\end{center}
\vspace{-5mm}
\end{figure*}

\begin{figure*}[htb!]
\begin{center}
\begin{tabular}{cc}
\includegraphics[height=2.1in,width=2.8in] {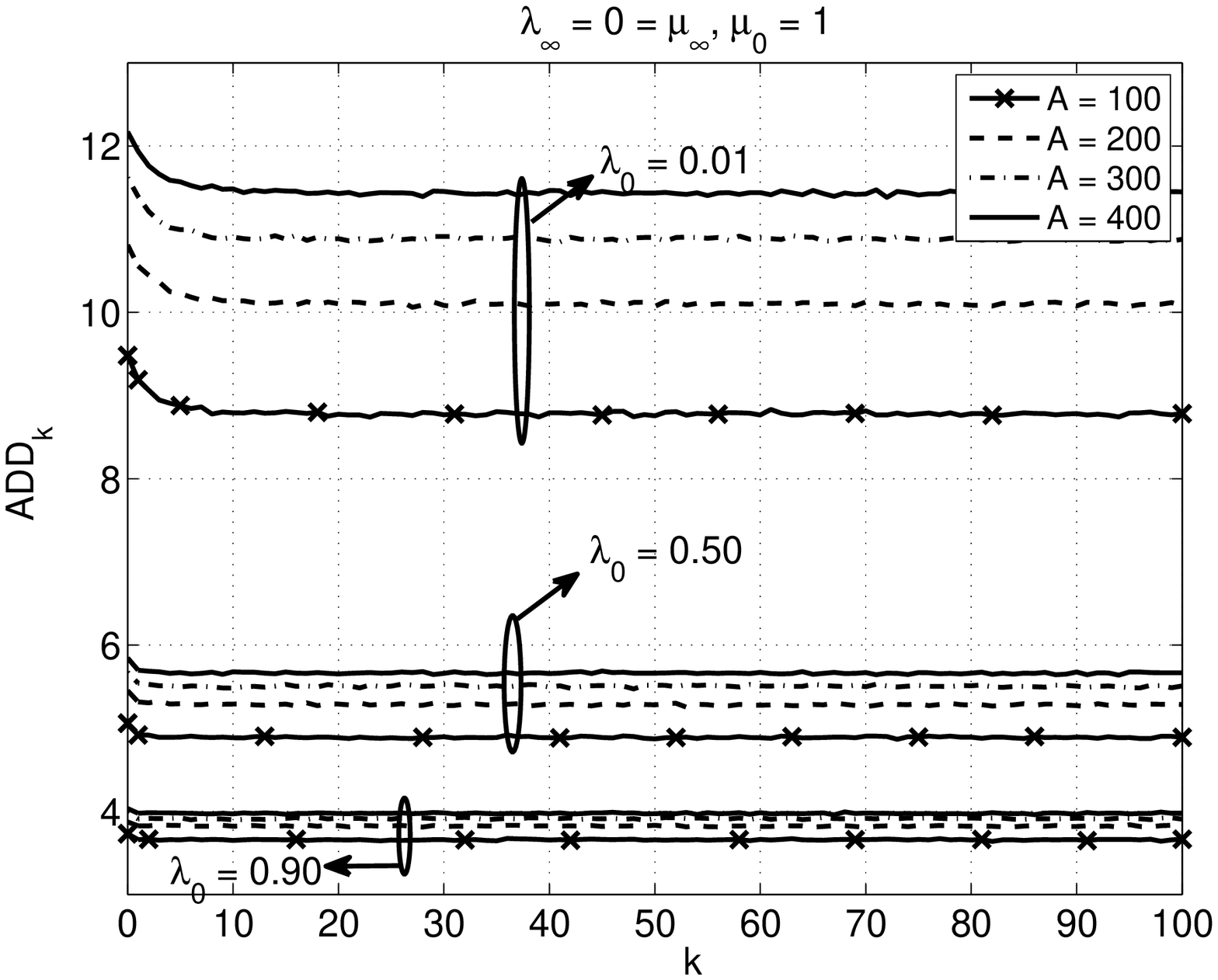}
&
\includegraphics[height=2.1in,width=2.8in] {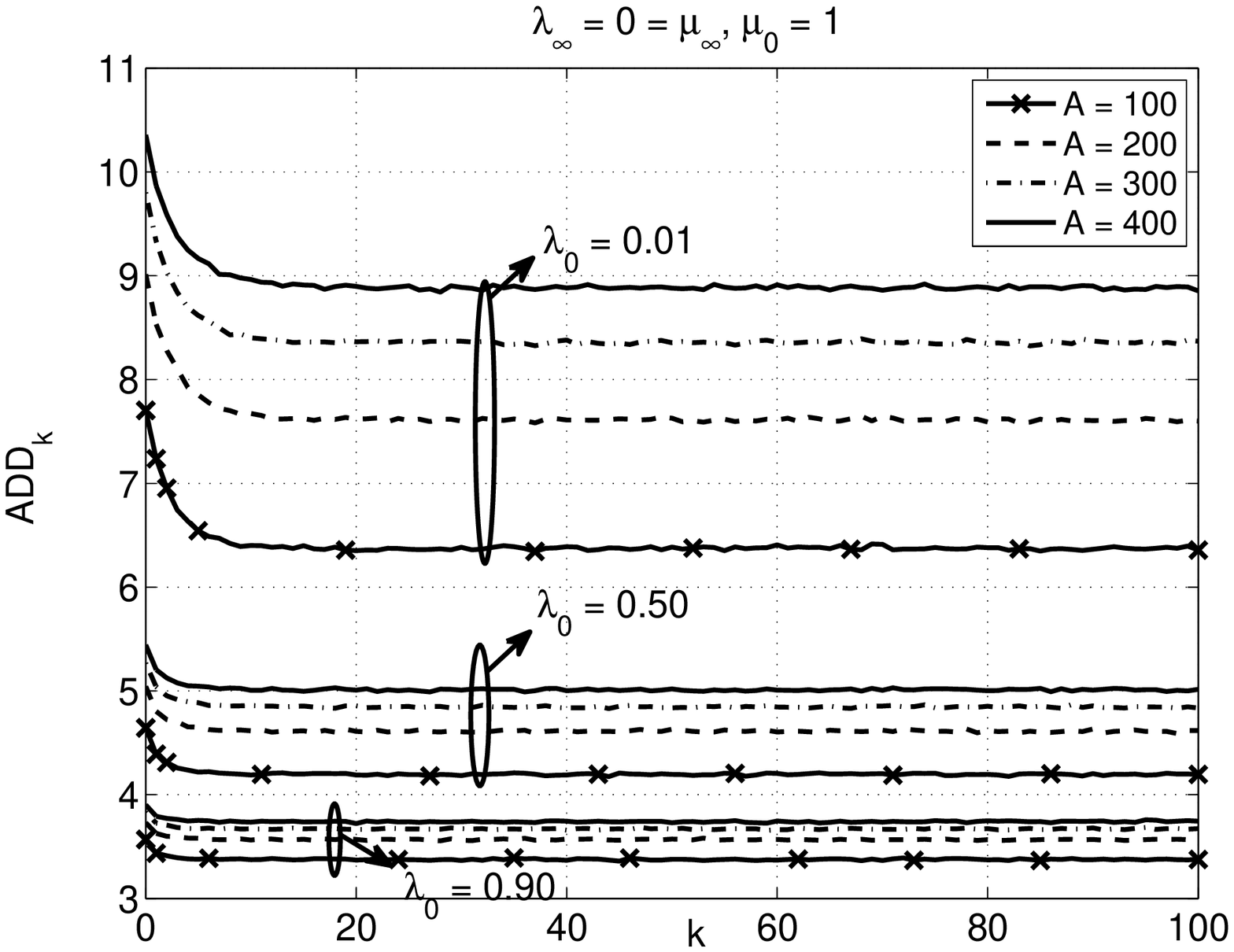}
\\ {\hspace{0.05in}} (a) & (b) \\
\includegraphics[height=2.1in,width=2.8in] {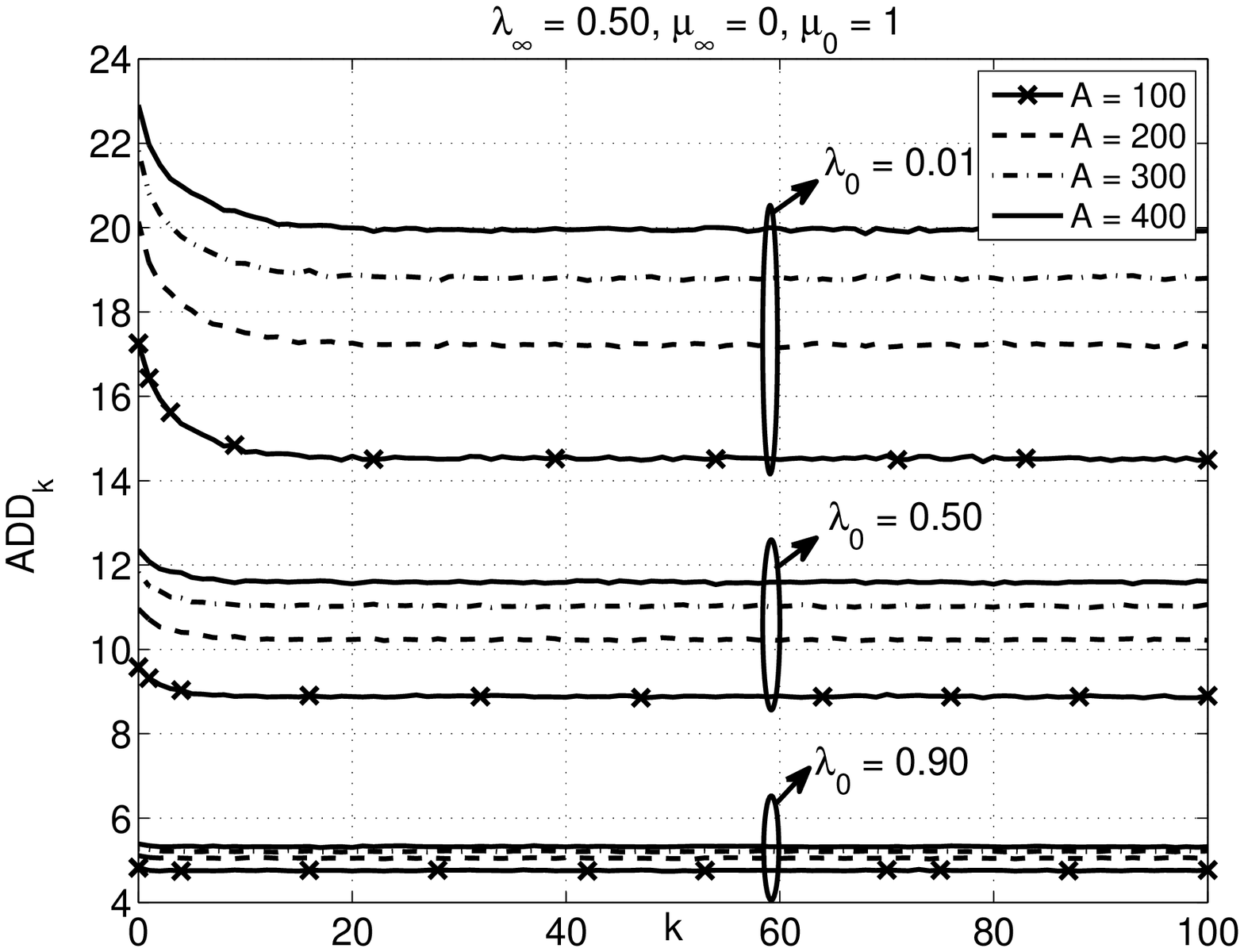}
&
\includegraphics[height=2.1in,width=2.8in] {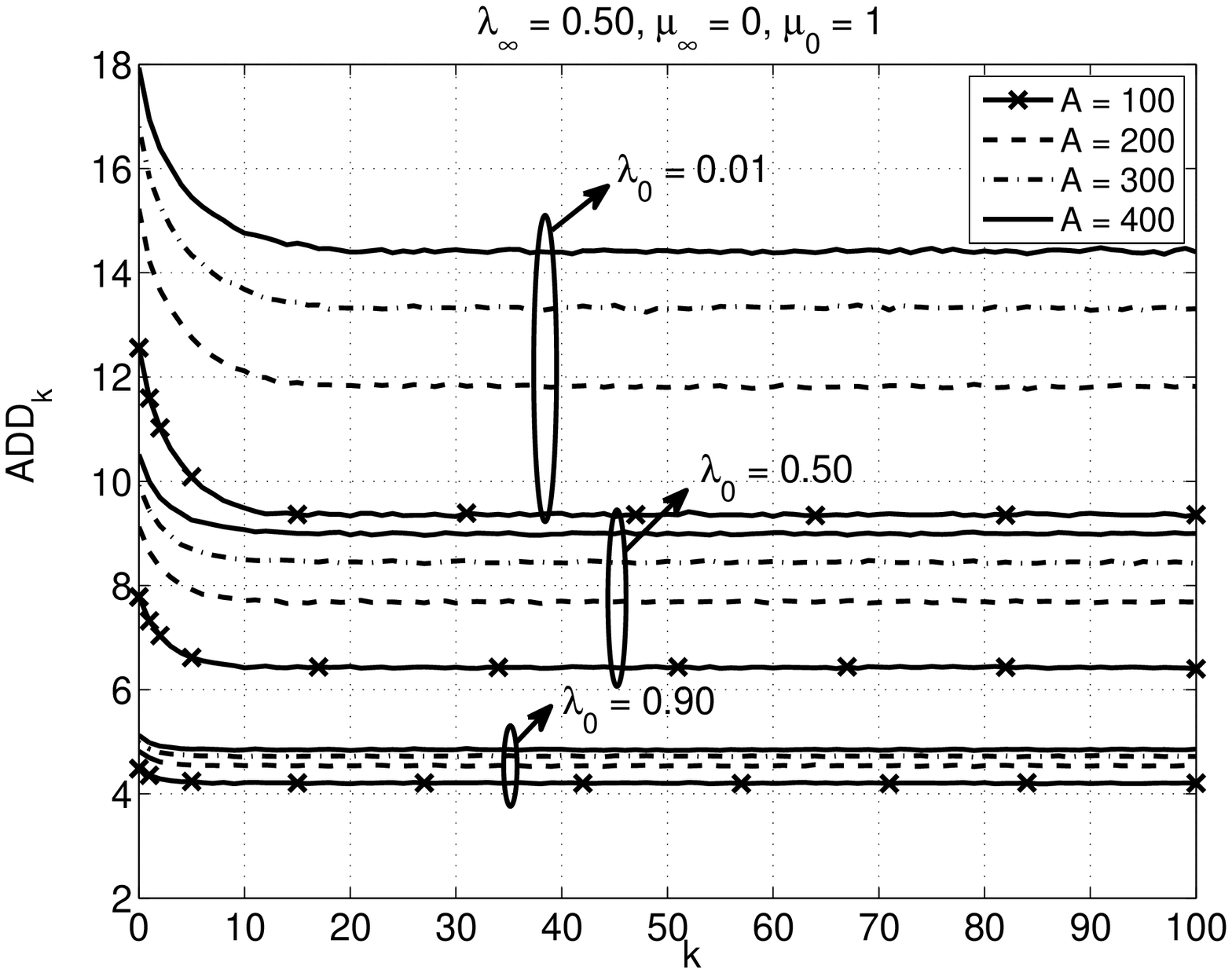}
\\ {\hspace{0.05in}} (c) & (d) \\
\includegraphics[height=2.1in,width=2.8in] {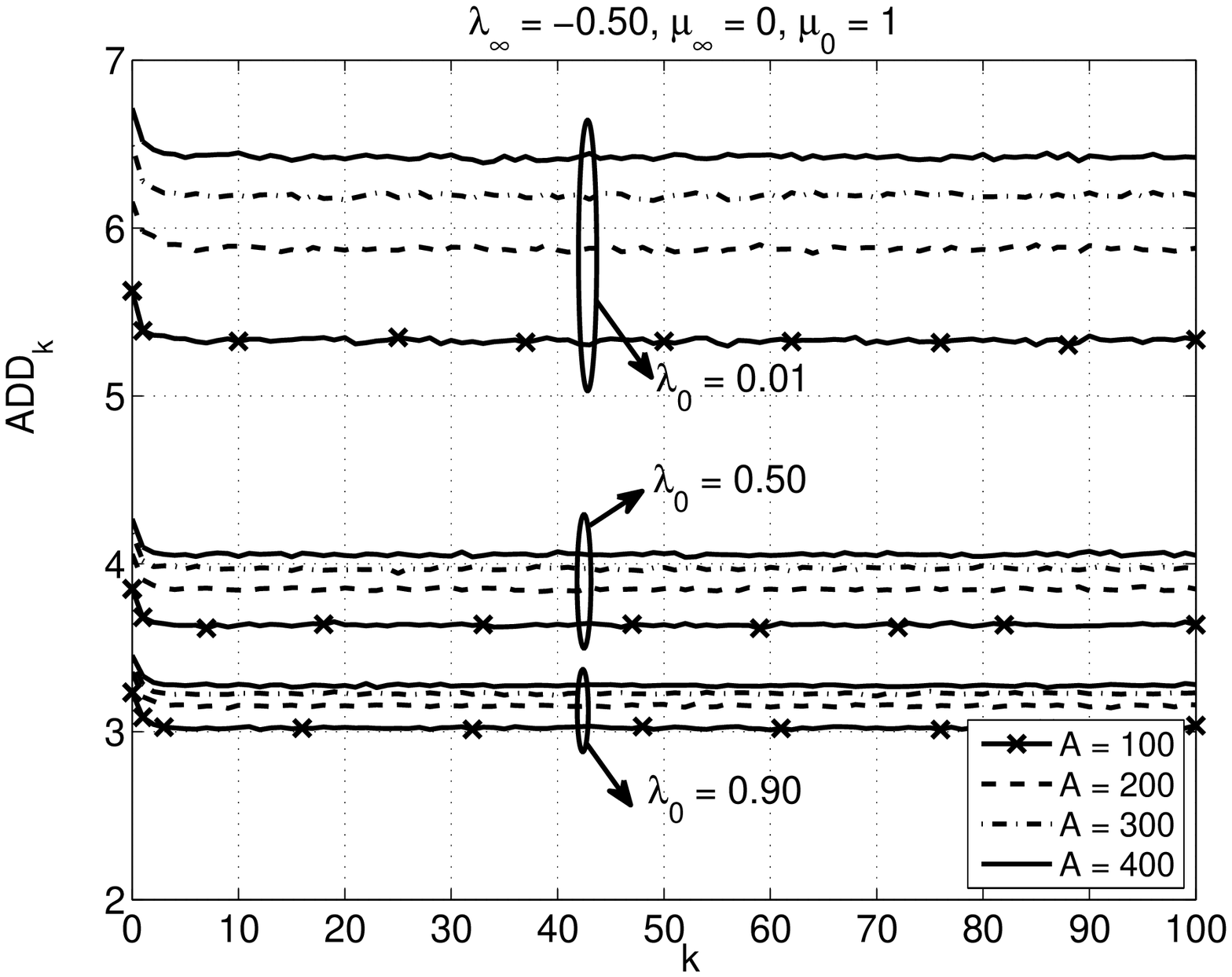}
&
\includegraphics[height=2.1in,width=2.8in] {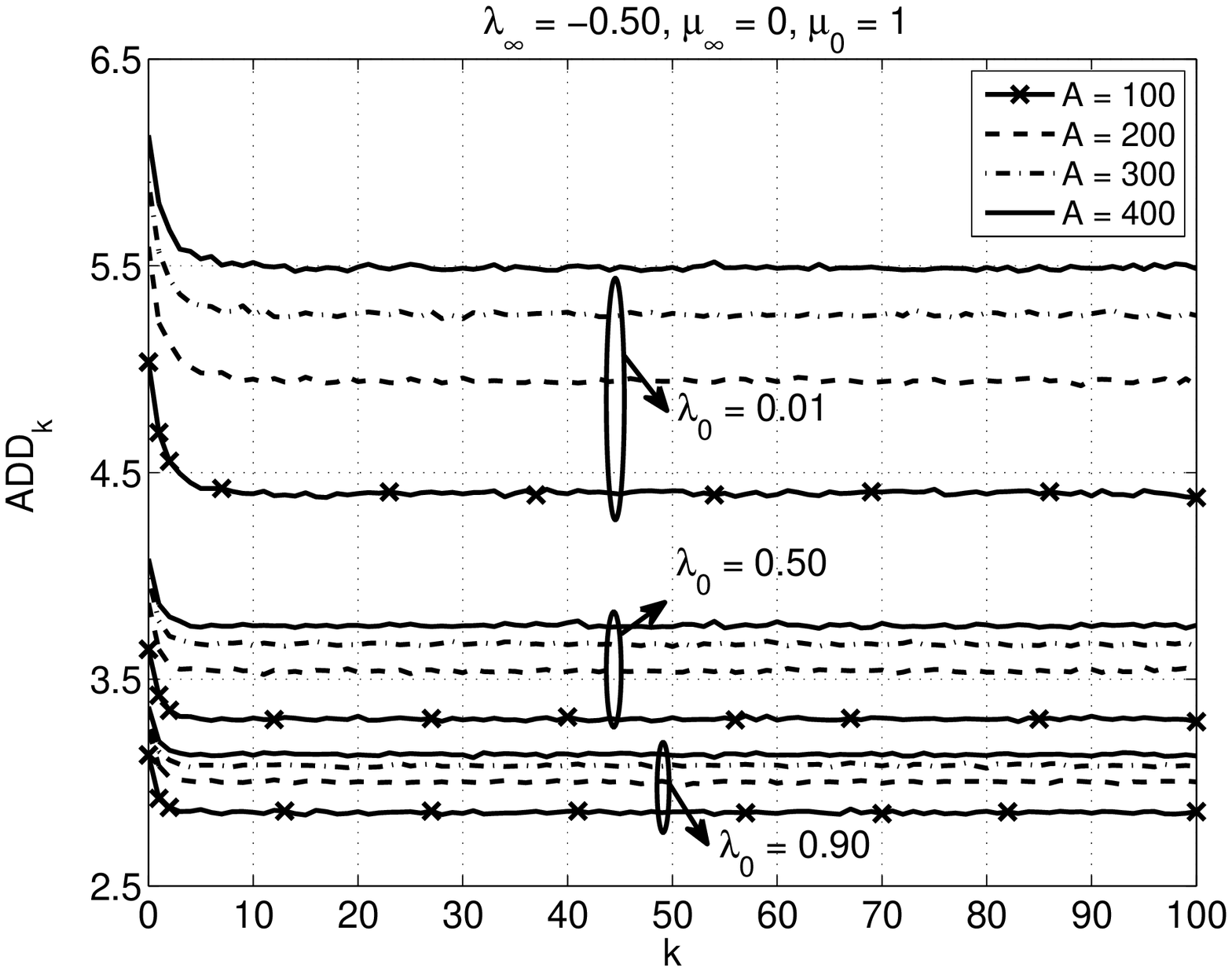}
\\ {\hspace{0.05in}} (e) & (f)
\end{tabular}
\caption{\label{fig_addk}
$\ADD_k$ as a function of $k$ in the i.i.d.\ pre-change setting with
$\mu_{\infty} = 0$ for the (a) CUSUM chart and (b) SR procedure.
$\ADD_k$ vs.\ $k$ for the CUSUM chart and SR procedure in the case
where (c)-(d) $\lambda_{\infty} = 0.50$ and (e)-(f) $\lambda_{\infty} =
-0.50$.}
\end{center}
\vspace{-5mm}
\end{figure*}

\ignore{
How to set thresholds for CUSUM and SR Procedures,
${\rm ADD}_k$ as a Function of $k$

Is the SR procedure inferior to the CUSUM chart?
}

\section{Numerical Studies}
\label{sec:Num}

We now apply the numerical techniques illustrated in the preceding section
to compute ARL and SADD and perform a comparative analysis of the CUSUM chart
and SR procedure. But before this, we need to design the procedures carefully.

\subsection{Designing the CUSUM chart and SR procedure}
The design of the CUSUM chart and the SR procedure requires an understanding
of how the threshold $A$ should be set (as a function of $\gamma$) to ensure
that the $\ARL$ with either procedure is at least $\gamma$. For this, we study
the $\ARL$ behavior of both procedures numerically as a function of the
threshold $A$. In the i.i.d.\ pre-change setting with $\mu_{\infty} = 0$ and
$\mu_0 = 1$, Figs.~\ref{fig_threshold}(a) and~(b) plot $\ARL$ for the CUSUM
chart and the SR procedure, respectively, for different values of $A$ and
$\lambda_0$. In Figs.~\ref{fig_threshold}(c)-(d), $\ARL$ vs.\ $A$ is plotted for
the CUSUM chart and SR procedure
with $\lambda_{\infty} = 0.50$ and several $\lambda_0$ values.
Similarly, in Figs.~\ref{fig_threshold}(e)-(f), $\ARL$ vs.\ $A$ is plotted
for the two procedures
with $\lambda_{\infty} = -0.50$ and several $\lambda_0$ values. The robust
linear dependence in our studies (across different parameter values) suggests
the following empirical relationship (as $A \rightarrow \infty$):
\begin{eqnarray}
\ARL(\TCS) = \alpha_{ \sf cs} \cdot A  +
\beta_{\sf cs},
& &
\ARL(\TSR) = \alpha_{ \sf sr} \cdot A +
\beta_{\sf sr},
\label{arl_a}
\end{eqnarray}
for some constants
$\{ \alpha_{\bullet} \}$ and $\{ \beta_{\bullet} \}$ depending only on the
model parameters.
To further understand the behavior of $\alpha_{\sf cs}$ and $\alpha_{\sf sr}$
as a function of the model parameters, in Figs.~\ref{fig_partial}(a)-(c), we plot
estimates of these quantities as a function of $\lambda_0$ in three
settings where $\lambda_{\infty} = 0$, $\lambda_{\infty} = 0.50$ and
$\lambda_{\infty} = -0.50$, each with: i) $\mu_{\infty} = 0$ and $\mu_0 = 1$,
ii) $\mu_{\infty} = 1$ and $\mu_0 = -1$, and iii) $\mu_{\infty} = 1$ and
$\mu_0 = 0$. 

Noting that $\SADD$ is the supremum of the conditional average detection
delay ($\ADD_k$), we are interested in the choice of $k$ that maximizes
$\ADD_k$. In the $\lambda_{\infty} = 0 = \lambda_0$ case, it is well-understood
that this maximum occurs at $k = 0$. However, generalizing this result to the
AR setting seems difficult. Thus, we pursue a numerical approach in understanding
this problem. In Figs.~\ref{fig_addk}(a) and~(b), we plot the behavior of
$\ADD_k$ as a function of $k$ for the CUSUM chart and the SR procedure, respectively,
for different $A$ and $\lambda_0$ values. The same trend is plotted in
Figs.~\ref{fig_addk}(c)-(d) and~(e)-(f) for the CUSUM chart and SR procedure in the
$\lambda_{\infty} = 0.50$ and $\lambda_{\infty} = -0.50$ settings, respectively.
In all the cases considered, the maximum of $\ADD_k$ occurs at $k = 0$ thus
suggesting that $\SADD = \ADD_0$ in the AR setting also. This fact is critical
since Sec.~\ref{sec:performance-evaluation} allows us to compute $\ADD_0(T) = \EV_0(T)$.

Further, from these studies, we also observe that an increase in $A$ leads to
an increased $\ADD_k$ for all $k$, and the same threshold results in a higher
$\ADD_k$ for the CUSUM chart relative to the SR procedure --- both of which
are not surprising conclusions. Also, note that for both procedures, $\ADD_k$
converges to a steady-state value ($\ADD_{\infty}$) quickly. $\ADD_{\infty}$
can be treated as the average delay in detecting a change upon repeated trials
of the monitoring process.
It can also be seen from Table~\ref{table_add0_addinf} that the SR procedure
is more sensitive to the  change-point than the CUSUM chart as captured by a
larger value for the metric $\ADD_0 - \ADD_{\infty}$. Further, note that
$\ADD_0 - \ADD_{\infty}$ decreases as $A$ increases confirming the intuition
that in the large $A$ regime (and thus large $\ARL$ regime from~(\ref{arl_a})),
$\ADD_k$ is essentially independent of $k$.

\begin{table*}[!]
\caption{$\ADD_0$ and $\ADD_{\infty}$ for different choices of $A$,
$\lambda_{\infty}$ and $\lambda_0$ with the CUSUM chart and the SR
procedure.}
\label{table_add0_addinf}
\begin{center}
    \scalebox{0.8}{
    \begin{tabular}{|l||c|c|c|c||c|c|c|}
\hline
&  & \multicolumn{3}{|c||}{CUSUM} &
 \multicolumn{3}{|c|}{SR} \\
\hline
& $A$ & $\ADD_0$ & $\ADD_{\infty}$ & $\ADD_0 - \ADD_{\infty}$
& $\ADD_0$ & $\ADD_{\infty}$ & $\ADD_0 - \ADD_{\infty}$ \\
\hline  \hline
& &
\multicolumn{6}{|c|}{$\lambda_{\infty} = 0$} \\ \hline
\multirow{3}[8]{0.5cm}{\rotatebox{90}{$\lambda_0=0.01$}} &
$100$ & $9.4794$ & $8.7804$ & $0.6990$ & $7.7031$ & $6.3756$ & $1.3275$ \\
& $200$ & $10.8089$ & $10.1006$ & $0.7083$ & $9.0141$ & $7.6100$ & $1.4041$ \\
& $300$ & $11.6191$ & $10.8822$ & $0.7369$ & $9.8014$ & $8.3579$ & $1.4435$ \\
& $400$ & $12.1713$ & $11.4335$ & $0.7378$ & $10.3568$ & $8.8858$ & $1.4710$ \\
\hline
\multirow{3}[8]{0.5cm}{\rotatebox{90}{$\lambda_0=0.50$}} &
$100$ & $5.0596$ & $4.8889$ & $0.1707$ & $4.6441$ & $4.1965$ & $0.4476$ \\
& $200$ & $5.4508$ & $5.2807$ & $0.1701$ & $5.0438$ & $4.6122$ & $0.4316$ \\
& $300$ & $5.6700$ & $5.5059$ & $0.1641$ & $5.2674$ & $4.8460$ & $0.4214$ \\
& $400$ & $5.8421$ & $5.6598$ & $0.1823$ & $5.4434$ & $5.0101$ & $0.4333$ \\
\hline
\multirow{3}[8]{0.5cm}{\rotatebox{90}{$\lambda_0=0.90$}} &
$100$ & $3.7302$ & $3.6574$ & $0.0728$ & $3.5688$ & $3.3758$ & $0.1930$ \\
& $200$ & $3.8730$ & $3.8218$ & $0.0512$ & $3.7294$ & $3.5667$ & $0.1627$ \\
& $300$ & $3.9778$ & $3.9136$ & $0.0642$ & $3.8311$ & $3.6697$ & $0.1614$ \\
& $400$ & $4.0328$ & $3.9745$ & $0.0583$ & $3.9018$ & $3.7407$ & $0.1611$ \\
\hline
\hline
& &
\multicolumn{6}{|c|}{$\lambda_{\infty} = -0.50$} \\ \hline
\multirow{3}[8]{0.5cm}{\rotatebox{90}{$\lambda_0=0.01$}}
& $100$ & $5.6261$ & $5.3293$ & $0.2968$ & $5.0340$ & $4.4005$ & $0.6335$ \\
& $200$ & $6.1553$ & $5.8750$ & $0.2803$ & $5.5918$ & $4.9433$ & $0.6485$ \\
& $300$ & $6.4911$ & $6.1949$ & $0.2962$ & $5.9061$ & $5.2635$ & $0.6426$ \\
& $400$ & $6.7129$ & $6.4242$ & $0.2887$ & $6.1317$ & $5.4905$ & $0.6412$ \\
\hline
\multirow{3}[8]{0.5cm}{\rotatebox{90}{$\lambda_0=0.50$}}
& $100$ & $3.8513$ & $3.6362$ & $0.2151$ & $3.6449$ & $3.3084$ & $0.3365$ \\
& $200$ & $4.0576$ & $3.8497$ & $0.2079$ & $3.8682$ & $3.5411$ & $0.3271$ \\
& $300$ & $4.1810$ & $3.9722$ & $0.2088$ & $3.9885$ & $3.6696$ & $0.3189$ \\
& $400$ & $4.2681$ & $4.0559$ & $0.2122$ & $4.0818$ & $3.7593$ & $0.3225$ \\
\hline
\multirow{3}[8]{0.5cm}{\rotatebox{90}{$\lambda_0=0.90$}}
& $100$ & $3.2341$ & $3.0222$ & $0.2119$ & $3.1334$ & $2.8556$ & $0.2778$ \\
& $200$ & $3.3550$ & $3.1547$ & $0.2003$ & $3.2529$ & $3.0027$ & $0.2502$ \\
& $300$ & $3.4152$ & $3.2281$ & $0.1871$ & $3.3089$ & $3.0838$ & $0.2251$ \\
& $400$ & $3.4542$ & $3.2747$ & $0.1795$ & $3.3666$ & $3.1348$ & $0.2318$ \\
\hline
\hline
& &
\multicolumn{6}{|c|}{$\lambda_{\infty} = 0.50$} \\ \hline
\multirow{3}[8]{0.5cm}{\rotatebox{90}{$\lambda_0=0.01$}}
& $100$ & $17.2517$ & $14.5254$ & $2.7263$ & $12.5621$ & $9.3600$ & $3.2021$ \\
& $200$ & $20.1402$ & $17.2179$ & $2.9223$ & $15.2273$ & $11.8155$ & $3.4118$ \\
& $300$ & $21.8208$ & $18.8028$ & $3.0180$ & $16.8002$ & $13.3295$ & $3.4707$ \\
& $400$ & $22.9024$ & $19.9526$ & $2.9498$ & $17.9379$ & $14.4156$ & $3.5223$ \\
\hline
\multirow{3}[8]{0.5cm}{\rotatebox{90}{$\lambda_0=0.50$}}
& $100$ & $9.5787$ & $8.8826$ & $0.6961$ & $7.7808$ & $6.4275$ & $1.3533$ \\
& $200$ & $10.9615$ & $10.2309$ & $0.7306$ & $9.1228$ & $7.6872$ & $1.4356$ \\
& $300$ & $11.8190$ & $11.0263$ & $0.7927$ & $9.9153$ & $8.4447$ & $1.4706$ \\
& $400$ & $12.3540$ & $11.5900$ & $0.7640$ & $10.5127$ & $8.9938$ & $1.5189$ \\
\hline
\multirow{3}[8]{0.5cm}{\rotatebox{90}{$\lambda_0=0.90$}}
& $100$ & $4.8245$ & $4.7586$ & $0.0659$ & $4.4862$ & $4.2058$ & $0.2804$ \\
& $200$ & $5.1108$ & $5.0501$ & $0.0607$ & $4.8206$ & $4.5406$ & $0.2800$ \\
& $300$ & $5.2656$ & $5.2123$ & $0.0533$ & $4.9900$ & $4.7238$ & $0.2662$ \\
& $400$ & $5.3880$ & $5.3240$ & $0.0640$ & $5.1271$ & $4.8481$ & $0.2790$ \\
\hline
\hline
    \end{tabular}
        } 
    \end{center}
\end{table*}

\begin{figure*}[htb!]
\begin{center}
\begin{tabular}{cc}
\includegraphics[height=2.5in,width=2.8in] {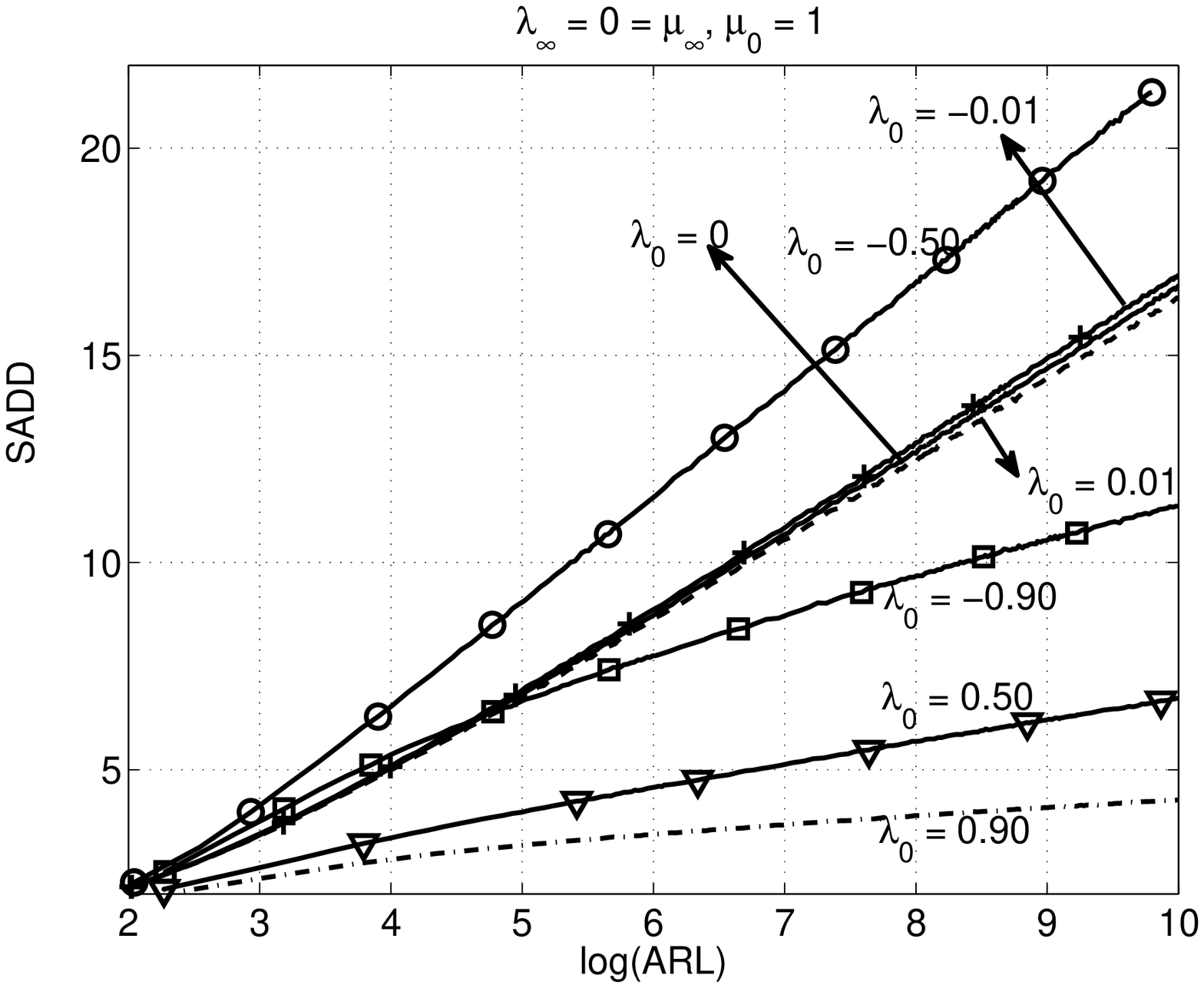}
&
\includegraphics[height=2.5in,width=2.8in] {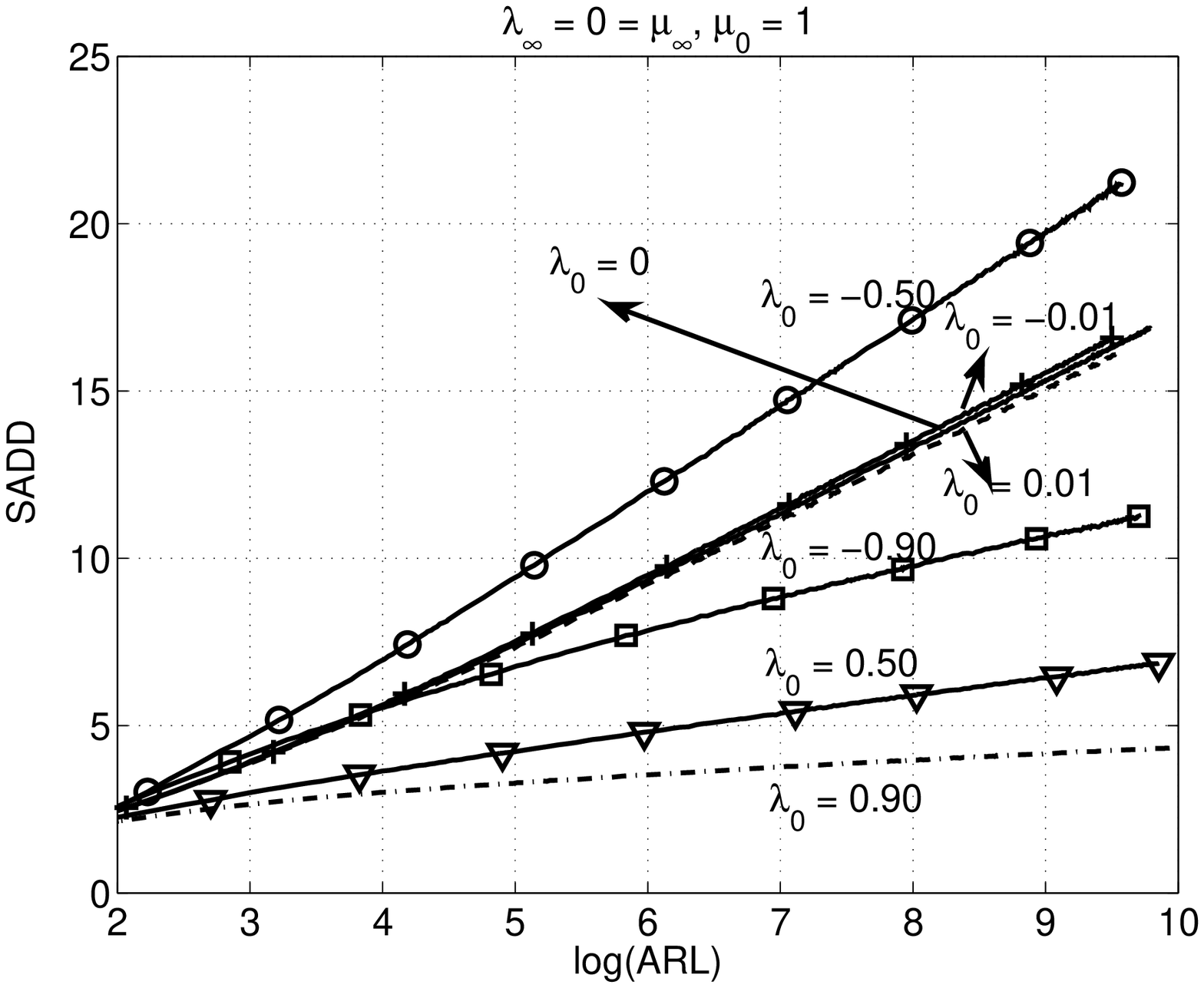}
\\ {\hspace{0.05in}} (a) & (b)
\\
\includegraphics[height=2.5in,width=2.8in] {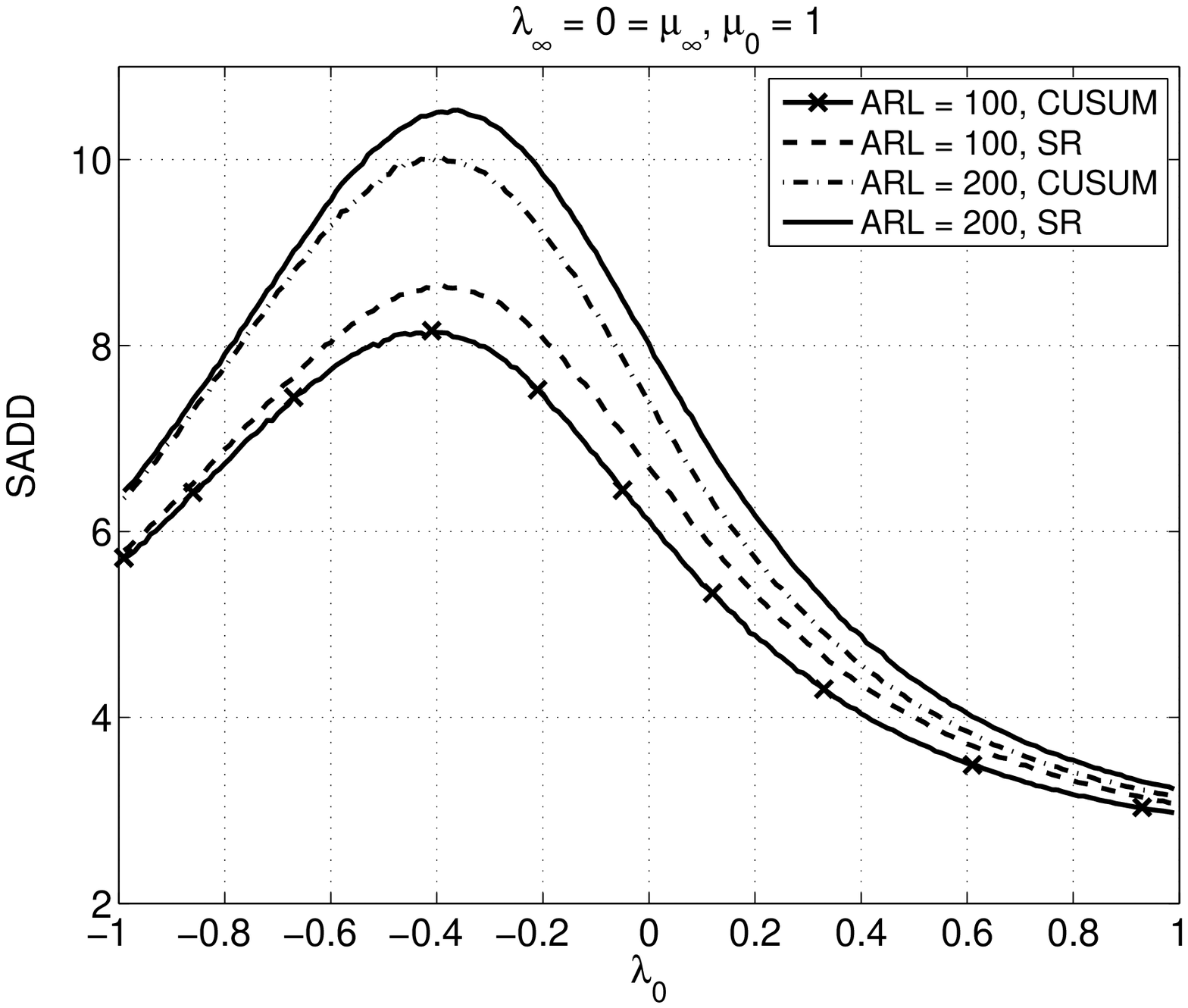}
&
\includegraphics[height=2.5in,width=2.8in] {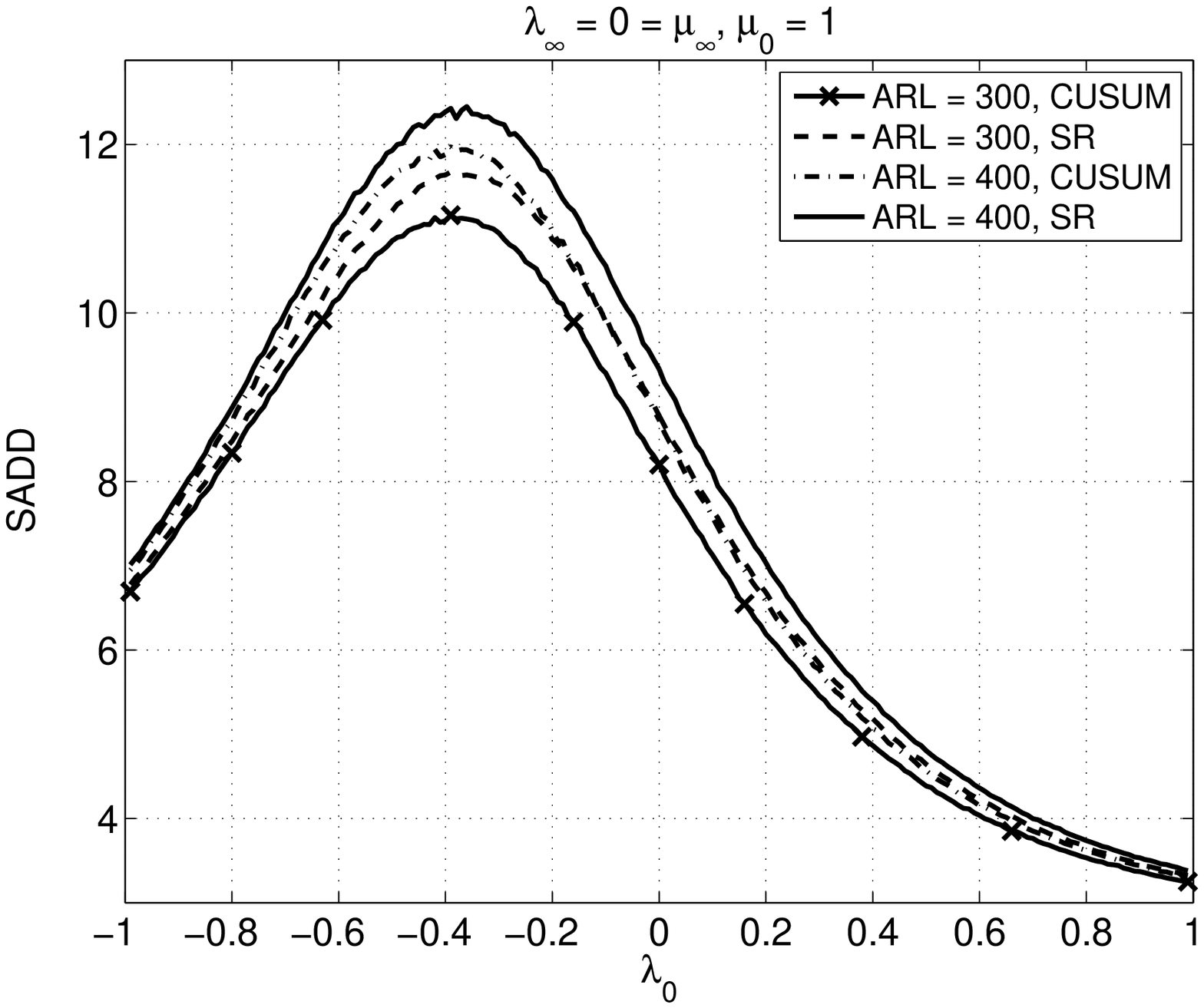}
\\
{\hspace{0.05in}} (c) & (d)
\end{tabular}
\caption{\label{fig_cusum_sr}
Performance of (a) CUSUM chart and (b) SR procedure for the i.i.d.\ pre-change
setting for different values of $\lambda_0$ with $\mu_0 = 1$. (c)-(d)
$\SADD$ as a function of $\lambda_0$ for the CUSUM chart and SR procedure
for different $\ARL$ values.}
\end{center}
\vspace{-5mm}
\end{figure*}

\begin{figure*}[htb!]
\begin{center}
\begin{tabular}{cc}
\includegraphics[height=2.5in,width=2.8in] {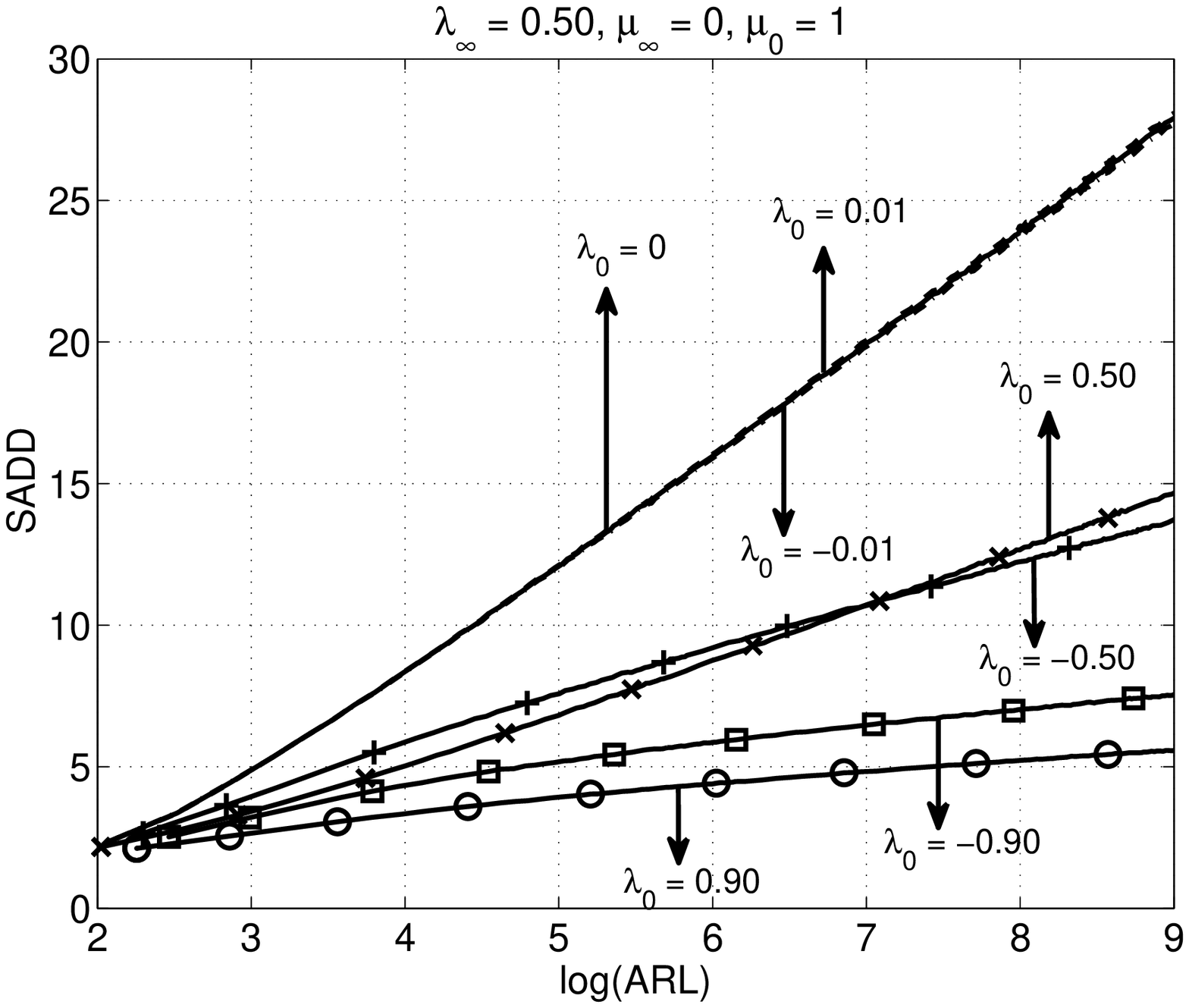}
&
\includegraphics[height=2.5in,width=2.8in] {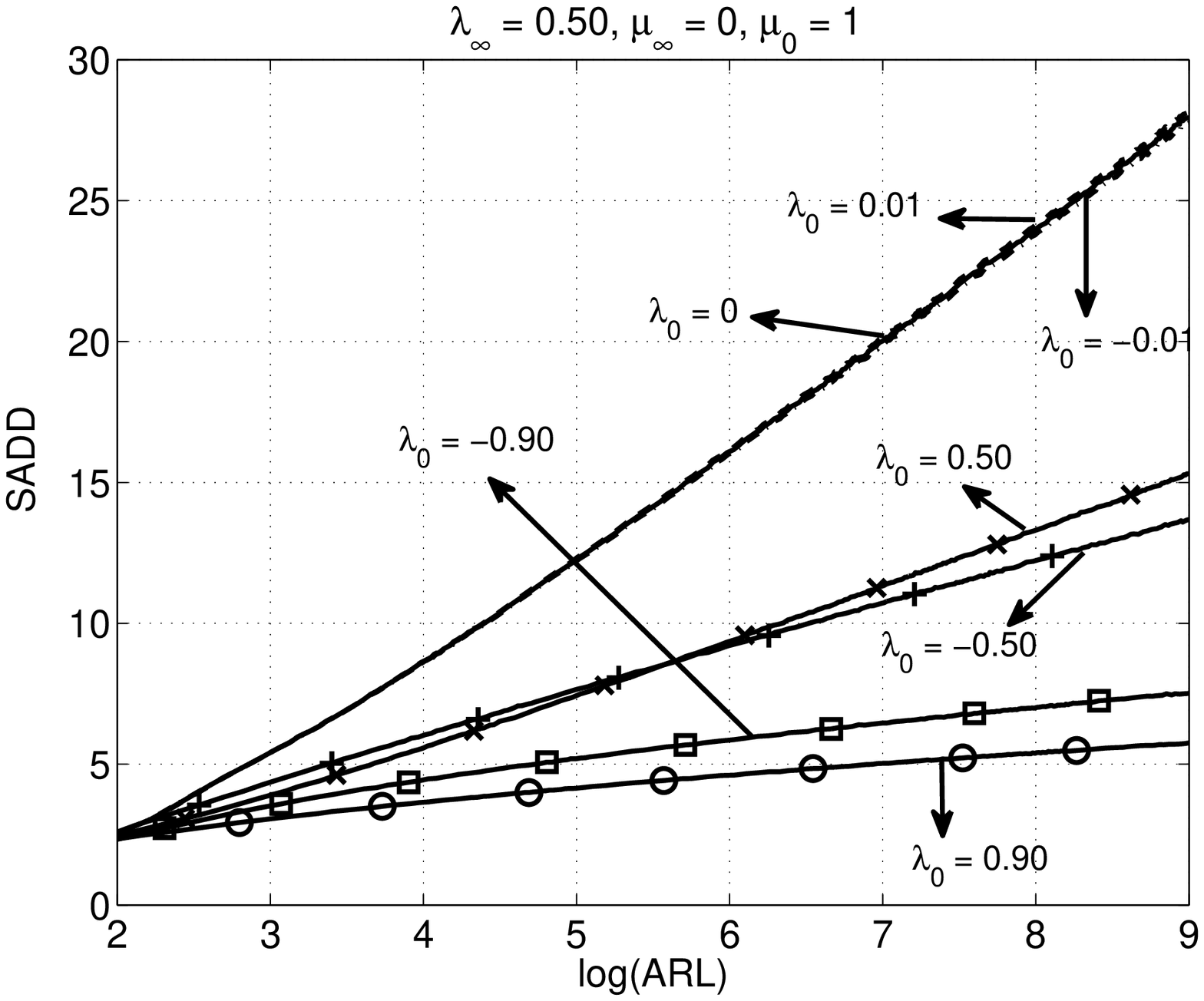}
\\ {\hspace{0.05in}} (a) & (b)
\\
\includegraphics[height=2.5in,width=2.8in] {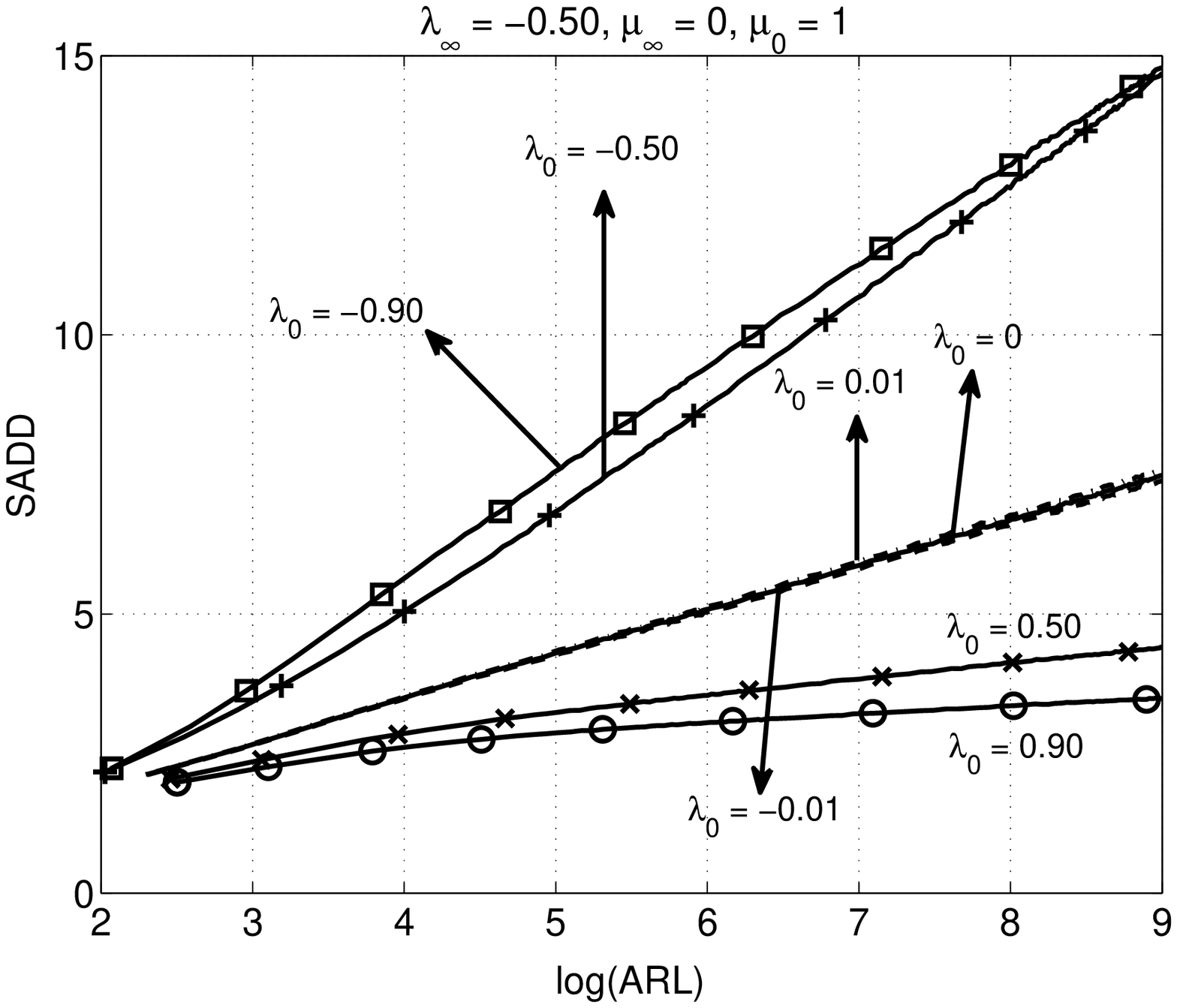}
&
\includegraphics[height=2.5in,width=2.8in] {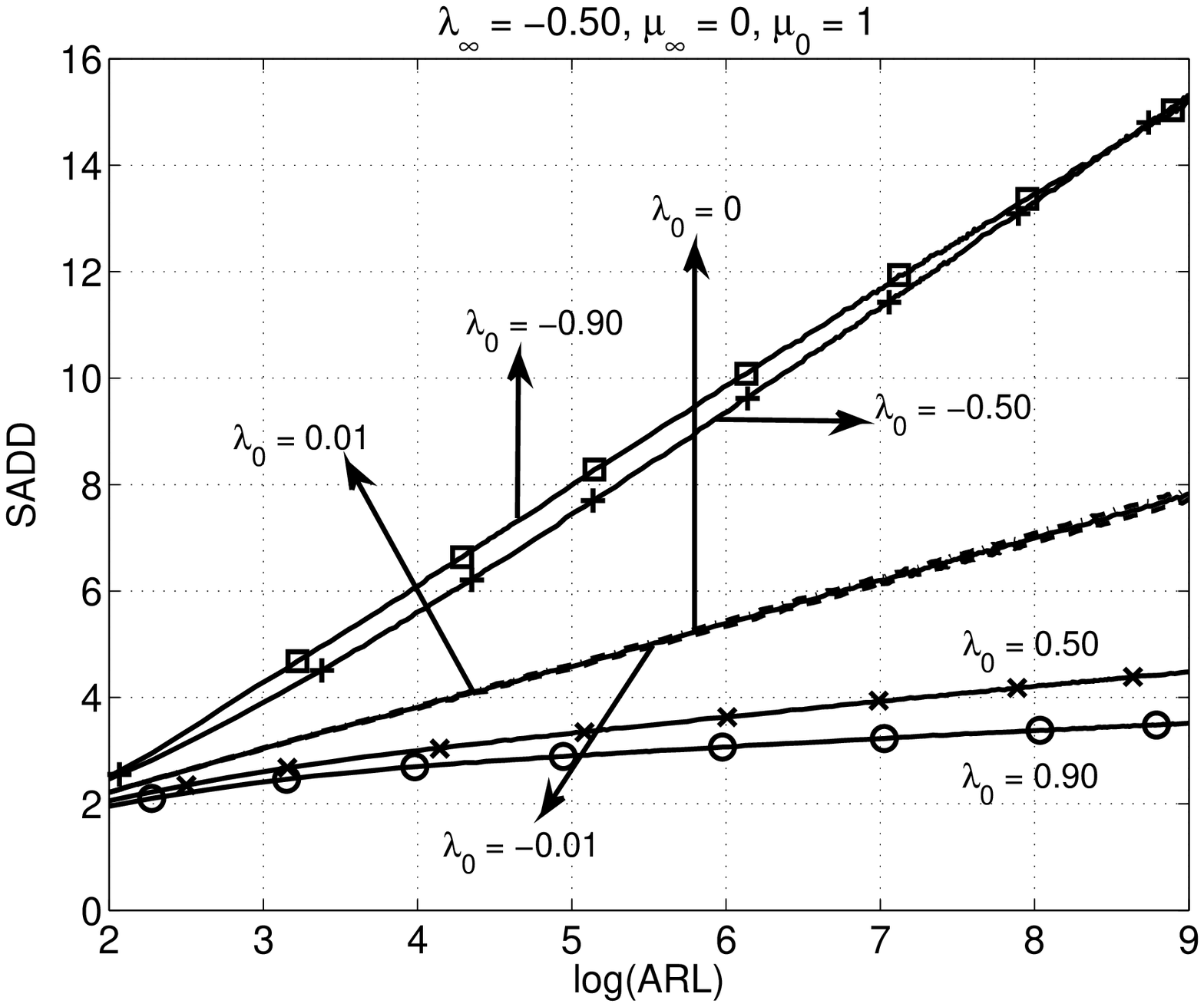}
\\
{\hspace{0.05in}} (c) & (d)
\end{tabular}
\caption{\label{fig_cusum_sr_general}
Performance of CUSUM chart and SR procedure (a)-(b) with $\lambda_{\infty}
= 0.50$ and (c)-(d) with $\lambda_{\infty} = -0.50$ for different values
of $\lambda_0$ and $\mu_0 = 1$.}
\end{center}
\vspace{-5mm}
\end{figure*}

\begin{figure*}[htb!]
\begin{center}
\begin{tabular}{cc}
\includegraphics[height=2.1in,width=2.8in] {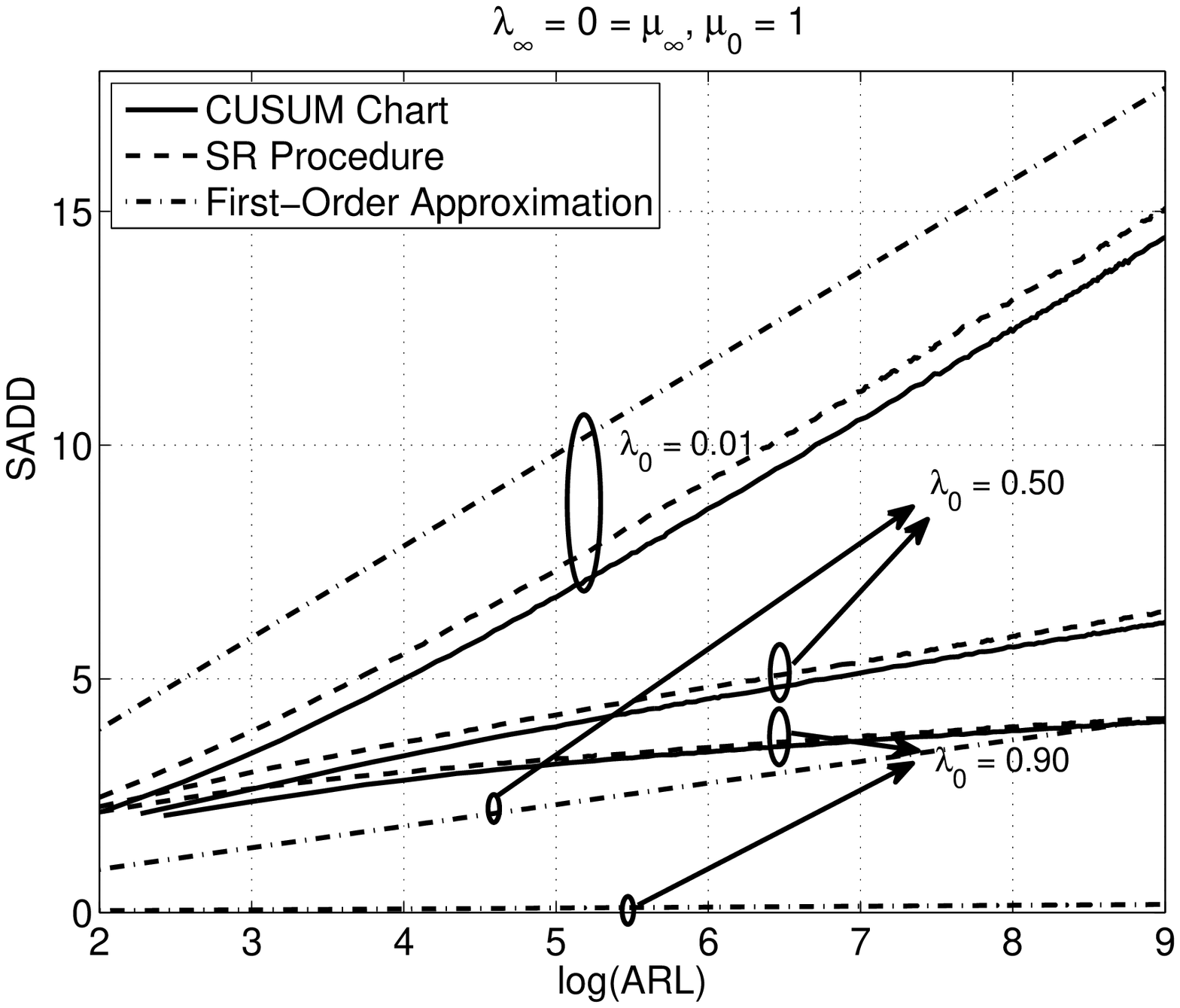}
&
\includegraphics[height=2.1in,width=2.8in] {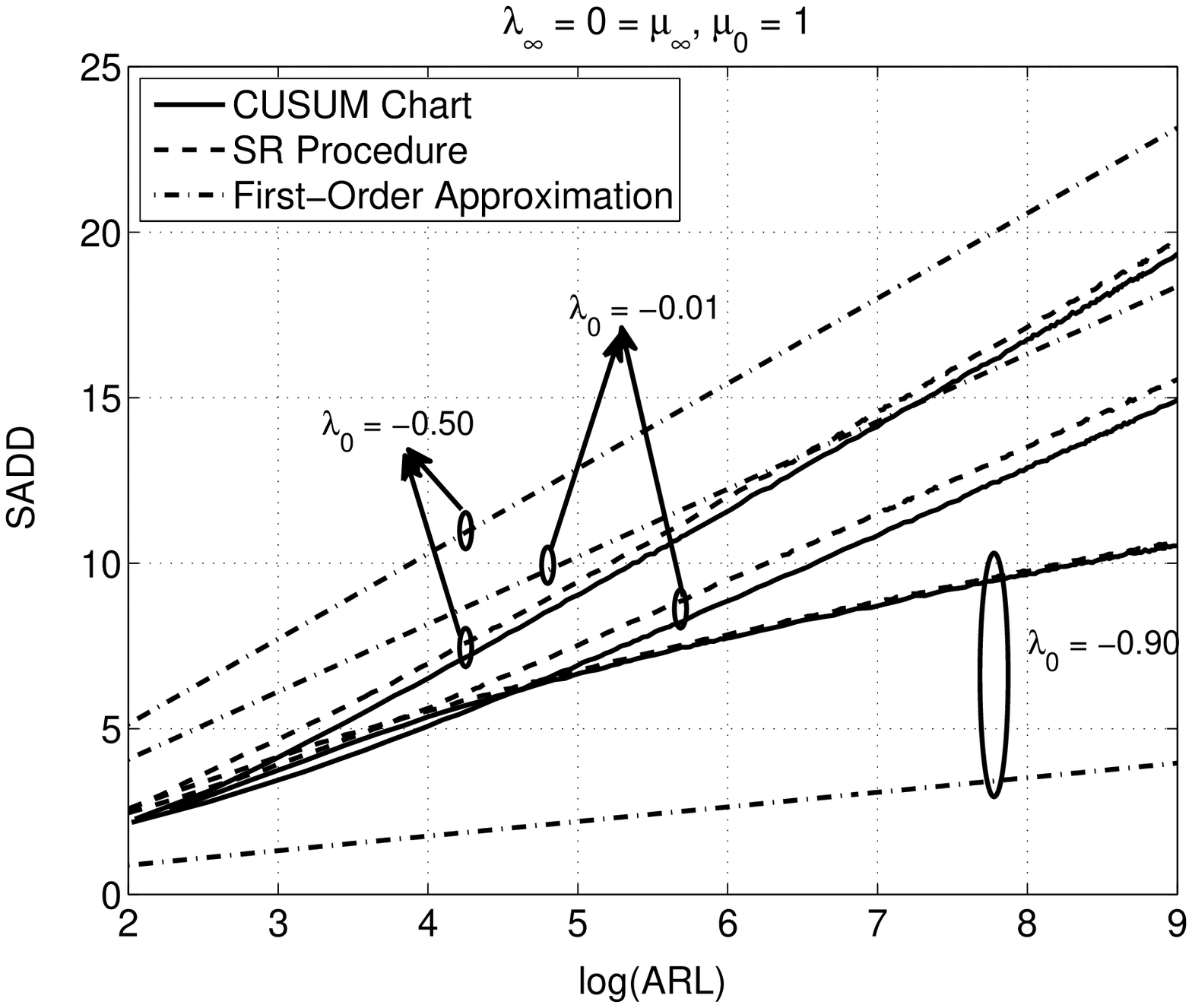}
\\ {\hspace{0.05in}} (a) & (b)
\\
\includegraphics[height=2.1in,width=2.8in] {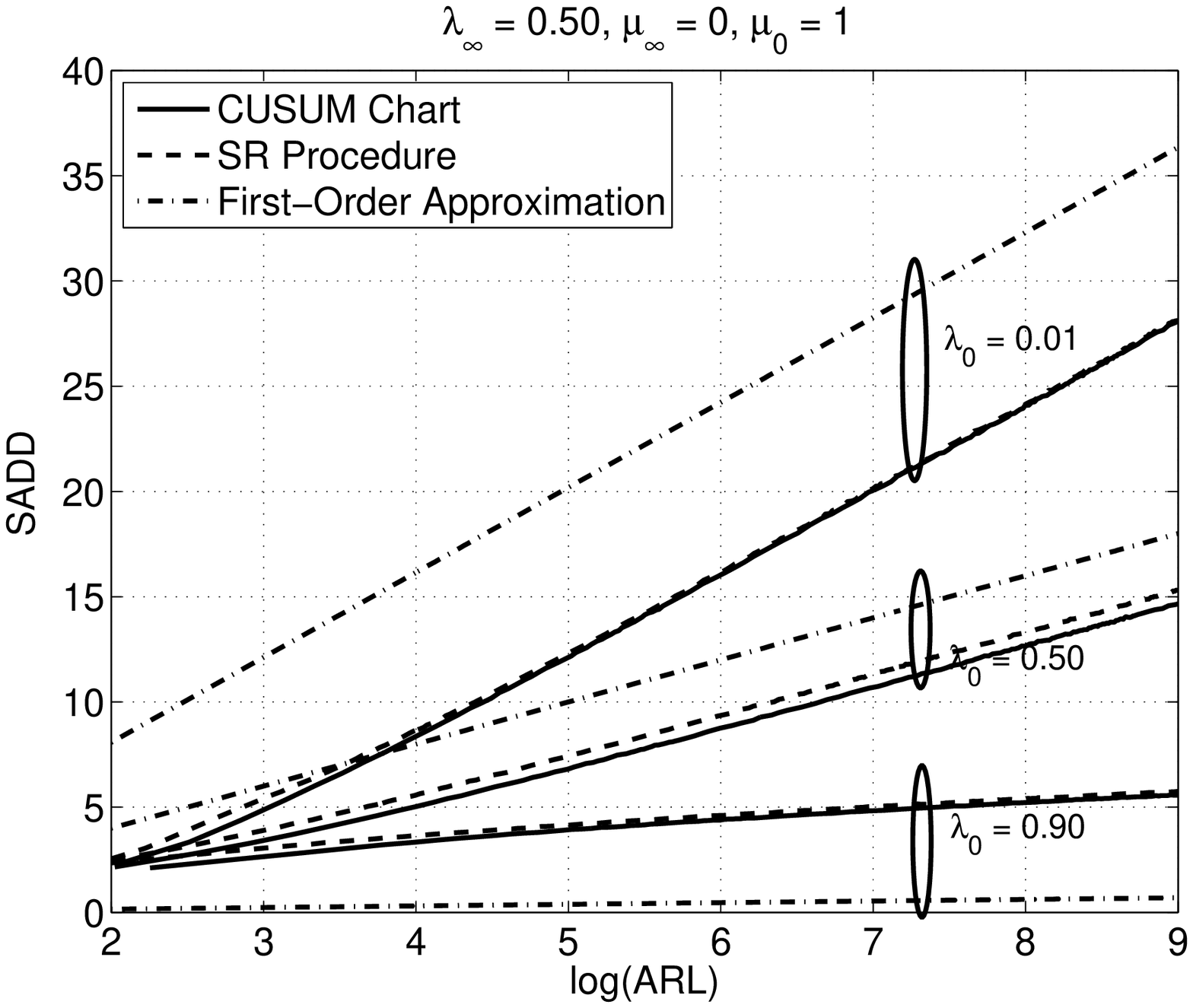}
&
\includegraphics[height=2.1in,width=2.8in] {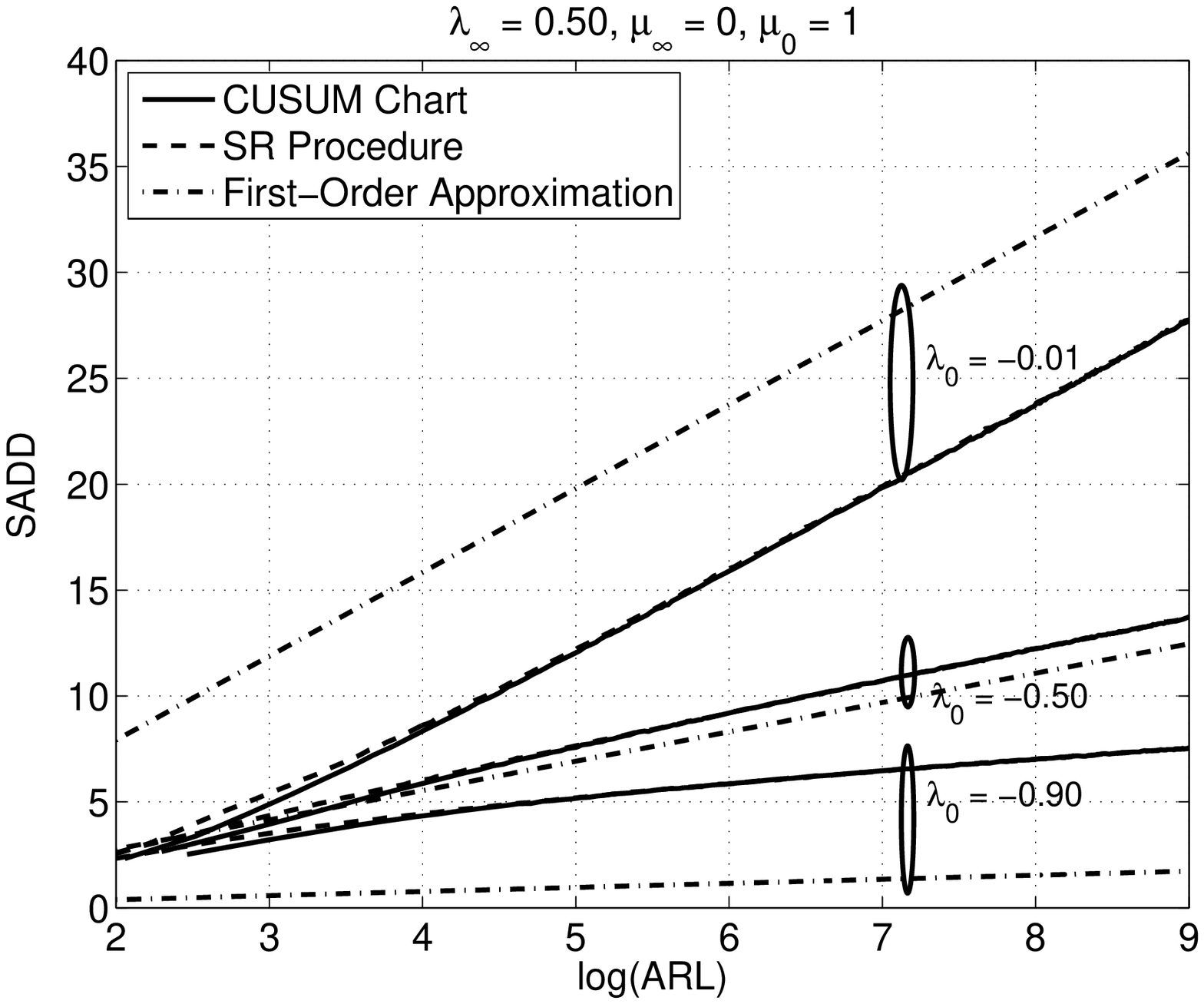}
\\ {\hspace{0.05in}} (c) & (d)
\\
\includegraphics[height=2.1in,width=2.8in] {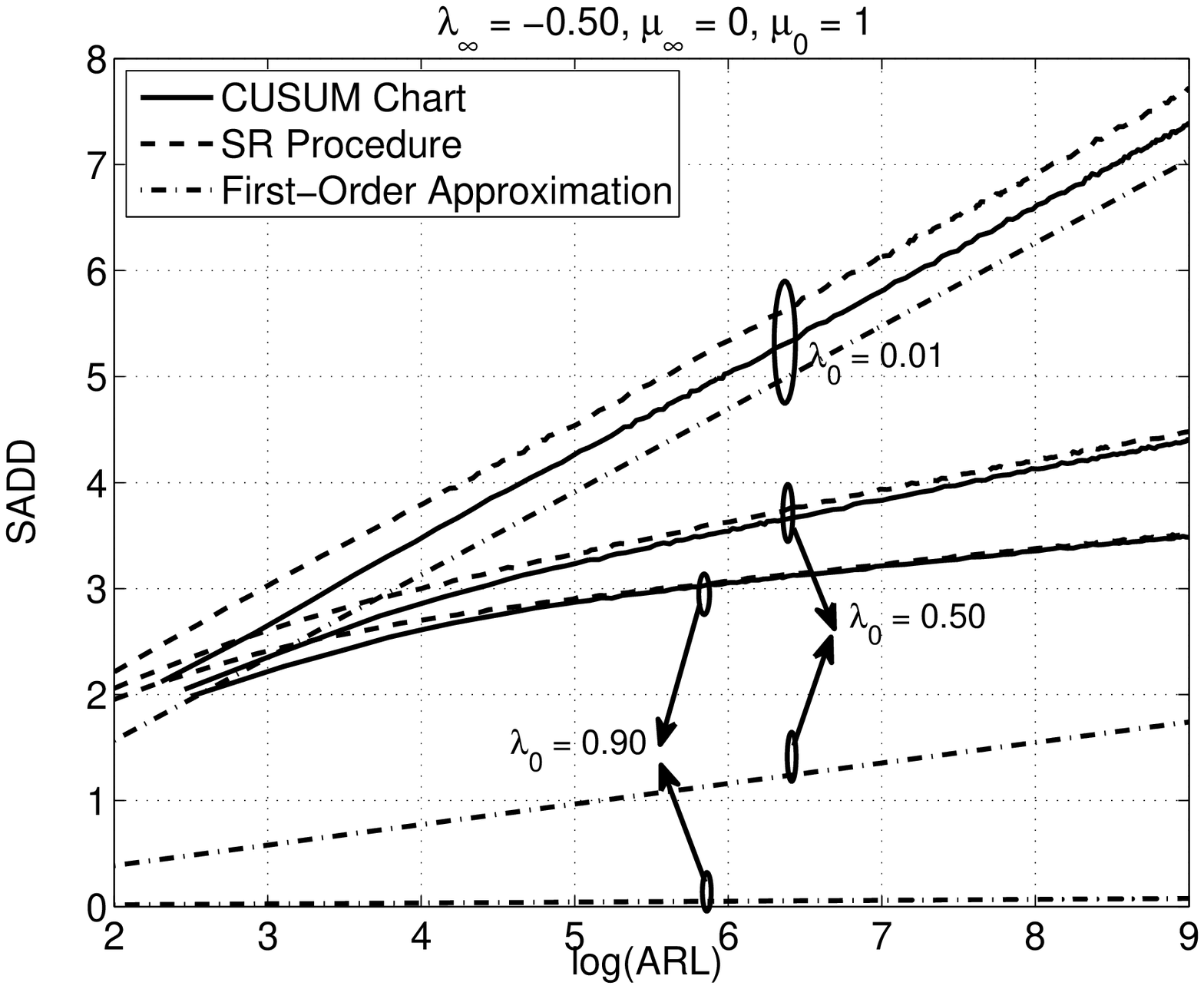}
&
\includegraphics[height=2.1in,width=2.8in] {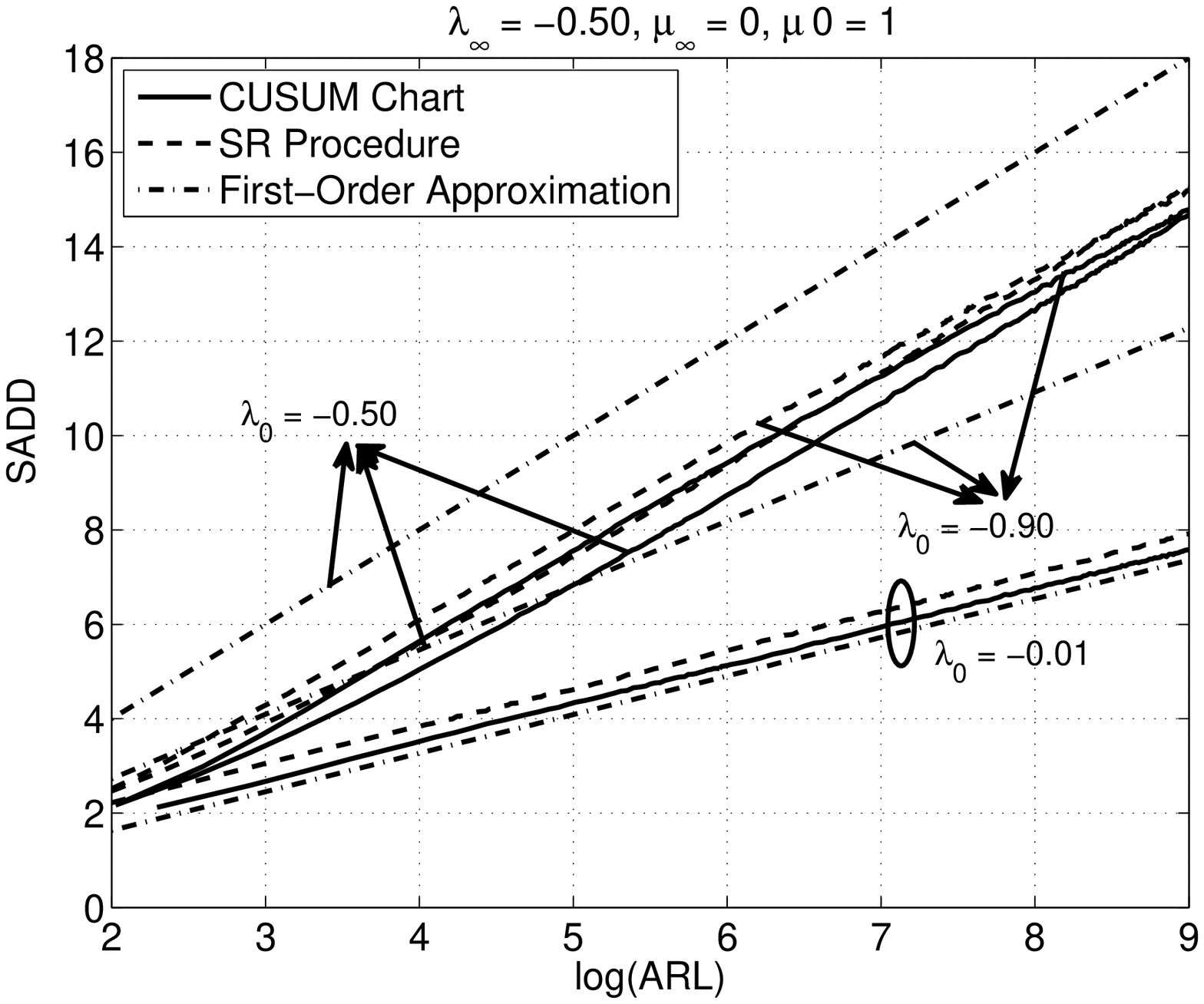}
\\
{\hspace{0.05in}} (e) & (f)
\end{tabular}
\caption{\label{fig_cusum_sr_foa}
CUSUM chart and SR procedure performance along with the first-order
approximation from~(\ref{first_order_approx}) for the i.i.d.\ pre-change
setting with $\mu_{\infty} = 0$ for (a) select positive $\lambda_0$ values
and (b) select negative $\lambda_0$ values. Similar plots are provided
for the case where (c)-(d) $\lambda_{\infty} = 0.50$ and (e)-(f) $\lambda_{\infty}
= -0.50$.
}
\end{center}
\vspace{-5mm}
\end{figure*}

\subsection{Comparison between the CUSUM chart and SR procedure}

We start with a relative performance comparison between the CUSUM chart
and the SR procedure in the i.i.d.\ pre-change setting ($\lambda_{\infty}
= 0$) with $\mu_{\infty} = 0$ and $\mu_0 = 1$ 
as a function of $\lambda_0$. For this case, $\SADD$ corresponding to
the CUSUM chart and the SR procedure are plotted as a function of
$\log({\ARL})$ in Figs.~\ref{fig_cusum_sr}(a)-(b), respectively. In
these plots, we consider six post-change settings with correlation parameter
given as $\lambda_0 = \pm 0.01, \hsppp \pm 0.50, \hsppp \pm 0.90$ in addition
to the i.i.d.\ post-change setting ($\lambda_0 = 0$).

As expected from similar studies on the i.i.d.\ problem (see,
e.g.,~\cite{Polunchenko+Tartakovsky:MCAP2012,Tartakovsky+Moustakides:SA2010,Tartakovsky+Veeravalli:TPA05}), $\SADD$ for either procedure is linear in $\log(\ARL)$ in the
AR framework as ${\ARL} \rightarrow \infty$. Further, with
either procedure, the change is more easily detectable (marked by a smaller
$\SADD$ value for the same $\ARL$ value) relative to the $\lambda_0 = 0$ case
as $\lambda_0$ increases from $0.01$ to $0.50$ and $0.90$. On the other hand,
as $\lambda_0$ decreases to $-0.01$ and $-0.50$, the change gets relatively
more difficult to detect. However, with a further decrease in $\lambda_0$ to
$-0.90$, the change becomes easier to detect.

Reinforcing the above observation, we plot $\SADD$ as a function of $\lambda_0$
for four different $\ARL$ values ($\ARL = 100$, $200$, $300$ and $400$) for the
CUSUM chart and the SR procedure in Figs.~\ref{fig_cusum_sr}(c) and~(d),
respectively. For the CUSUM chart, we observe that $\SADD$ in the correlated
case is as large as the $\SADD$ corresponding to the $\lambda_0 = 0$ case if
$\lambda_0 \in [-0.91, \hsppp 0]$ in the $\ARL = 100$ scenario. The
corresponding intervals in the $\ARL = 200, 300$ and $400$ scenarios are
$[-0.84, \hsppp 0]$, $[-0.81, \hsppp 0]$ and $[-0.79, \hsppp 0]$, respectively.
The maximum value of $\SADD$ is observed in the four scenarios at $\lambda_0 =
-0.41$, $-0.39$, $-0.39$ and $-0.39$, respectively. In the case of the SR procedure,
the corresponding intervals in $\ARL = 100$, $200$, $300$ and $400$ scenarios
are $\lambda_0 \in [-0.83, \hsppp 0]$, $\lambda_0 \in [-0.78, \hsppp 0]$,
$\lambda_0 \in [-0.77, \hsppp 0]$ and $\lambda_0 \in [-0.76, \hsppp 0]$,
respectively. The maximum value of $\SADD$ is observed at $\lambda_0 =
-0.39$, $-0.36$, $-0.39$ and $-0.36$, respectively. From a theoretical
perspective, for the AR framework considered here ($\lambda_{\infty} = 0
= \mu_{\infty}$ and $\mu_0 = 1$), it can be checked that
$\lambda_{0, \hsppp {\sf upper}} = 0$, $\lambda_{0, \hsppp {\sf lower}}
\approx -0.6180$ and $\lambda_{0, \hsppp {\sf crit}} = -\frac{1}{3}$ and
the KL number in the correlated case is smaller than the $\lambda_0 = 0$
case over the interval $\left[\lambda_{0, \hsppp {\sf lower}} , \hsppp 0
\right]$ with the minimum attained at $\lambda_0 =  \lambda_{0,
\hsppp {\sf crit}}$~(see Fig.~\ref{fig_kl}(c) for ${\KL}$ as a function
of $\lambda_0$). As $\ARL \rightarrow \infty$, the observed interval where
$\SADD$ is larger and the worst-case correlation value converge to the
theoretically expected values.

\begin{table}[p]
    \centering
    \caption{Operating characteristics of the CUSUM and SR procedures:
    $\lambda_{\infty} = 0$, $\mu_\infty = 0$, and $\mu_0 = 1$. 
    Standard errors are presented in parentheses.}
        \scalebox{0.8}{
        \begin{tabular}{|l||l||c||c|c|c|c|c|c|}
    \hline & Procedure &$\gamma$ &50 &100 &500 &1000 &5000 &10000\\
        \hline
    \hline
       \multirow{6}[8]{0.5cm}{\rotatebox{90}{$\lambda_0=0.90$}}
        &\multirow{3}[8]{1.7cm}{CUSUM}
        &$A$ & $5.6500$ & $9.8750$ & $39.5000$ & $73.9000$ & $324.4000$ & $618.8975$ \\\cline{3-9}
        &&$\ARL$ & $49.81$ & $100.31$ & $499.58$ & $999.64$ & $4999.95$ & $10000.31$ \\
        & & & $(0.04)$ & $(0.07)$ & $(0.35)$ & $(0.71)$ & $(3.54)$ & $(7.07)$ \\
        \cline{3-9}
        &&$\SADD$ & $2.7995$ & $3.0575$ & $3.4895$ & $3.6493$ & $3.9920$ & $4.1264$ \\
        & & & $(0.0015)$ & $(0.0016)$ & $(0.0017)$ & $(0.0017)$ & $(0.0018)$ & $(0.0019)$ \\
        \cline{2-9}
        &\multirow{3}[8]{1.7cm}{SR}
        &$A$ & $14.1500$ & $25.8000$ & $107.8750$ & $202.2350$ & $885.9000$ & $1685.9350$ \\\cline{3-9}
        &&$\ARL$ & $49.99$ & $99.93$ & $499.79$ & $1000.71$ & $5000.99$ & $9999.93$ \\
        & & & $(0.03)$ & $(0.07)$ & $(0.35)$ & $(0.71)$ & $(3.53)$ & $(7.07)$ \\
        \cline{3-9}
        &&$\SADD$ & $2.9775$ & $3.1811$ & $3.5841$ & $3.7438$ & $4.0737$ & $4.2039$ \\
        & & & $(0.0014)$ & $(0.0015)$ & $(0.0017)$ & $(0.0017)$ & $(0.0018)$ & $(0.0018)$ \\
        \hline
    \hline
       \multirow{6}[8]{0.5cm}{\rotatebox{90}{$\lambda_0=0.50$}}
        &\multirow{3}[8]{1.7cm}{CUSUM}
        &$A$ & $6.5750$ & $11.9000$ & $53.2500$ & $103.2500$ & $492.7500$ & $971.2000$ \\\cline{3-9}
        &&$\ARL$ & $50.02$ & $99.65$ & $500.35$ & $999.65$ & $5000.60$ & $10000.97$ \\
        & & & $(0.04)$ & $(0.07)$ & $(0.35)$ & $(0.71)$ & $(3.54)$ & $(7.07)$
        \\
        \cline{3-9}
        &&$\SADD$ & $3.2926$ & $3.7446$ & $4.6894$ & $5.0794$ & $5.9552$ & $6.3127$ \\
        & & & $(0.0020)$ & $(0.0022)$ & $(0.0026)$ & $(0.0028)$ & $(0.0032)$ & $(0.0033)$ \\
        \cline{2-9}
        &\multirow{3}[8]{1.7cm}{SR}
        &$A$ & $18.5000$ & $35.3500$ & $164.1000$ & $320.4500$ & $1532.9250$ & $3024.18$ \\\cline{3-9}
        &&$\ARL$ & $50.12$ & $99.71$ & $499.96$ & $1000.04$ & $4998.86$ & $10002.44$ \\
        & & & $(0.03)$ & $(0.07)$ & $(0.35)$ & $(0.70)$ & $(3.53)$ & $(7.07)$ \\
        \cline{3-9}
        &&$\SADD$ & $3.5868$ & $4.0039$ & $4.9385$ & $5.3144$ & $6.1772$ & $6.5323$
        \\
        & & & $(0.0019)$ & $(0.0021)$ & $(0.0026)$ & $(0.0028)$ & $(0.0032)$ & $(0.0033)$
        \\
        \hline
    \hline
        \multirow{6}[8]{0.5cm}{\rotatebox{90}{$\lambda_0=0.01$}}
        &\multirow{3}[8]{1.7cm}{CUSUM}
        &$A$ & $9.1850$ & $17.1640$ & $80.1035$ & $158.5061$ & $783.2500$ & $1563.1025$
        \\ \cline{3-9}
        &&$\ARL$ & $49.94$ & $99.99$ & $500.19$ & $1000.40$ & $4999.07$ & $9999.14$ \\ 
        &&       & $(0.05)$ & $(0.07)$ & $(0.35)$ & $(0.70)$ & $(3.53)$ & $(7.07)$
        \\ \cline{3-9}
        &&$\SADD$ & $4.8373$ & $6.0403$ & $9.0262$ & $10.3655$ & $13.4867$ &
        $14.8425$ \\ 
        &&        & $(0.0031)$ & $(0.0026)$ & $(0.0050)$ & $(0.0054)$
        & $(0.0065)$ & $(0.0069)$ \\
        \cline{2-9}
        &\multirow{3}[8]{1.7cm}{SR}
        &$A$ & $27.4112$ & $55.0144$ & $278.0016$ & $555.2155$ & $2776.7500$
        & $5553.0500$
        \\
        \cline{3-9}
        &&$\ARL$ & $49.97$ & $99.70$ & $500.75$ & $999.95$ & $4999.74$
        & $9999.47$ \\
        & & & $(0.03)$ & $(0.07)$ & $(0.35)$ & $(0.70)$ & $(3.53)$ & $(7.06)$
        \\ \cline{3-9}
        &&$\SADD$ & $5.3853$ & $6.6115$ & $9.6433$ & $10.9817$ & $14.1190$ & $15.4596$ \\
        & & & $(0.0027)$ & $(0.0033)$ & $(0.0046)$ & $(0.0051)$ & $(0.0062)$ & $(0.0066)$ \\
        \hline
    \hline
        \multirow{6}[8]{0.5cm}{\rotatebox{90}{$\lambda_0=0$}}
        &\multirow{3}[8]{1.7cm}{CUSUM}
        &$A$ & $9.2412$ & $17.2500$ & $80.5000$ & $159.1250$ & $788.5000$ &$1573.1500$ \\\cline{3-9}
        &&$\ARL$ & $49.97$ & $99.92$ & $499.99$ & $1000.07$ & $5000.90$ & $10000.96$\\
        & & & $(0.03)$ & $(0.07)$ & $(0.35)$ & $(0.70)$ & $(3.53)$ & $(7.06)$
        \\
        \cline{3-9}
        &&$\SADD$ & $4.8471$ & $6.0554$ & $9.1504$ & $10.3719$ & $13.7190$  & $15.0838$ \\
        & & & $(0.0031)$ & $(0.0037)$ & $(0.0050)$ & $(0.0055)$ & $(0.0066)$ & $(0.0070)$ \\
        \cline{2-9}
        &\multirow{3}[8]{1.7cm}{SR}
        &$A$ & $27.5500$ & $55.7500$ & $279.0000$ & $559.0000$ & $2801.0000$ & $5607.0050$ \\
        \cline{3-9}
        &&$\ARL$ & $50.00$ & $100.25$ & $499.01$ & $999.58$ & $5000.46$ & $10000.88$ \\
        & & & $(0.03)$ & $(0.07)$ & $(0.35)$ & $(0.70)$ & $(3.53)$ & $(7.05)$ \\
        \cline{3-9}
        &&$\SADD$ & $5.4281$ & $6.6911$ & $9.7689$ & $11.1363$  & $14.3394$ & $15.7182$ \\
        & & & $(0.0028)$ & $(0.0033)$ & $(0.0046)$ & $(0.0051)$ & $(0.0062)$ & $(0.0066)$ \\
    \hline
    \end{tabular}
    }
    \label{tab:sr_cusum_case1}
\end{table}

Further, the operating characteristics of the CUSUM chart
and the SR procedure corresponding to $\ARL$ values of $50$, $100$, $500$, $1000$,
$5000$ and $10000$ (namely, the corresponding thresholds and $\SADD$ values) are
presented in Table~\ref{tab:sr_cusum_case1} for the pre-change i.i.d.\ setting
with $\mu_{\infty} = 0$ and $\mu_0 = 1$ for four $\lambda_0$ values: $0, 0.01, 0.50$
and $0.90$. Also, presented in this table are the standard errors of $\ARL$ and
$\SADD$ computed according to the formula $\frac{s}{\sqrt{n}}$, where $s$ is the
sample standard deviation and $n$ is the number of samples. For the $\ARL$ calculations
presented in Table~\ref{tab:sr_cusum_case1}, $n = 2 \times 10^6$ independent runs
of the procedures were used, whereas for the $\SADD$ calculations, $n = 10^6$ runs
were used.

In the $\lambda_{\infty} = 0.50$ case, Figs.~\ref{fig_cusum_sr_general}(a)-(b)
plot the $\SADD$ performance of the CUSUM chart and the SR procedure as a
function of $\log({\ARL})$, respectively. The KL number in the seven
scenarios studied ($\lambda_0 = -0.90, -0.50, -0.01, 0$, $0.01, 0.50$ and $0.90$)
are $5.1925, 0.7222, 0.2526, 0.25, 0.2476, 0.50$ and $12.9211$, respectively.
The $\SADD$ vs.\ $\log(\ARL)$ slopes in Figs.~\ref{fig_cusum_sr_general}(a)-(b)
are in agreement with the KL number values. Specifically, while the $\SADD$
is larger in the $\lambda_0 = -0.50$ scenario relative to the $\lambda_0 = 0.50$
scenario for small $\ARL$ values, the larger KL number value leads to a
smaller $\SADD$ at larger $\ARL$ values. Similarly,
Figs.~\ref{fig_cusum_sr_general}(c)-(d) plot the performance of the two procedures
in the $\lambda_{\infty} = -0.50$ case. The KL number in the seven scenarios
are $0.7327, 0.50, 1.2229, 1.25, 1.2779, 5.1667$ and $117.6579$, respectively and
the $\SADD$ vs.\ $\log(\ARL)$ slopes are again in agreement. Specifically, the
slopes in the $\lambda_0 = -0.90$ and $-0.50$ scenarios behave similar to the
description above.

Recall that a change-point detection procedure ${\tau}$ from the class
${\Delta}(\gamma)$ 
is said to be {\em second-order optimal} if
\begin{eqnarray}
\SADD(\tau) - \inf \limits_{\tau \in { \Delta}(\gamma) } \SADD(\tau)
= {\cal O}(1), \hspp {\rm as} \hspp \gamma \rightarrow \infty.
\nonumber
\end{eqnarray}
As noted in Sec.~\ref{sec:problem-formulation+background}, a first-order
approximation of the performance of either procedure is given by
\begin{eqnarray}
\SADD = \frac{ \log(\ARL) } { {\KL} }, \hspp
{\rm as} \hspp \ARL \rightarrow \infty. 
\label{first_order_approx}
\end{eqnarray}
In Fig.~\ref{fig_cusum_sr_foa}(a), the $\SADD$ vs.\ $\log({\ARL})$
performance of the CUSUM chart and SR procedure are compared in three
cases: $\lambda_0 = 0.01, \hsppp 0.50$ and $0.90$. Also plotted is
the first-order approximation from~(\ref{first_order_approx}).
Fig.~\ref{fig_cusum_sr_foa}(b) plots the performance of the CUSUM chart,
SR procedure and first-order approximation in three cases: $\lambda_0 =
-0.01, \hsppp -0.50$ and $-0.90$. On the other hand,
Figs.~\ref{fig_cusum_sr_foa}(c)-(d) and (e)-(f) illustrate the same trends
in the $\lambda_{\infty} = 0.50$ and $\lambda_{\infty} = -0.50$ cases,
respectively. From our studies, we observe that the CUSUM chart
out-performs the SR procedure for any set of parameter values with the
gap in performance (generally) decreasing as $|\lambda_0|$ increases.
Nevertheless, both procedures have the same slope, which is the same as
the first-order approximation. Thus, the constant gap between the true
performance of the CUSUM chart and SR procedure in Fig.~\ref{fig_cusum_sr_foa}
and the first-order approximation suggests that both procedures are
second-order optimal.

\section{Conclusion}
\label{sec:conclusion}

While change-point detection for AR processes has been extensively studied in the statistical process control literature, a systematic characterization of the performance of the CUSUM chart as a
function of the model parameters has not received significant attention. Further still, a
comparative analysis of the CUSUM chart and a worthy competitor to it (the SR procedure) has
received even lesser attention. The focus of this work is on filling in some of these gaps
in the context of data generated by an AR(1) process that undergoes a change in the mean level
and the correlation coefficient at an unknown change-point.

Extending prior results on the i.i.d.\ problem, we developed recipes for
setting the threshold with either procedure to achieve a certain $\ARL$
performance. We also established that the worst-case detection delay (in
the Pollak sense) is realized when the change-point is at the start
of observation. Toward understanding the $\SADD$ vs.\ $\log(\ARL)$
performance of either procedure, we studied the KL number between AR
processes as a function of the model parameters. We established the existence
of a worst-case post-change parameter value that leads to the smallest KL
number (and hence, poorest detectability of change) and characterized
its structure as a function of other AR process model parameters.

Our numerical studies further reinforced the importance of the role played by the KL
number between the post- and pre-change processes. While our results
showed that the CUSUM chart slightly out-performs the SR procedure, both procedures are also second-order optimal with correlated data. Future
work will consider the problem of establishing the second-order optimality of either procedure for detecting a change in AR processes.

\section*{Acknowledgment}

The authors greatly appreciate the help received from Distinguished Prof. Shelemyahu Zacks of the Department of Mathematical Sciences at the State University of New York at Binghamton who not only encouraged this work, but also diligently read the first draft and provided valuable feedback that helped improve the quality of the manuscript.

The effort of A.\ S.~Polunchenko was supported, in part, by the Simons Foundation via a Collaboration Grant in Mathematics under Award \#\,304574.

\bibliographystyle{asmbi}
\bibliography{bib_qprc,integral-equations,numerical-analysis,miscellaneous,mcmc}

\end{document}